\documentclass[11pt, a4paper]{article}

\usepackage{fullpage}
\usepackage{amsmath}
\numberwithin{equation}{section}
\usepackage{amsthm}
\usepackage{amsfonts}
\usepackage{amssymb}
\usepackage{graphicx}
\usepackage{hyperref}
\usepackage{color}
\usepackage[inline]{enumitem}
\usepackage{accents}
\usepackage{multimedia}
\newlength{\dhatheight}
\newcommand{\doublehat}[1]{%
    \settoheight{\dhatheight}{\ensuremath{\hat{#1}}}%
    \addtolength{\dhatheight}{-0.35ex}%
    \hat{\vphantom{\rule{1pt}{\dhatheight}}%
    \smash{\hat{#1}}}}

\newcommand{\choice}[1]{\overset{\bullet}{#1}}
\newcommand{\alta}{\mathfrak{a}}

\newcommand{\C}{\tilde{C}}
\newcommand{\Chat}{\hat{C}}

\newcommand{\compC}{\mathbb{C}}

\newcommand{\E}{\mathrm{E}}

\newcommand{\Fhat}{\hat{F}}
\newcommand{\G}{\tilde{G}}
\newcommand{\gfn}{\mathbf{g}}
\newcommand{\gfntilde}{\tilde{\mathbf{g}}}

\newcommand{\h}{\tilde{h}}
\newcommand{\halfH}{\mathbb{H}}

\newcommand{\Khat}{\hat{K}}

\newcommand{\bigO}{\mathcal{O}}

\newcommand{\realR}{\mathbb{R}}
\newcommand{\tR}{\mathcal{R}}
\newcommand{\s}{\widetilde{s}}

\newcommand{\What}{\hat{W}}
\newcommand{\Vhat}{\hat{V}}
\newcommand{\X}{\tilde{X}}

\newcommand{\Y}{\tilde{Y}}

\newcommand{\intZ}{\mathbb{Z}}
\newcommand{\DeltaR}{\Delta_{\pmb{\mathit{\Sigma}}}}
\newcommand{\DeltaRhat}{\hat{\Delta}_{\pmb{\mathit{\Sigma}}}}
\newcommand{\DeltaRdoublehat}{\doublehat{\Delta}_{\pmb{\mathit{\Sigma}}}}

\newcommand{\Phihat}{\hat{\Phi}}

\newcommand{\Psihat}{\hat{\Psi}}
\newcommand{\SigmaR}{\pmb{\mathit{\Sigma}}}

\newcommand{\twoandhalf}{\ensuremath{2\textrm{\textonehalf}}}

\DeclareMathOperator{\Ai}{Ai}

\DeclareMathOperator{\ext}{ext}

\DeclareMathOperator{\diag}{diag}

\DeclareMathOperator{\inside}{in}

\DeclareMathOperator{\MB}{MB}
\DeclareMathOperator{\odd}{odd}

\DeclareMathOperator{\outside}{out}
\DeclareMathOperator{\Pe}{Pe}

\DeclareMathOperator{\pre}{pre}
\DeclareMathOperator{\Span}{Span}
\DeclareMathOperator{\Tr}{Tr}
\DeclareMathOperator{\tri}{multi}
\DeclareMathOperator{\HPe}{HPe}

\newtheorem{prop}{Proposition}
\newtheorem{thm}{Theorem}
\newtheorem{cor}{Corollary}
\newtheorem{RHP}{RH Problem}
\newtheorem{lemma}{Lemma}

\theoremstyle{remark}
\newtheorem{rmk}{Remark}

\title{Critical Hermitian matrix model with external source and Boussinesq hierarchy}

\author{
  Dong Wang\thanks{School of Mathematical Sciences, University of Chinese Academy of Sciences, Beijing, P.~R.~China 100049 \newline
    email: \href{mailto:wangdong@wangd-math.xyz}{\protect\nolinkurl{wangdong@wangd-math.xyz}}}
   \and   
    Shuai-Xia Xu\thanks{Institut Franco-Chinois de l’Energie Nucl\'eaire, Sun Yat-sen University, Guangzhou, P.~R.~China 510275 \newline
    e-mail: \href{xushx3@mail.sysu.edu.cn}{\protect\nolinkurl{xushx3@mail.sysu.edu.cn}}}
}

\begin{document}

\maketitle

\begin{abstract}
We consider the random Hermitian matrix model of dimension $2n$, with external source, defined by the probability density function
\begin{equation*}
  \frac{1}{Z_{2n}} \lvert \det(M) \rvert^{\alpha} e^{-2n\Tr (V(M) - AM)}, \quad V(x) = \frac{x^4}{4} - t\frac{x^2}{2},
\end{equation*}
where the external source $A$ has two eigenvalues $\pm a$ of equal multiplicity. We investigate the limiting local statistics of the eigenvalues of $M$ around $0$ in certain critical regimes as $n \to \infty$. When the parameters $t$ and $a$ lie on a critical curve along which the limiting mean eigenvalue density vanishes as $|x|^{1/3}$, the double scaling limit of the correlation kernel is constructed from functions associated with the Boussinesq equation. This new limiting kernel reduces to the classical Pearcey kernel when $\alpha = 0$. Furthermore, in the multi-critical case where the limiting mean eigenvalue density vanishes as $|x|^{5/3}$, the limiting kernel is built from the second member of the Boussinesq hierarchy. We derive the results by transforming the random matrix model into biorthogonal ensembles that are analogous to the Muttalib-Borodin ensemble, and then analyzing its asymptotic behavior via a vector Riemann-Hilbert problem.
\end{abstract}

\tableofcontents

\section{Introduction}

Consider the Hermitian matrix model with external source defined by the probability measure
\begin{equation} \label{eq:rmt_ext_source}
  \frac{1}{Z_n}e^{-n\Tr(V(M) - AM)}, 
\end{equation}
where the external source $A$ is a fixed Hermitian matrix, the potential function $V$ is a real analytic function, and $Z_n$ is a normalization constant. The external source term breaks the unitary invariance of the measure. 

Using the Harish-Chandra-Itzykson-Zuber integral to integrate out the eigenvectors, the joint probability density function for the eigenvalues can be written explicitly in determinantal form; see Br\'ezin and Hikami \cite{Brezin-Hikami96}, \cite{Brezin-Hikami98} and Zinn-Justin \cite{Zinn_Justin97}, \cite{Zinn_Justin98}. Hence, the distribution of eigenvalues of $M$ can be characterized by its correlation kernel. Our paper is devoted to the study of the double scaling limit of the correlation kernel in certain critical regimes as the size of the matrix $n\to\infty$. 

In the Gaussian case $V(x)=x^2/2$ where the external source has two eigenvalues $\pm a$ of equal multiplicities, the asymptotics of the eigenvalue correlation kernel have been established in a series of works; see \cite{Bleher-Kuijlaars04}, \cite{Aptekarev-Bleher-Kuijlaars05}, \cite{Bleher-Kuijlaars07} for the Riemann-Hilbert approach, and \cite{Brezin-Hikami98}, \cite{Tracy-Widom06} for alternative approaches. It is remarkable that the limiting mean eigenvalue density function can vanish at the origin as a cubic root, forming a $1/3$-cusp point as $a$ tends to the critical value $1$. Near the $1/3$-cusp point, a new limiting kernel arises, which is now known as the Pearcey kernel. The Pearcey kernel can be expressed by \cite[Equation (1.7)]{Bleher-Kuijlaars07}
\begin{equation} \label{eq:Pearcey_kernel}
  K^{(\Pe)}(\xi, \eta;\tau)=\frac{p(\xi)q''(\eta)-p'(\xi)q'(\eta)+p''(\xi)q(\eta)-\tau
    p(\xi)q(\eta)}{\xi-\eta},
\end{equation}
which is denoted as $K^{cusp}(\xi, \eta; \tau)$ in \cite[Equation (1.7)]{Bleher-Kuijlaars07}, with \cite[Equation (1.8)]{Bleher-Kuijlaars07}
\begin{align}\label{eq:pearcey_integral}
  p(x) = {}& p(x;\tau)=\frac{1}{2\pi}\int_{-\infty}^\infty
             e^{-\frac14s^4-\frac{\tau}{2} s^2+isx}d s && \text{and} &
                                                                       q(x) = {}& q(x;\tau)=\frac{1}{2\pi} \int_\Sigma e^{\frac14
                                                                                  s^4+\frac{\tau}{2} s^2+isx} d s.
\end{align}
Here the contour $\Sigma$ in the definition of $q$ consists of the four rays $\{ \arg s=\pi/4,3\pi/4,5\pi/4,7\pi/4 \}$, oriented as follows: the first and third rays are oriented from infinity toward the origin, while the second and fourth rays are oriented outward from the origin. We remark that the Pearcey kernel also appears in more general external source models \cite{Capitaine-Peche16}, \cite{Hachem-Hardy-Najim16}, in the non-intersecting Brownian motions model \cite{Adler-van_Moerbeke07}, \cite{Tracy-Widom06}, \cite{Adler-Orantin-van_Moerbeke10}, \cite{Adler-Cafasso-van_Moerbeke11}, the two matrix model \cite{Geudens-Zhang15}, planar partitions \cite{Okounkov07} and tiling models \cite{Borodin-Kuan10}, \cite{Duits-Kuijlaars21}, \cite{Huang-Yang-Zhang24}. Recently, Erd\H{o}s and coauthors showed that the Pearcey kernel is universal near the $1/3$-cusp point for Wigner-type matrices with independent entries in \cite{Erdos-Kruger-Schroder20} and for general correlated complex Hermitian random matrices in \cite{Erdos-Henheik-Riabov25}.

For large $n$ and a general even polynomial potential $V$, Bleher, Delvaux and Kuijlaars \cite{Bleher-Delvaux-Kuijlaars10} studied the limiting mean eigenvalue distribution of this model \eqref{eq:rmt_ext_source} where the external source $A$ has two eigenvalues $\pm a$ of equal multiplicities. They showed that this limiting empirical distribution of eigenvalues is characterized by the solution of a constrained vector equilibrium problem. The limiting density of this distribution is supported on a finite union of disjoint intervals, but the explicit computation of the limiting density is still challenging. Before the work of \cite{Bleher-Delvaux-Kuijlaars10}, McLaughlin considered the explicit limiting density for $V(x) = x^4/4$ in \cite{McLaughlin07}. In the case of a quartic potential $V(x)=x^4/4 - tx^2/2$, \cite{Bleher-Delvaux-Kuijlaars10} provided a detailed analysis of the distribution and described all possible singular behaviors in the $(t, a)$-plane where the density vanishes on the support. See also \cite{Aptekarev-Lysov-Tulyakov09} and \cite{Aptekarev-Lysov-Tulyakov11} for the quartic potential. Later, Mart\'inez-Finkelshtein and Silva \cite{Martinez_Finkelshtein-Silva21} extended this analysis to a general polynomial potential and an external source possessing two distinct eigenvalues with multiplicities $n_1$ and $n_2 = n - n_1$. They classified the possible singular behaviors of the density function. Besides the singularities already present in the classical Hermitian random matrix model ($A=0$), the density can only have a $1/3$-cusp point
 \begin{equation}\label{eq:cusp}
\rho(x) \sim  c|x-x_0|^{\frac{1}{3}},
\end{equation}
and a $5/3$-critical point 
\begin{equation}\label{eq:higher_cusp}
\rho(x) \sim  c|x-x_0|^{\frac{5}{3}}. 
\end{equation}
Notably, no higher-order critical points of the form $\rho(x) \sim  c|x-x_0|^{\frac{k}{3}}$ with $k>5$ can occur in this model. Later in this paper, we call the $1/3$-cusp point the \emph{Pearcey critical} point, and the $5/3$-critical point the \emph{multi-critical} point. The nomenclature will be explained below. 
The previous research mentioned above concentrates on the limiting empirical distribution density of eigenvalues, or equivalently, the vector equilibrium measure. Based on the equilibrium measure result, standard Riemann-Hilbert techniques can be applied to obtain universality results if the limiting density does not have a multi-critical point; see the discussion in \cite[Section 2.5]{Bleher-Delvaux-Kuijlaars10}. The scaling limits of the correlation kernel near the multi-critical point require new technical input, and the problem has been emphasized in \cite[Section 5.2]{Kuijlaars10}.

In this paper, we study the limiting correlation kernel near both the \emph{Pearcey critical} point and the \emph{multi-critical} point by considering the following $2n$-dimensional Hermitian matrix with external source and quartic potential:
\begin{equation} \label{eq:pdf_ext_source_intro}
  \frac{1}{Z_{2n}} \lvert \det(M) \rvert^{\alpha} e^{-2n\Tr(V(M) - AM)}, \quad V(x) = \frac{x^4}{4} - t\frac{x^2}{2}, \quad A = \diag(\underbrace{a, \dotsc, a}_n, \underbrace{-a, \dotsc, -a}_n),
\end{equation}
where $\alpha>-1$ and $Z_{2n}$ is a normalization constant. Comparing with the model \eqref{eq:rmt_ext_source}, the extra term $\lvert \det(M) \rvert^{\alpha}$ introduces algebraic singular behavior at the origin. 
When $\alpha=0$, the above model was considered in \cite[Section 7]{Bleher-Delvaux-Kuijlaars10} and \cite{Aptekarev-Lysov-Tulyakov11}, where they obtained the limiting empirical distribution density of eigenvalues and described all possible singular behaviors of the density in the $(t, a)$-plane. 
As shown in Figure \ref{fig:phase_diagram}, in the regions labelled as Cases I, II and III that are generic regions of the parameter space, the limiting distribution of eigenvalues is either supported on a single interval (Case III) or on two separated intervals (Cases I and II). On the solid curve separating Case II and Case III, the limiting distribution is supported on one interval with the density vanishing like $x^2$ at $0$. Along the dashed curve parameterized by $c \in (0, 3^{-1/4})$ as follows
\begin{align} \label{eq:dashed_curve}
  t = {}& 6c^2 - \frac{1}{c^2}, & a = {}& - 2c^3+\frac{1}{c},
\end{align}
the limiting density vanishes like $\lvert x \rvert^{1/3}$ at $0$. As $c\to 3^{-1/4}$, the dashed curve intersects the solid curves at $(t, a) = (3^{1/2}, 3^{-3/4})$, and the limiting density of the eigenvalues vanishes like $\lvert x \rvert^{5/3}$ at $0$. We refer to this singular behavior as the multi-critical point.

\begin{figure}[htb]
  \centering
  \includegraphics[width=0.8\linewidth]{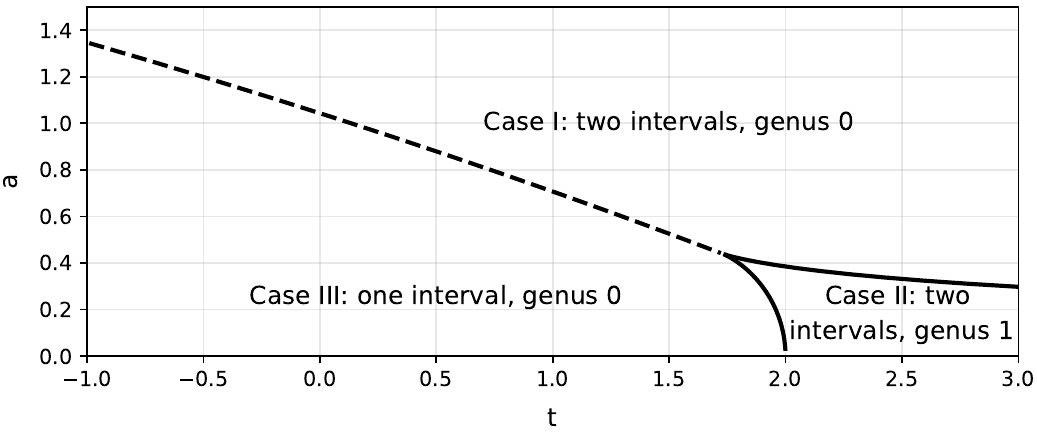}
  \caption{The phase diagram of $t$ and $a$. Adapted from \cite[Figure 7.1]{Bleher-Delvaux-Kuijlaars10}.}
  \label{fig:phase_diagram}
\end{figure}

We now consider the correlation kernel, namely $K^{\ext}_{2n}$, of the eigenvalues of $M$ in \eqref{eq:pdf_ext_source_intro}. The explicit expression of $K^{\ext}_{2n}$ will be given later in \eqref{eq:kernel_for_K^ext}. A key observation of this paper is that $K^{\ext}_{2n}$ can be expressed by the reproducing kernels of biorthogonal ensembles that generalize the Muttalib-Borodin (MB) ensemble \cite{Borodin99}, \cite{Wang-Zhang21}, \cite{Wang-Xu25}. Below we define the biorthogonal ensemble. We define the biorthogonal polynomials $\{ p_n(x) = p^{W, \gamma}_n(x) \}$ and $\{ q_n(x^2)=q^{W, \gamma}_n(x^2) \}$, where $p_n$ and $q_n$ are degree $n$ monic polynomials satisfying the biorthogonality
\begin{align} \label{eq:biorthogonal_def}
   \langle p_j(x), q_k(x^2) \rangle_{W, \gamma} = {}& \delta_{jk} h^{W, \gamma}_j, & \langle f(x), g(x) \rangle_{W, \gamma} = {}& \int_{\realR}  f(x)g(x) (1_{x \geq 0}(x) + \gamma \cdot 1_{x < 0}(x)) W(x) dx.
\end{align}
Here $\gamma$ is a real constant, and $W(x)$ is a weight function. The associated reproducing kernel is then defined by
\begin{equation}\label{eq:kernel_MB}
  K^{W, \gamma}_m(x, y) = \sum^{m - 1}_{k = 0} \frac{1}{h^{W, \gamma}_k} p^{W, \gamma}_k(x) q^{W, \gamma}_k(y^2) (1_{x \geq 0}(x) + \gamma \cdot 1_{x < 0}(x)) W(x),
\end{equation}
and we introduce the symmetrized kernel
\begin{equation} \label{eq:defn_Khat}
  \Khat^{W, \gamma}_m(x, y) = K^{W, \gamma}_m(x, y) + K^{W, \gamma}_m(-x, y).
\end{equation}
With this notation, the external source kernel admits the representation
\begin{equation}\label{eq:K_ext}
  K^{\ext}_{2n}(x, y) = \frac{1}{2} \Khat^{W, 1}_n(x, y) + \frac{y}{2x} \Khat^{W^{\odd}, -1}_n(x, y), 
\end{equation}
where
\begin{align} \label{eq:W_specified}
  W(x) = {}& \lvert x \rvert^{\alpha} e^{-2n(V(x) - ax)}, & W^{\odd}(x) = {}& \lvert x \rvert W(x).
\end{align}
The identity \eqref{eq:K_ext} will be proved later in Proposition \ref{prop:ext_in_biorth} in Section \ref{sec:bi_poly_multiple}.

We obtain scaling limits of the correlation kernels $K^{\ext}_{2n}$ for the model \eqref{eq:pdf_ext_source_intro} near both the Pearcey critical and multi-critical points by exploring the asymptotics of $\Khat^{W, \gamma}_n$ with appropriately chosen parameters and scaling regimes. 

Before giving the results on limiting local eigenvalue statistics as the size of the matrix tends to infinity, we remark that the Hermitian matrix model with external source is also studied in mathematical physics since it has applications in geometric properties of surfaces and 2D gravity \cite{Brezin-Hikami16}. The constant $Z_n$ in \eqref{eq:rmt_ext_source} is $F^{\{ V \}}_n[A]$ in \cite[Equation (1.1)]{Kharchev-Marshakov-Mironov-Morozov-Zabrodin92}, an essential object in the generalized Kontsevich model (GKM), which becomes, if $V(x)=x^3/3$, Kontsevich's Airy matrix model that is used to prove the Witten conjecture \cite{Kontsevich92}, \cite{Witten90}. We note that the integrable properties of the GKM have been extensively studied \cite{Morozov99}. For example, when the potential $V(x)=x^{p+1}/(p+1)$ or $x^{-p + 1}/(-p + 1)$, $p \geq 2$ is a positive integer, the GKM of finite size $n$ is related to $\tau$-functions of the Gelfand-Dickey hierarchy (KdV for $p = 2$ and Boussinesq for $p = 3$; see Remark \ref{rmk:Beq_def}) \cite{Kontsevich92}, \cite{Adler-van_Moerbeke92}, and \cite{Mironov-Morozov-Semenoff96}. The integrability of the Hermitian matrix model with external source is also considered in \cite{Adler-van_Moerbeke-Vanhaecke09} and \cite{Wang09}. 

\subsection{Statement of main results} \label{subsec:main_results}

To state our main results, we introduce the limiting kernels, which can be defined in terms of $\Phi^{(\gamma, \alpha, \rho, \sigma, \tau)}(\xi)$, the solution to the model Riemann-Hilbert (RH) problem \ref{RHP:model} associated with the Boussinesq hierarchy. The precise formulation and properties of this model RH problem will be given in Section \ref{subsec:model_RHP}. As shown later, $\Phi^{(\gamma, \alpha, \rho, \sigma, \tau)}(\xi)$ is a uniquely defined analytic function in the sector $\{ \arg \xi \in (\pi/4, 3\pi/4) \}$, and it admits an analytic continuation to the larger sector $\{ \arg \xi \in (-\pi, \pi) \}$. Hence, for $\xi \in \realR_+$, we denote by $\tilde{\Phi}^{(\gamma, \alpha, \rho, \sigma, \tau)}(\xi)$ the analytic continuation of $\xi^{ \frac{\alpha}{3}-\frac{1}{2} } \Phi^{(\gamma, \alpha, \rho, \sigma, \tau)}(\xi) $ from the sector $\{ \arg \xi \in (\pi/4, 3\pi/4) \}$ to $\{ \arg \xi \in (-\pi, \pi) \}$. Then, we define the kernel for $\xi, \eta \in \realR_+$
   \begin{equation} \label{eq:limit_kernel_mult}
  K^{(\gamma, \alpha, \rho, \sigma, \tau)}(\xi, \eta):=\frac{e^{\frac{2}{3}\alpha \pi i}}{2\pi i (\xi-\eta)}
\begin{pmatrix}
      1 & 0 & 0
    \end{pmatrix}
\tilde{\Phi}^{(\gamma, \alpha, \rho, \sigma, \tau)}(\eta)^{-1}\tilde{\Phi}^{(\gamma, \alpha, \rho, \sigma, \tau)}(\xi)
    \begin{pmatrix}
      0 \\ 1 \\ 0
    \end{pmatrix}.
  \end{equation}
Alternatively, the kernel can be expressed equivalently in terms of the entries of $\Phi^{(\gamma, \alpha, \rho, \sigma, \tau)}$ as
\begin{equation} \label{eq:limit_K^mult_entries} 
       K^{(\gamma, \alpha, \rho, \sigma, \tau)}(\xi, \eta) := \frac{1}{2\pi i( \xi-\eta)} \sum^{2}_{j = 0} \phi^{(\gamma, \alpha, \rho, \sigma, \tau)}_j(\xi) \psi^{(\gamma, \alpha, \rho, \sigma, \tau)}_j(\eta),
    \end{equation}
    where  
    \begin{align}
      \phi^{(\gamma, \alpha, \rho, \sigma, \tau)}_j(\xi) = {}& \xi^{\frac{\alpha}{3} - \frac{1}{2}} (e^{-\frac{\alpha}{3} \pi i} \Phi^{(\gamma, \alpha, \tau, \sigma,s)}_{j0,+}(\xi) + e^{\frac{\alpha}{3} \pi i} \Phi^{(\gamma, \alpha, \rho, \sigma, \tau)}_{j1,+}(\xi)), \\
      \psi^{(\gamma, \alpha, \rho, \sigma, \tau)}_j(\eta) = {}& \eta^{-\frac{\alpha}{3} + \frac{1}{2}} (e^{\frac{\alpha}{3} \pi i} (\Phi^{(\gamma, \alpha, \rho, \sigma, \tau)})^{-1}_{0j,+}(\eta) - e^{-\frac{\alpha}{3} \pi i} (\Phi^{(\gamma, \alpha, \rho, \sigma, \tau)})^{-1}_{1j,+}(\eta)),
    \end{align}
for $\xi, \eta\in(0,\infty)$ and $f_{+}(\xi)$ denotes the limiting value of $f(z)$ as $z\to \xi$ with $z\in \compC_+$.  

We obtain the double scaling limits of the correlation kernels $K^{\ext}_{2n}$ of the model \eqref{eq:pdf_ext_source_intro} for the critical parameters $(t, a)$ along the dashed curve defined by \eqref{eq:dashed_curve} where the Pearcey critical point arises and at the critical parameters $(t, a) = (3^{1/2}, 3^{-3/4})$ where the multi-critical point arises. When the parameters $(t, a)$ tend to these critical parameters at an appropriate speed as $n\to\infty$, we obtain a new family of limiting kernels $K^{(\gamma, \alpha, \rho, \sigma, \tau)}$ associated with the Boussinesq hierarchy.

\begin{thm} \label{thm:main}
  Let $K^{\ext}_{2n}$ be the correlation kernel of the eigenvalues of the random matrix model \eqref{eq:pdf_ext_source_intro} 
  defined equivalently by \eqref{eq:K_ext} and \eqref{eq:kernel_for_K^ext}, and $\Khat^{W, \gamma}_n$ be the symmetrized kernel function for the generalized Muttalib-Borodin ensemble whose precise definition is given in \eqref{eq:defn_Khat}. Here we take $W(x) = \lvert x \rvert^{\alpha} e^{-2n(V(x) - ax)}$ as in \eqref{eq:W_specified}, $\alpha, a, V$ the same as in \eqref{eq:pdf_ext_source_intro} and $\gamma \in [-1, 1]$.
  \begin{enumerate}
  \item (Limiting kernels near the Pearcey critical point)
    Let $t$ and $a$ be parametrized by $c \in (0, 3^{-1/4})$ and $\tau \in \realR$ as
    \begin{align} \label{eq:def_t_c_Pearcey}
      t = {}& 6c^2 - \frac{1}{c^2} - \frac{2}{3c^2} \sqrt{\frac{1 - 3c^4}{2n}} \tau, & a = {}&- 2c^3+ \frac{1}{c}  + \frac{4}{3c} \sqrt{\frac{1 - 3c^4}{2n}} \tau,
    \end{align}
    and define the constant 
    \begin{equation} \label{eq:c_1_Pearcey}
      c_1 = 2c^{-\frac{4}{3}} (1 - 3c^4).
    \end{equation}
    It then holds uniformly for $\xi, \eta$ in compact subsets of $\mathbb{R}\setminus\{0\}$ that
    \begin{equation} \label{eq:limit_K^mult}
      \lim_{n \to \infty} \frac{1}{(c_1 n)^{3/4}} \Khat^{W, \gamma}_n(\frac{\xi}{(c_1 n)^{3/4}}, \frac{\eta}{(c_1n)^{3/4}}) = 2\xi K^{(\gamma, \alpha, 0, \frac{3}{4},\tau)}(\xi^2, \eta^2), 
    \end{equation}
    and 
     \begin{multline} \label{eq:limit_K^ext}
       \lim_{n \to \infty} \frac{1}{(c_1 n)^{3/4}}   K^{\ext}_{2n}(\frac{\xi}{(c_1 n)^{3/4}}, \frac{\eta}{(c_1 n)^{3/4}}) = K^{(\Pe, \alpha)}(\xi, \eta; 2\tau), \\
       K^{(\Pe, \alpha)}(\xi, \eta; 2\tau) = \xi K^{(1, \alpha, 0, \frac{3}{4},\tau)}(\xi^2, \eta^2)+  \eta K^{(-1, \alpha+1, 0, \frac{3}{4},\tau)}(\xi^2, \eta^2).
        \end{multline}
  \item (Limiting kernels near the multi-critical point)
    Let $t$ and $a$ be parametrized by $c$ and $\tau \in \realR$ as
    \begin{align} \label{eq:def_t_c_critical}
      t = {}& 6c^2 - \frac{1}{c^2} - \frac{n^{-\frac{3}{4}}}{3c^2} \tau, & a = {}& - 2c^3 +\frac{1}{c} + \frac{2n^{-\frac{3}{4}}}{3c} \tau,
    \end{align}
    where $c$ tends to $3^{-1/4}$ at the speed
    \begin{equation} \label{eq:c_multi_crit}
      c = 3^{-\frac{1}{4}} \left( 1 - \frac{\sigma}{\sqrt{n}} \right)^{\frac{1}{4}}, \quad \sigma\in\mathbb{R}.
    \end{equation}
    It then holds uniformly for $\xi, \eta$ in compact subsets of $\mathbb{R}\setminus\{0\}$ that
    \begin{equation} \label{eq:limit_K^mult_crit}
      \lim_{n \to \infty} \frac{e^{\frac{2}{3} n^{1/4}(\eta^2 - \xi^2)} }{ (3^{\frac{2}{3}} n)^{3/8}}    \Khat^{W, \gamma}_{n}(\frac{\xi}{(3^{\frac{2}{3}} n)^{3/8}}, \frac{\eta}{(3^{\frac{2}{3}} n)^{3/8}})= 2\xi K^{(\gamma, \alpha, -\frac{1}{4}, \frac{3}{2}\sigma, \tau)}(\xi^2, \eta^2),
    \end{equation}
    and
     \begin{multline} \label{eq:limit_K^ext_crit}
       \lim_{n \to \infty} \frac{ e^{\frac{2}{3} n^{1/4}(\eta^2 - \xi^2)} }{(3^{\frac{2}{3}}  n)^{3/8}}   K^{\ext}_{2n}(\frac{\xi}{(3^{\frac{2}{3}}  n)^{3/8}}, \frac{\eta}{(3^{\frac{2}{3}}  n)^{3/8}})= K^{(\tri, \alpha)}(\xi, \eta; \sigma, \tau), \\
       K^{(\tri, \alpha)}(\xi, \eta; \sigma, \tau) = \xi K^{(1, \alpha,  -\frac{1}{4}, \frac{3}{2}\sigma, \tau)}(\xi^2, \eta^2)+  \eta K^{(-1, \alpha+1,  -\frac{1}{4}, \frac{3}{2}\sigma, \tau)}(\xi^2, \eta^2). 
        \end{multline}
  \end{enumerate}
\end{thm}

We remark that the convergence in \eqref{eq:limit_K^mult}, \eqref{eq:limit_K^ext}, \eqref{eq:limit_K^mult_crit} and \eqref{eq:limit_K^ext_crit} excludes the point $0$, because the kernels may have a mild singularity if $\xi$ or $\eta = 0$. This does not affect the convergence as correlation kernels; see \cite[Theorem 1.2 and Remark 1.4]{Claeys-Kuijlaars-Vanlessen08} for an analogous situation. The limiting kernels on the right-hand side of \eqref{eq:limit_K^ext} and in \eqref{eq:limit_K^ext_crit} have alternative expressions in Proposition \ref{pro:sum_kernel_Psi}. 

For $(t, a)$ along the dashed curve defined by \eqref{eq:dashed_curve} as shown in Figure \ref{fig:phase_diagram}, the Pearcey critical point arises. We obtain new limiting kernels in \eqref{eq:limit_K^mult} and \eqref{eq:limit_K^ext}, which are constructed out of functions associated with the Boussinesq equation as shown in Section \ref{subsec:model_RHP}. When the parameters $(\gamma, \alpha, \rho, \sigma, \tau)=(1, 0, 0, 3/4, \tau)$, the solution of the Boussinesq equation applied here is elementary and the associated model RH problem can be solved in terms of Pearcey integrals like \eqref{eq:pearcey_integral}; see Remarks \ref{rmk:poly_solution}, \ref{rem:Pearcey_red} and \ref{rem:Pearcey_red_2} below. Therefore, when $\alpha=0$, the limit \eqref{eq:limit_K^ext} reduces to 
\begin{equation} \label{eq:limit_K^ext_Pearcey}
  K^{(\Pe, 0)}(\xi, \eta; 2\tau) = K^{(\Pe)}(\eta,\xi;2\tau),
\end{equation}
as seen from \eqref{eq:Peaecey_red}, where $K^{(\Pe)}(\eta,\xi; \tau)$ is the classical Pearcey kernel defined in \eqref{eq:Pearcey_kernel}. Similarly, the limiting kernel \eqref{eq:limit_K^mult} reduces to a folded Pearcey kernel defined on the positive real line as given in \eqref{eq:Hard_Pearcey_kernel_1}. Our result confirms the universality of the Pearcey limit near the cusp singularity in the matrix model \eqref{eq:pdf_ext_source_intro} with $\alpha=0$, as expected in \cite[The Pearcey Transition paragraph of Section 7]{Bleher-Delvaux-Kuijlaars10} and \cite[Section 5.2]{Kuijlaars10}.

Notably, the result for $\alpha > -1$ at the Pearcey critical point is not a trivial extension of the Pearcey limit for $\alpha = 0$. One would expect that the Pearcey kernel admits an $\alpha$-deformation, in analogy with the deformation of the Airy kernel into the Painlev\'{e} XXXIV kernel considered in \cite{Its-Kuijlaars-Ostensson08}, \cite{Its-Kuijlaars-Ostensson09}. However, such an $\alpha$-deformation of the Pearcey kernel has not previously appeared in the literature. One contribution of our paper is that we find a natural $\alpha$-deformation of the Pearcey kernel given by $K^{(\Pe, \alpha)}(\xi, \eta; 2\tau)$. Furthermore, we show that $K^{(\Pe, \alpha)}(\xi, \eta; 2\tau)$ possesses an integrable structure associated with the Boussinesq equation, as given in Section \ref{subsec:model_RHP}. Since the Pearcey kernel occurs in various situations, we expect the $\alpha$-deformed kernel $K^{(\Pe, \alpha)}(\xi, \eta; 2\tau)$ will also arise as a universal limit in other models. 

At the multi-critical point where the parameter $(t, a) = (3^{1/2}, 3^{-3/4})$ is the intersection of the dashed curve and the two solid curves in Figure \ref{fig:phase_diagram}, the limiting density of the eigenvalues vanishes like $\lvert x \rvert^{5/3}$. To the best of the authors' knowledge, no limiting correlation kernel in the random matrix literature has been found for this kind of critical point, despite the problem having been raised many years ago in \cite[Section 5.2]{Kuijlaars10}. Another contribution of our paper is that we obtain the new limiting kernels $K^{(\gamma, \alpha, -\frac{1}{4}, \frac{3}{2}\sigma, \tau)}(\xi^2, \eta^2)$ and $K^{(\tri, \alpha)}(\xi, \eta; \sigma, \tau)$ in \eqref{eq:limit_K^mult_crit} and \eqref{eq:limit_K^ext_crit} near the multi-critical point. Furthermore, we show that they possess an integrable structure associated with the second member of the Boussinesq hierarchy, as shown in Section \ref{subsec:model_RHP}.  

As mentioned before, we focus on the critical behavior near the \emph{Pearcey critical} point and the \emph{multi-critical} point of the matrix model \eqref{eq:pdf_ext_source_intro}. However, we mention that for $(t, a)$ in regions labelled as Cases I, II and III in Figure \ref{fig:phase_diagram}, the limiting density is regular and the local limits of $K^{\ext}_{2n}$ are expected to be the universal sine kernel in the bulk and the universal Airy kernel at the edges, as claimed in \cite[End of Section 2.5 and Section 5.7]{Bleher-Delvaux-Kuijlaars10}. On the solid curve between the regions labelled Case II and Case III in Figure \ref{fig:phase_diagram}, the limiting density of eigenvalues vanishes like $x^2$ at $0$ and the local limit of $K^{\ext}_{2n}$ for $\alpha=0$ is expected to be the universal Painlev\'e II kernel that occurs in the unitary matrix model \cite{Bleher-Its99}, \cite{Claeys-Kuijlaars06}, as claimed in \cite[The Painlev\'e Transition paragraph of Section 7]{Bleher-Delvaux-Kuijlaars10} and \cite[Section 5.2]{Kuijlaars10}. The transitional behavior of the model as $(t, a)$ on the solid curve between the regions labelled Case I and Case II is more subtle; see \cite[The Painlev\'e Transition paragraph of Section 7]{Bleher-Delvaux-Kuijlaars10}. We expect that our approach also works in other regimes, and especially proves that on the solid curve between Case II and Case III, the limiting kernel is the Painlev\'{e} II kernel with $\alpha$ parameter like in \cite{Claeys-Kuijlaars-Vanlessen08}.

In this paper, we only consider the quartic potential function $V$. Our approach is based on a vector Riemann-Hilbert problem and it applies for more general even, real analytic $V$. We expect our results to be universal, in the sense that for a more general even, real analytic potential $V$, if the limiting density of the eigenvalues vanishes like $\lvert x \rvert^{1/3}$ or $\lvert x \rvert^{5/3}$ at $0$, the limiting correlation kernel at $0$ should be given by $K^{(\Pe, \alpha)}$ or $K^{(\tri, \alpha)}$. Away from the origin, such vanishing exponents as $1/3$ or $5/3$ are not expected to occur in the limiting eigenvalue density.

According to \cite{Martinez_Finkelshtein-Silva21}, for the Hermitian random matrix model with external source that has only two distinct eigenvalues, besides the singularities that occur in random matrix models without external source, the only possible new singularities are featured with vanishing orders $1/3$ and $5/3$. Hence, we conjecture that our limiting kernels complete the classification of the singular behaviour of the Hermitian random matrix model with external source that has only two distinct eigenvalues. We remark that the method in our paper cannot be extended straightforwardly to the case when $V$ is not even or the numbers of the external eigenvalues of distinct values are unequal, so the proof of this conjecture needs alternative methods.

\subsection{Model Riemann-Hilbert problem} \label{subsec:model_RHP}

We consider the $3\times 3$ matrix-valued function $\Phi^{(\gamma, \alpha, \rho, \sigma, \tau)}(\xi)$ ($\Phi(\xi)$ for short) that satisfies the following model RH problem:
\begin{RHP} \label{RHP:model}
  $\Phi^{(\gamma, \alpha, \rho, \sigma, \tau)}(\xi)$ ($\Phi(\xi)$ for short) is a $3 \times 3$ vector-valued function that is analytic on $\compC \setminus \Gamma_{\Phi}$ and continuous up to the boundary, where $\Gamma_{\Phi} = \realR_+ \cup \realR_- \cup \{ \arg \xi = \pm \pi/4 \} \cup \{ \arg \xi = \pm 3\pi/4 \}$, with all six rays oriented outward from $0$. It satisfies the following properties.
  \begin{enumerate}
   \item \label{enu:RHP:model_1}
    $\Phi^{(\gamma, \alpha, \rho, \sigma, \tau)}_+(\xi) = \Phi^{(\gamma, \alpha, \rho, \sigma, \tau)}_-(\xi) J_{\Phi}(\xi)$, where 
    \begin{equation} \label{eq:RHP:model_1}
      J_{\Phi}(\xi) = \left\{
        \begin{aligned}
          & \begin{pmatrix}
            1 & \pm e^{\mp \frac{2\alpha}{3} \pi i} & 0 \\
            0 & 1 & 0 \\
            0 & 0 & 1
          \end{pmatrix}, &
                           \arg \xi = {}& \pm \frac{\pi}{4}, \\
          & \begin{pmatrix}
            1 & 0 & 0 \\
            0 & 1 & 0 \\
            0 & \mp \gamma e^{\mp \frac{\alpha}{3} \pi i} & 1
          \end{pmatrix}, &
                           \arg \xi = {}& \pm \frac{3\pi}{4},  \\
          & \begin{pmatrix}
            0 & 1 & 0 \\
            1 & 0 & 0 \\
            0 & 0 & 1
          \end{pmatrix}, &
                           \xi \in {}& \realR_+, \\
                           &
                                  \begin{pmatrix}
                                    1 & 0 & 0 \\
                                    0 & 0 & 1 \\
                                    0 & 1 & 0
                                  \end{pmatrix}, &
                                  \xi \in {}& \realR_-.
        \end{aligned}
      \right.
    \end{equation}
  \item \label{enu:RHP:model_2}
    $\Phi^{(\gamma, \alpha, \rho, \sigma, \tau)}(\xi)$ has the following boundary condition as $\xi \to \infty$
    \begin{equation} \label{eq:expansion_at_infty_Phi}
      \Phi^{(\gamma, \alpha, \rho, \sigma, \tau)}(\xi) = (I + \bigO(\xi^{-1})) \Upsilon(\xi) \Omega_{\pm} e^{-\Theta^{(\rho, \sigma, \tau)}(\xi)},
    \end{equation}
    where
    \begin{align} \label{def:Lpm}
      \Upsilon(\xi) = {}& \diag(1, \omega \xi^{\frac{1}{3}}, \omega^2 \xi^{\frac{2}{3}}),
      & \Omega_+ = {}&
        \begin{pmatrix}
          1 & 1 & 1 \\
          1 & \omega & \omega^2 \\
          1 & \omega^2 & \omega
        \end{pmatrix},
      & \Omega_- = {}&
        \begin{pmatrix}
          1 & 1 & 1 \\
          \omega & 1 & \omega^2 \\
          \omega^2 & 1 & \omega
        \end{pmatrix}
    \end{align}
with $\omega = e^{2\pi i/3}$, and
    \begin{multline} \label{eq:Theta}
      \Theta^{(\rho, \sigma, \tau)}(\xi) = \\
      \begin{cases}
        \diag \left( \rho \omega \xi^{\frac{4}{3}} + \sigma \omega^2 \xi^{\frac{2}{3}} + \tau \omega \xi^{\frac{1}{3}}, \rho \omega^2 \xi^{\frac{4}{3}} + \sigma\omega \xi^{\frac{2}{3}} + \tau  \omega^2 \xi^{\frac{1}{3}}, \rho \xi^{\frac{4}{3}} + \sigma \xi^{\frac{2}{3}} + \tau  \xi^{\frac{1}{3}} \right),
        & \xi \in \compC_+, \\
        \diag \left( \rho \omega^2 \xi^{\frac{4}{3}} + \sigma \omega \xi^{\frac{2}{3}} + \tau \omega^2 \xi^{\frac{1}{3}}, \rho \omega \xi^{\frac{4}{3}} + \sigma \omega^2 \xi^{\frac{2}{3}} + \tau  \omega \xi^{\frac{1}{3}}, \rho \xi^{\frac{4}{3}} + \sigma\xi^{\frac{2}{3}} + \tau  \xi^{\frac{1}{3}} \right),
        & \xi \in \compC_-.
      \end{cases}
    \end{multline}
    
      \item \label{enu:RHP:model_3}
    $\Phi^{(\gamma, \alpha, \rho, \sigma, \tau)}(\xi)$ has the following boundary condition as $\xi \to 0$
    \begin{equation} \label{eq:limit_Phi_at_0_generic}
      \Phi^{(\gamma, \alpha, \rho, \sigma, \tau)}(\xi) = N^{(\gamma, \alpha, \rho, \sigma, \tau)}(\xi) \diag \left( \xi^{\frac{1}{2} - \frac{\alpha}{3}}, \xi^{\frac{1}{2} + \frac{\alpha}{6}}, \xi^{\frac{\alpha}{6}} \right) E,
    \end{equation}
    where $N^{(\gamma, \alpha, \rho, \sigma, \tau)}(\xi)$ is analytic in a neighbourhood of $0$, and $E$ is a piecewise constant function on $\compC \setminus \Gamma_{\Phi}$ (except for possibly one $\log \xi$ entry). If $\alpha$ is not an integer, $E$ is a piecewise constant function on $\compC \setminus \Gamma_{\Phi}$, and in the sector $\arg \xi \in (\pi/4, 3\pi/4)$, we have
    \begin{equation} \label{eq:generic_E}
      E =
      \begin{pmatrix}
        1 & 0 & 0 \\
        \frac{e^{\frac{2\alpha}{3} \pi i}}{1 - e^{\alpha \pi i}} (1 - \gamma) & 1 - \gamma & -e^{\frac{\alpha}{3} \pi i} \\
        \frac{e^{\frac{2\alpha}{3} \pi i}}{1 + e^{\alpha \pi i}} (1 + \gamma) & 1 + \gamma & e^{\frac{\alpha}{3} \pi i}
      \end{pmatrix}.
    \end{equation}
    If $\alpha$ is an odd integer, in the sector $\arg \xi \in (\pi/4, 3\pi/4)$, we have
    \begin{equation}
      E = 
      \begin{pmatrix}
        1 & 0 & 0 \\
        \frac{e^{\frac{2\alpha}{3} \pi i}}{1 - e^{\alpha \pi i}} (1 - \gamma) & 1 - \gamma & -e^{\frac{\alpha}{3} \pi i} \\
        -\frac{e^{\frac{2\alpha}{3} \pi i}}{2\pi i} (1 + \gamma) \log \xi & 1 + \gamma & e^{\frac{\alpha}{3} \pi i}
      \end{pmatrix}.
    \end{equation}
    Especially, if $\alpha$ is an odd number and $\gamma = -1$, $E$ is a piecewise constant function on $\compC \setminus \Gamma_{\Phi}$, and in the sector $\arg \xi \in (\pi/4, 3\pi/4)$, we have
    \begin{equation} \label{eq:E_odd}
      E =
      \begin{pmatrix}
        1 & 0 & 0 \\
        e^{\frac{2\alpha}{3} \pi i} & 2 & -e^{\frac{\alpha}{3} \pi i} \\
        0 & 0 & e^{\frac{\alpha}{3} \pi i}
      \end{pmatrix}.
    \end{equation}
    If $\alpha$ is an even integer, in the sector $\arg \xi \in (\pi/4, 3\pi/4)$, we have
    \begin{equation}
      E = 
      \begin{pmatrix}
        1 & 0 & 0 \\
        -\frac{e^{\frac{2\alpha}{3} \pi i}}{2\pi i} (1 - \gamma) \log \xi & 1 - \gamma & -e^{\frac{\alpha}{3} \pi i} \\
        \frac{e^{\frac{2\alpha}{3} \pi i}}{1 + e^{\alpha \pi i}} (1 + \gamma) & 1 + \gamma & e^{\frac{\alpha}{3} \pi i}
      \end{pmatrix}.
    \end{equation}
    Especially, if $\alpha$ is an even number and $\gamma = 1$, $E$ is a piecewise constant function on $\compC \setminus \Gamma_{\Phi}$, and in the sector $\arg \xi \in (\pi/4, 3\pi/4)$,
    \begin{equation} \label{eq:E}
      E =
      \begin{pmatrix}
        1 & 0 & 0 \\
        0 & 0 & -e^{\frac{\alpha}{3} \pi i} \\
        e^{\frac{2\alpha}{3} \pi i} & 2 & e^{\frac{\alpha}{3} \pi i}
      \end{pmatrix}.
    \end{equation}
  \end{enumerate}
\end{RHP}

From the behavior of $\Phi^{(\gamma, \alpha, \rho, \sigma, \tau)}(\xi)$ at $\infty$ and $0$ and its jump on $\Gamma_{\Phi}$, we have
\begin{equation}\label{eq:Phi_det}
\det \Phi^{(\gamma, \alpha, \rho, \sigma, \tau)}(\xi)=\mp 3\sqrt{3} i \xi, 
\end{equation}
for $\xi \in \compC_{\pm}$. Therefore $\tilde{\Phi}^{(\gamma, \alpha, \rho, \sigma, \tau)}(\xi)^{-1} $ is well defined.  

\begin{rmk}\label{rem:Pearcey_red}
  When $(\gamma,\alpha)=(1,0)$ or $(-1,1)$, the RH problem for $\Phi^{(\gamma, \alpha, \rho, \sigma, \tau)}(\xi)$ is equivalent to the hard-edge Pearcey parametrix $\Psi^{ (\HPe, \alpha') }$, with the corresponding parameter $\alpha'=-1/2$ and $\alpha'=1/2$, respectively (\cite{Kuijlaars-Martinez_Finkelshtein-Wielonsky11} and \cite[RH problem 3.1]{Dai-Xu-Zhang23}), after a gauge transformation. 
  Specifically, we have 
  \begin{equation}\label{eq:Psi_HPeaecey}
    \Phi^{(1, 0, 0, \frac{3}{2}, \tau)}(\xi)=\sqrt{3}i\xi^{\frac{1}{2}}C_{\Psi}(\tau)^{T} (\Psi^{(\HPe, -\frac{1}{2})} (\xi, \tau)^{-1})^T
    \begin{cases}
      \diag(i,-i,1), & \arg(\xi)\in(0,\pi),\\
      \diag(-i,-i,1), & \arg(\xi)\in(-\pi,0),\\
    \end{cases}
  \end{equation}
  and 
     \begin{equation}\label{eq:Psi_HPeaecey_2}
    \Phi^{(-1, 1, 0, \frac{3}{2}, \tau)}(\xi)=\sqrt{3}i\xi^{\frac{1}{6}}C_{\Psi}(\tau)^{T} (\Psi^{(\HPe, \frac{1}{2})} (\xi, \tau)^{-1})^T
    \begin{cases}
       \diag(-e^{-\frac{1}{6}\pi i},-e^{\frac{1}{6}\pi i},1), & \arg(\xi)\in(0,\pi),\\
   \diag(-e^{\frac{1}{6}\pi i},e^{-\frac{1}{6}\pi i},1), & \arg(\xi)\in(-\pi,0),\\
    \end{cases}
  \end{equation}
  where the matrix $C_{\Psi}(\tau)$ is an upper triangular matrix with all diagonal entries equal to $1$; see \cite[Equations (3.18)--(3.20)]{Dai-Xu-Zhang23}. For other values of $\alpha$, similar to \eqref{eq:Psi_HPeaecey}, $\Phi^{(\gamma, \alpha, \rho, \sigma, \tau)}(\xi)$ and $\xi^{\frac{1-\alpha'}{3}} (\Psi^{(\HPe,  \alpha')} (\xi, \tau)^{-1})^{T}$ have the same singularity at infinity. However, $\Phi^{(\gamma, \alpha, \rho, \sigma, \tau)}(\xi)$ has a regular singularity at the origin with characteristic exponents $\left(\frac{1}{2} - \frac{\alpha}{3}, \frac{1}{2} + \frac{\alpha}{6},\frac{\alpha}{6} \right)$, which are pairwise distinct whenever $\alpha\neq 0,1$. In contrast, $\xi^{\frac{\alpha'-1}{3}}(\Psi^{(\HPe,  \alpha')} (\xi, \tau)^{-1})^{T}$ has a regular singularity at the origin with characteristic exponents $\left( \frac{1-\alpha'}{3}, \frac{2\alpha'+1}{3}, \frac{1-\alpha'}{3}\right)$ among which two coincide. Consequently, when $\alpha\neq 0,1$, $\Phi^{(\gamma, \alpha, \rho, \sigma, \tau)}(\xi)$ is not gauge equivalent to $\xi^{\frac{\alpha'-1}{3}}(\Psi^{(\HPe,  \alpha')} (\xi, \tau)^{-1})^{T}$ and therefore not gauge equivalent to $\Psi^{(\HPe, \alpha')}(\xi) $.
\end{rmk}

We have the following result regarding the solvability of the model RH problem. 
\begin{thm} \label{thm:RHP_unique_solvability}
  Let $\alpha>-1$ and $\gamma\in[-1,1]$. Under either of the two conditions:
  \begin{enumerate}
  \item \label{enu:thm:RHP_unique_solvability:1}
    $\rho > 0$, $\sigma \in \realR$ and $\tau \in \realR$, or
  \item \label{enu:thm:RHP_unique_solvability:2}
    $\rho = 0$ and $\sigma > 0$ and $\tau \in \realR$,
  \end{enumerate}
  RH problem \ref{RHP:model} has a unique solution. Moreover, $\Phi^{(\gamma, \alpha, \rho, \sigma, \tau)}$ has the asymptotic expansion at $\infty$
  \begin{equation} \label{eq:asy_expansion_Phi}
    \Phi^{(\gamma, \alpha, \rho, \sigma, \tau)}(\xi) = \left( I + \frac{M}{\xi} + \bigO(\xi^{-2}) \right) \Upsilon(\xi) \Omega_{\pm} e^{-\Theta(\xi)}, \quad M = (m^{(\gamma, \alpha)}_{ij}(\rho, \sigma, \tau))^3_{i, j = 1},
  \end{equation}
  where $m^{(\gamma, \alpha)}_{ij}(\rho, \sigma, \tau)$ are analytic in $\rho, \sigma, \tau$.
\end{thm}

Next, we establish the following integrability property for $\Phi = \Phi^{(\gamma, \alpha, \rho, \sigma, \tau)}$. For convenience, we introduce the transformation 
\begin{equation}\label{eq:transform}
  \Phihat=P\Phi, \quad \text{where} \quad P =
  \begin{pmatrix}
    1 & 0 & 0 \\
    m_{13} & 1 & 0 \\
    2m^2_{13} + 2m_{23} - m_{12}  & m_{13} & 1
  \end{pmatrix},
\end{equation}
with $m_{ij} = m^{(\gamma, \alpha)}_{ij}(\rho, \sigma, \tau)$ in \eqref{eq:asy_expansion_Phi}. We denote
\begin{align}
  u = {}& u^{(\gamma, \alpha)}(\rho, \sigma, \tau) = -3(m^2_{13} + m_{23} - m_{12}), \label{eq:uv_u} \\
  v = {}& v^{(\gamma, \alpha)}(\rho, \sigma, \tau) = 3(m^3_{13} - m_{12} m_{13} + 2m_{13} m_{23} - m_{22} + m_{33}). \label{eq:uv_v}
\end{align}
Later we will see that
\begin{align} \label{eq:u_v_expr}
  u = {}& 3(m_{13})_{\tau}, & v = {}& -\frac{3}{2} (m^2_{13} + 2m_{23})_{\tau}.
\end{align}
Then we have the following Lax pair for $\Phihat$, the compatibility of which is described by the Boussinesq hierarchy and their self-similarity reduction.

\begin{thm} \label{prop:boussinesq}
  Let $\Phihat$ be defined in \eqref{eq:transform}. We have the Lax pair 
\begin{equation} \label{eq:LaxTriple}
  -\xi \frac{\partial}{\partial \xi} \Phihat = \hat{L} \Phihat, \quad -\frac{\partial}{\partial \tau} \Phihat = \hat{A} \Phihat, \quad   -\frac{\partial}{\partial \sigma} \Phihat = \hat{B} \Phihat, \quad  -\frac{\partial}{\partial \rho} \Phihat =  \hat{C} \Phihat, 
\end{equation}
where
\begin{equation}\label{eq:Lhat}
  \hat{L}=-\frac{1}{3} \diag(0,1,2)+\frac{\tau}{3}\hat A +\frac{2\sigma}{3} \hat B+\frac{4\rho}{3}\hat C,
\end{equation}
\begin{align}\label{eq:Ahat}
  \hat A= {}&
  \begin{pmatrix}
    0 & 1 & 0 \\
    0 & 0 & 1 \\
    \xi +v& -u& 0
  \end{pmatrix}, &
  \hat B= {}&
  \begin{pmatrix}
    0 & 0 & 1 \\
    \xi & 0 & 0 \\0 & \xi & 0
  \end{pmatrix}
  +
  \begin{pmatrix}
    \frac{2}{3}u & 0 & 0 \\
    v - \frac{2}{3}u_{\tau} & - \frac{1}{3}u & 0 \\
    -v_{\tau} + \frac{2}{3}u_{\tau \tau }& v - \frac{1}{3}u_{\tau} & - \frac{1}{3}u
  \end{pmatrix},
\end{align}
\begin{equation}\label{eq:Chat}
  \hat C=\xi^2
  \begin{pmatrix}
    0 & 0 & 0 \\
    0& 0 & 0 \\1 &0 & 0
  \end{pmatrix}
  +\xi
  \begin{pmatrix}
    0 & 1 & 0 \\
    \frac{1}{3}u & 0 & 1 \\
    -\frac{1}{3}u_{\tau} + \frac{2}{3}v & -\frac{2}{3}u & 0
  \end{pmatrix}
  +
  \begin{pmatrix}
    c_{11} & -\frac{1}{3}(v - u_{\tau} )& \frac{1}{3}u \\
    c_{21}& c_{22} & -\frac{1}{3}v\\
    c_{31} &c_{32} & c_{33}
  \end{pmatrix},
\end{equation}
with
\begin{align}
  c_{11} = {}&\frac{1}{9}(2u_{\tau \tau}-3v_{\tau} +2u^2), \quad c_{22} = \frac{1}{9}(-u_{\tau \tau} - u^2), \quad c_{33} = \frac{1}{9}(-u_{\tau \tau}+3v_{\tau} -u^2), \label{eq:c11} \\
  c_{21} = {}&\frac{1}{9}(-2u_{\tau \tau \tau}+3v_{\tau \tau}-2(u^2)_{\tau} + 3uv), \label{eq:c21} \\
  c_{31} = {}&\frac{1}{9}(2u_{\tau \tau \tau \tau}-3v_{\tau \tau \tau}+2(u^2)_{\tau \tau} - 3(uv)_{\tau} -3v^2),\\
  c_{32}={}& \frac{1}{9}(-u_{\tau \tau \tau}-(u^2)_{\tau}+3v_{\tau \tau} + 6uv). \label{eq:c32}
\end{align}  
The compatibility of the Lax pair is described by the first equation in the Boussinesq hierarchy \cite[Section 1.4.7]{Dickey03}
\begin{align} \label{eq:1st_Bous_hier}
  \frac{\partial u}{\partial \sigma} = {}& u_{\tau \tau} - 2v_{\tau}, &
  \frac{\partial v}{\partial \sigma} = {}& -v_{\tau \tau} + \frac{2}{3} u_{\tau \tau \tau} + \frac{2}{3} uu_{\tau},
\end{align}
which implies the classical Boussinesq equation \cite[Equation (1.4.8)]{Dickey03}
  \begin{equation}\label{eq: Boussinesq equation}
    u_{\sigma\sigma} = -\frac{1}{3} u_{\tau \tau \tau \tau} - \frac{4}{3}(u u_{\tau})_{\tau}
\end{equation}
and the next member in the Boussinesq hierarchy \cite[Equation (4.6)]{Wang-Mei21}
  \begin{align}
    u_{\rho} = {}& \frac{1}{3}(u_{\tau \tau \tau \tau} - 2v_{\tau \tau \tau} + (u^2)_{\tau \tau} - 4(uv)_{\tau} ), \label{eq:2nd_Boussinesq_hierarchy:1} \\  
    v_{\rho}= {}& \frac{1}{9}(2u_{\tau \tau \tau \tau \tau} - 3v_{\tau \tau \tau \tau} + 2(u^2)_{\tau \tau \tau} + 2uu_{\tau \tau \tau} - 6(uv_{\tau} )_{\tau}  - 6(v^2)_{\tau}  + 4u^2u_{\tau}), \label{eq:2nd_Boussinesq_hierarchy:2}
\end{align}
and their self-similarity reduction
\begin{align}
  &\frac{4}{3}\rho \left( -u_{\tau \tau \tau \tau } + 2v_{\tau \tau \tau } - (u^2)_{\tau \tau } + 4(uv)_\tau  \right) + 2\sigma (2v_\tau  - u_{\tau \tau }) - \tau  u_\tau  - 2u = 0, \label{eq:Boussinesq_equation_coupled_equation:1} \\
  &\frac{4}{9}\rho \left( -2u_{\tau \tau \tau \tau \tau } + 3v_{\tau \tau \tau \tau } - 2(u^2)_{\tau \tau \tau } - 2uu_{\tau \tau \tau } + 6(uv_\tau )_\tau  + 6(v^2)_\tau  - 4u^2 u_\tau  \right) \nonumber \\
&~~~~~~~~~~~~~~~~~~~~~~~~~~~~~~~~~~~ + 2\sigma \left( v_{\tau \tau } - \frac{2}{3}u_{\tau \tau \tau } - \frac{1}{3}(u^2)_\tau  \right) - \tau v_\tau  - 3v = 0. \label{eq:Boussinesq_equation_coupled_equation:2}
\end{align}
Particularly, in the case $\rho=0$, the self-similarity reduction can be written as a single equation
\begin{equation}\label{eq:Boussinesq_equation_Reduction}
  4\sigma^2 (-u_{\tau \tau \tau \tau } - 2(u^2)_{\tau \tau }) - 3\tau ^2 u_{\tau \tau } - 21\tau  u_\tau  - 24u = 0.
\end{equation}
Taking $\sigma = 3/2$ in \eqref{eq:Boussinesq_equation_Reduction} without loss of generality, and denoting
\begin{equation} \label{eq:defn_f(tau)}
  f(\tau) = m_{13}(\gamma, \alpha, 0, \frac{3}{2}, \tau),
\end{equation}
then we have Chazy's equations
\begin{equation}\label{eq:Third_order_ODE} 
  y_{\tau \tau \tau}+6y_\tau ^2+\tau y-\frac{1}{72} \tau ^4+\frac{1}{6}(\alpha-\alpha^2)=0, \quad y(\tau )=f(\tau )+\frac{\tau ^3}{108},
\end{equation}
\begin{equation}\label{eq:second_order_ODE} 
  (\tilde{y})_{\tau \tau } ^2+4\tilde{y}_\tau ^3-4(\tau \tilde{y}_\tau -\tilde{y})^2-\frac{4}{3}(\alpha^2-\alpha + 1)\tilde{y}_\tau +\frac{4}{27}(\alpha+1)(2\alpha-1)(\alpha-2)=0, \quad \tilde{y}(\tau )=\sqrt{2}f(\sqrt{2}\tau )+\frac{4}{27}\tau ^3,
\end{equation}
and \eqref{eq:second_order_ODE} is also equivalent to the Painlev\'{e} IV equation \cite{Cosgrove-Scoufis93}. 
\end{thm}

\begin{rmk} \label{rmk:Beq_def}
  In the literature, the Boussinesq hierarchy is identified with the $3$rd Gelfand-Dickey (GD) hierarchy. Generally, for $n\geq 2$, the $n$th GD hierarchy is defined via the Lax equation
  \begin{equation} \label{eq:Lax_Eq}
    \frac{d}{dt_k}L=[(L^{k/n})_+, L]=(L^{k/n})_+L-L(L^{k/n})_+, \quad k=1,2,3,\dots
  \end{equation}
  where $L$ is a scalar differential operator of order $n$:
  \begin{equation} \label{eq:L_Operator}
    L=\partial^n+u_{n-2}\partial^{n-2}+\cdots+u_0, \quad \partial=\frac{d}{dx}=\frac{d}{dt_1},
  \end{equation}
  and $(L^{k/n})_+$ denotes the purely differential part of the fractional powers of the operator $L^{k/n}$; see \cite[Section 1.4]{Dickey03} for the definition. For $n=2$, the GD hierarchy reduces to the well-known KdV hierarchy. For $n=3$, one obtains the Boussinesq hierarchy. In this case, we write the Lax operator as
  \begin{equation} \label{eq:L_Operator_3}
    L=\partial^3+u\partial+v.
  \end{equation}
  Its fractional powers admit the following expansions
  \begin{equation} \label{eq:L_Operator_1/3}
    L^{1/3}=\partial+\frac{u}{3}\partial^{-1}+\frac{1}{3}(v-u')\partial^{-2}+\frac{1}{9}(2u''-3v'-u^2)\partial^{-3}+\bigO(\partial^{-4}),
  \end{equation}
  and then
  \begin{align} \label{eq:L_Operator_2/3}
    (L^{2/3})_+ = {}& \partial^2+\frac{2u}{3}, & (L^{4/3})_+ = {}& \partial^4+\frac{4u}{3}\partial^2+\frac{2}{3}(u'+2v)\partial+\frac{2}{9}(u''+3v'+u^2).
  \end{align}
  Substituting \eqref{eq:L_Operator_2/3} into the Lax equation yields the first and second nontrivial flows of the Boussinesq hierarchy as given in \eqref{eq:1st_Bous_hier}, \eqref{eq:2nd_Boussinesq_hierarchy:1} and \eqref{eq:2nd_Boussinesq_hierarchy:2}, respectively, under the identification $(x,t_2,t_4)=(\tau,- \sigma,-\rho)$. 
\end{rmk}

\begin{rmk}\label{rmk:poly_solution}
  In the special cases
  \begin{enumerate*}[label=(\arabic*)]
  \item \label{enu:rmk:poly_solution_1}
    $\rho = 0$, $\sigma = 3/2$, $\alpha=0$ and $\gamma=1$,
  \item \label{enu:rmk:poly_solution_2} 
    $\rho = 0$, $\sigma = 3/2$, $\alpha=1$ and $\gamma= -1$, 
  \end{enumerate*}
  we show at the end of Section \ref{sec:boussinesq} that 
  \begin{align} \label{eq:poly_f}
    f(\tau) = {}& -\frac{1}{27}\tau^3 \pm \frac{\tau}{6}, & u(\tau) = {}& -\frac{1}{3}\tau^2 \pm \frac{1}{2}, & y(\tau) = {}& -\frac{1}{36}\tau^3 \pm \frac{\tau}{6}, & \tilde{y}(\tau) = \pm \frac{\tau}{3},
  \end{align}
  where the sign is $+$ in Case \ref{enu:rmk:poly_solution_1} and $-$ in Case \ref{enu:rmk:poly_solution_2}.
\end{rmk}

\begin{rmk}
  Since $u(\tau) = u(\gamma, \alpha, 0, 3/2, \tau)$ satisfies \eqref{eq:Boussinesq_equation_Reduction} with $\sigma = 3/2$, we have that $\hat{v}(\tau) = 2 \cdot 3^{1/2} u(2^{-1/2} \cdot 3^{1/4} \tau)$ satisfies \cite[Equation (59)]{Wang-Xu25} where $\hat{v}(\tau)$ substitutes $v(\tau)$ there. Also Equations \eqref{eq:Third_order_ODE} coincides with \cite[Equation (57)]{Wang-Xu25}, and \eqref{eq:second_order_ODE} coincides with \cite[Equation (58)]{Wang-Xu25} where $u(\tau)$ is the $\tilde{y}(\tau)$ in \eqref{eq:second_order_ODE}. In fact, if $\gamma = 0$, $\rho = 0$ and $\sigma = 3/2$, by comparing RH problem \ref{RHP:model} and \cite[RH problem 1.8]{Wang-Xu25} with $\theta = 2$ and asymptotic expansion \cite[Equation (415)]{Wang-Xu25}, we find that $f(\tau)$ defined in \eqref{eq:defn_f(tau)} is equal to $c(2^{-1/2} \tau)$ where $c$ is defined in \cite[Equations (45) and (417)]{Wang-Xu25}. Hence, \cite[Equations (59), (57) and (58)]{Wang-Xu25} are special cases of \eqref{eq:Boussinesq_equation_Reduction}, \eqref{eq:Third_order_ODE} and \eqref{eq:second_order_ODE} with $\gamma = 0$. See \cite[Remark 1.12]{Wang-Xu25} for a literature review of the three formulas in integrable systems.
\end{rmk}

\begin{rmk}\label{rem:n_Boussinesq_hierarchy}The integrable structure of the model RH problem $\Phi^{(\gamma, \alpha, \rho, \sigma, \tau)}(\xi)$ is associated with the first two members of the Boussinesq hierarchy. It is expected that the RH problem with the exponent $\Theta^{(\rho, \sigma, \tau)}(\xi)$ in \eqref{eq:expansion_at_infty_Phi} replaced by the general form 
  \begin{equation}
    \Theta(\xi)=\diag\left(\sum_{k=1}^m t_k \omega^{jk}\xi^{\frac{k}{3}}\right)_{j=1}^3,
  \end{equation}
where $t_{k}=0$ if $k=3j$, would correspond to the higher members of the Boussinesq hierarchy.
\end{rmk}

\subsection{Alternative representation and specializations of the result}

From RH problem \ref{RHP:model}, we define another RH problem:
\begin{RHP} \label{RHP:model_Psi}
  $\Psi^{(\gamma, \alpha, \rho, \sigma, \tau)}(\xi)$ ($\Psi(\xi)$ for short) is a $3 \times 3$ vector-valued function that is analytic on $\compC \setminus \Gamma_{\Psi}$ and continuous up to the boundary, where $\Gamma_{\Psi} = \realR_+ \cup \realR_- \cup \{ \arg \xi = \pm \pi/8 \} \cup \{ \arg \xi = \pm 3\pi/8 \}\cup \{ \arg \xi = \pm 5\pi/8 \} \cup \{ \arg \xi = \pm 7\pi/8 \}$, with all rays oriented outward from $0$. It satisfies the following properties.
  \begin{enumerate}
  \item \label{enu:RHP:model_Psi}
    $\Psi^{(\gamma, \alpha, \rho ,\sigma, \tau)}_+(\xi) = \Psi^{(\gamma, \alpha, \rho, \sigma, \tau)}_-(\xi) J_{\Psi}(\xi)$, where 
     \begin{equation} \label{eq:RHP:model_Psi}
      J_{\Psi}(\xi) =
      \left\{
      \begin{aligned}
        & \begin{pmatrix}
          1 & \pm e^{\mp \frac{2\alpha}{3} \pi i} & 0 \\
          0 & 1 & 0 \\
          0 & 0 & 1
        \end{pmatrix}, &
        \arg \xi = {}&\pm  \frac{\pi}{8}, \\
        & \begin{pmatrix}
          1 & 0 & 0 \\
          0 & 1 & 0 \\
          0 & \mp\gamma e^{\mp \frac{\alpha}{3} \pi i} & 1
        \end{pmatrix}, &
        \arg \xi = {}& \pm \frac{3\pi}{8}, \\
         & \begin{pmatrix}
          1 & 0 & 0 \\
          0 & 1 & \pm \gamma e^{\pm \frac{\alpha}{3} \pi i}  \\
          0 & 0& 1
        \end{pmatrix}, &
        \arg \xi = {}& \pm \frac{5\pi}{8},\\
         & \begin{pmatrix}
          1 &0 &  \mp e^{\pm \frac{2\alpha}{3} \pi i}\\
          0 & 1 & 0 \\
          0 & 0 & 1
        \end{pmatrix}, &
        \arg \xi = {}& \pm \frac{7\pi}{8},\\
         & \begin{pmatrix}
          0 & 1 & 0 \\
          1 & 0 & 0 \\
          0 & 0 & 1
        \end{pmatrix}, &
        \xi \in {}& \realR_+, \\
     &   \begin{pmatrix}
          0 & 0 & 1 \\
          0 & 1 & 0 \\
          1 & 0 & 0
        \end{pmatrix}, &
         \xi \in {}& \realR_-.
      \end{aligned}
    \right.
    \end{equation}
  \item 
    $\Psi^{(\gamma, \alpha, \rho, \sigma, \tau)}(\xi)$ has the following boundary condition as $\xi \to \infty$
    \begin{equation} \label{eq:RHP_Psi_infty}
      \Psi^{(\gamma, \alpha, \rho, \sigma, \tau)}(\xi) = (I +\bigO(\xi^{-1})) \diag (1, \omega \xi^{-\frac{1}{3}}, \omega^2 \xi^{\frac{1}{3}})\Omega_{\pm} e^{-\Theta(\xi^2)}.
    \end{equation}
  \item 
    $\Psi^{(\gamma, \alpha, \rho, \sigma, \tau)}(\xi)$ has the following boundary condition as $\xi \to 0$
    \begin{equation} \label{eq:limit_Psi_at_0_generic}
      \Psi^{(\gamma, \alpha, \rho, \sigma, \tau)}(\xi) = \hat{N}(\xi) \diag \left( \xi^{ - \frac{2\alpha}{3}}, \xi^{\frac{\alpha}{3}}, \xi^{\frac{\alpha}{3}} \right) \hat{E}(\xi),
    \end{equation}
    where $\hat{N}(\xi)$ is analytic in a neighbourhood of $0$, and $\hat{E}(\xi)=E(\xi^2)$ for $\arg \xi \in(\pi/8, 3\pi/8)$ and the expression for $\hat{E}$ in other sectors is determined by the jump relations.
  \end{enumerate}
\end{RHP}
\begin{prop}\label{pro:sol_Psi}
  Under the same conditions as in Theorem \ref{thm:RHP_unique_solvability}, RH problem \ref{RHP:model_Psi} has a unique solution given by
  \begin{equation} \label{eq: Psi_def}
    \Psi  ^{(\gamma, \alpha, \rho, \sigma, \tau)}(\xi)=\xi^{-1} H^{(-\gamma, \alpha + 1, \rho, \sigma, \tau, +)}(\xi^2)
    \begin{cases}
      \Phi  ^{(\gamma, \alpha, \rho, \sigma, \tau)}(\xi^2) , & \arg \xi \in (-\frac{\pi}{2}, \frac{\pi}{2}), \\
      \Phi  ^{(\gamma, \alpha,\rho, \sigma, \tau)}(e^{\mp 2\pi i}\xi^2)(1\oplus\sigma_1) , &\pm  \arg \xi \in( \frac{\pi}{2}, \pi),
    \end{cases}
  \end{equation}
  where $\Phi  ^{(\gamma, \alpha, \rho, \sigma, \tau)}(\xi)$ is the unique solution of RH problem \ref{RHP:model} and $H^{(-\gamma, \alpha + 1, \rho, \sigma, \tau, +)}(\xi)$ is defined as
  \begin{equation}\label{eq:H}
    H^{(\gamma, \alpha, \rho, \sigma, \tau, +)}(\xi)=
    \begin{pmatrix}
      \xi^{\frac{1}{2}} & 0 & 0\\
      m^{(\gamma, \alpha)}_{13}(\rho, \sigma, \tau)  & 1 & 0 \\
      m^{(\gamma, \alpha)}_{23}(\rho, \sigma, \tau)  & 0 & 1
    \end{pmatrix},
  \end{equation}
  where $m^{(\gamma, \alpha)}_{ij}(\rho, \sigma, \tau)$ is defined in \eqref{eq:asy_expansion_Phi}.
\end{prop}   

From RH problem \ref{RHP:model_Psi}, we have alternative expressions for $K^{(\Pe, \alpha)}(\xi, \eta; 2\tau)$ and $K^{(\tri, \alpha)}(\xi, \eta; \sigma, \tau)$. Denote $\tilde{ \Psi}  ^{(\gamma, \alpha, \rho, \sigma, \tau)}(\xi)$ the analytic extension of $  \xi^{\frac{2\alpha}{3}} \Psi  ^{(\gamma, \alpha, \rho, \sigma, \tau)}(\xi)$ from the sector $\{ \arg \xi \in (\pi/8, 3\pi/8) \}$ to the complex plane with a cut on $(-\infty i, 0]$. We have the following proposition. 
\begin{prop}\label{pro:sum_kernel_Psi}
  For $\xi, \eta\in \mathbb{R}\setminus\{0\}$, we have
          \begin{align}\label{eq:sum_kernel_Psi}
 \xi K^{(\gamma, \alpha, \rho, \sigma, \tau)}(\xi^2, \eta^2)& + \eta K^{(-\gamma, \alpha+1, \rho, \sigma, \tau)}(\xi^2, \eta^2)= \frac{e^{\frac{2}{3}\alpha \pi i}}{2\pi i (\xi-\eta)} \nonumber\\
& \times
\begin{pmatrix}
      1 & 0 & 0
    \end{pmatrix}
\tilde{\Psi}^{(\gamma, \alpha, \rho, \sigma, \tau)}(\eta)^{-1}\tilde{\Psi}^{(\gamma, \alpha, \rho, \sigma, \tau)}(\xi)
    \begin{pmatrix}
      0 \\ 1 \\ 0
    \end{pmatrix}.
  \end{align}
\end{prop}

\begin{rmk}\label{rem:Pearcey_red_2}
  When $\gamma=1$, $\alpha=0$, $\rho=0$ and $\sigma = 3/4$, after an elementary transformation the RH problem for $\Psi^{(\gamma, \alpha, \rho, \sigma, \tau)}(\xi)$ is equivalent to that of the Pearcey parametrix $\Psi^{(\Pe)}$ given in \cite[Section 8]{Bleher-Kuijlaars07}; see also \cite[RH problem 2.1]{Dai-Xu-Zhang21}. Specifically, in the sector $\{ \arg \xi\in(3\pi/8, 5\pi/8) \}$, we have 
  \begin{equation}\label{eq:Psi_Pe}
    \Psi^{(1, 0, 0, \frac{3}{4}, \tau)}(\xi)=C_0(\tau) (\Psi^{(\Pe)} (\xi; 2\tau)^{-1})^T \diag(-1,1,1),
  \end{equation}
  where $\Psi^{(\Pe)}(\xi; 2\tau)$ is $\Psi(z; \rho)$ in \cite[RH problem 2.1]{Dai-Xu-Zhang21} with $z = \xi$ and $\rho = 2\tau$ there, and the matrix $C_0(\tau)$ is given by
  \begin{equation}
    C_0(\tau)= \sqrt{6\pi} i e^{\frac{2}{3}\tau^2}
    \begin{pmatrix}
      0 & 1 & 0 \\
      0 & 0 & 1 \\
      1& 0 & 0
    \end{pmatrix} 
    \begin{pmatrix}
      1 & 0 & \frac{4}{27}\tau^3+\tau \\
      0 & 1 & 0 \\
      0& 0 & 1
    \end{pmatrix}.
  \end{equation}
  From \eqref{eq:RHP:model_Psi} and the definition of $\tilde{ \Psi}  ^{(\gamma, \alpha, \rho, \sigma, \tau)}(\xi)$, we have in $\{ \arg \xi\in(3\pi/8, 5\pi/8) \}$
  \begin{equation}\label{eq:tildePsi_Pe}
    \tilde{\Psi}^{(1, 0, 0, \frac{3}{4}, \tau)}(\xi)= \Psi^{(1, 0, 0, \frac{3}{4}, \tau)}(\xi)
    \begin{pmatrix}
      1 & 0 &0 \\
      0 & 1 & 0 \\
      0& 1 & 1
    \end{pmatrix}
    =C_0(\tau) (\Psi^{(\Pe)} (\xi; 2\tau)^{-1})^T
    \begin{pmatrix}
      -1 & 0 &0 \\
      0 & 1 & 0 \\
      0& 1 & 1
    \end{pmatrix}.
  \end{equation}
  Substituting \eqref{eq:tildePsi_Pe} into \eqref{eq:limit_K^ext} and \eqref{eq:sum_kernel_Psi} yields 
  \begin{equation}\label{eq:Peaecey_red}
    \begin{split}
       K^{(\Pe, 0)}(\xi, \eta; 2\tau) =  & \xi K^{(1, 0, 0, \frac{3}{4}, \tau)}(\xi^2, \eta^2)+ \eta K^{(-1, 1, 0, \frac{3}{4}, \tau)}(\xi^2, \eta^2) \\
      = {}& \frac{1}{2\pi i (\eta-\xi)}
            \begin{pmatrix}
              1 & 0 & 0
            \end{pmatrix}
            \Psi^{(\Pe)}(\eta; 2\tau)^{T}(\Psi^{(\Pe)}(\xi; 2\tau)^{-1})^T
            \begin{pmatrix}
              0 \\ 1 \\ 1
            \end{pmatrix} \\
      = {}& \frac{1}{2\pi i (\eta-\xi)}
            \begin{pmatrix}
              0 & 1 & 1
            \end{pmatrix}
            \Psi^{(\Pe)}(\xi; 2\tau)^{-1} \Psi^{(\Pe)}(\eta; 2\tau)
            \begin{pmatrix}
              1 \\ 0 \\ 0
            \end{pmatrix} \\
      = {}& K^{(\Pe)}(\eta,\xi; 2\tau),
    \end{split}
  \end{equation}
  where $K^{(\Pe)}$ is the classical Pearcey kernel given in \eqref{eq:Pearcey_kernel}; see also \cite[Equation (10.19)]{Bleher-Kuijlaars07}. This, in turn, implies that for $\xi, \eta\in(0,\infty)$,
  \begin{align} \label{eq:Hard_Pearcey_kernel_1}
    K^{(1, 0, 0,\frac{3}{4},\tau)}(\xi, \eta) &=  \frac{1}{2\sqrt{\xi}}( K^{(\Pe)}(\sqrt{\eta},\sqrt{\xi}; 2\tau)+ K^{(\Pe)}(-\sqrt{\eta},\sqrt{\xi}; 2\tau))
                                                \nonumber \\
                                              & =\frac{p(\sqrt{\eta})q''(\sqrt{\xi})-\sqrt{\eta/\xi}p'(\sqrt{\eta})q'(\sqrt{\xi})+p''(\sqrt{\eta})q(\sqrt{\xi})-2\tau
                                                p(\sqrt{\eta})q(\sqrt{\xi})}{\eta-\xi},
  \end{align}
  and 
  \begin{align} \label{eq:Hard_Pearcey_kernel_2}
    K^{(-1, 1, 0,\frac{3}{4},\tau)}(\xi, \eta) &=  \frac{1}{2\sqrt{\eta}}( K^{(\Pe)}(\sqrt{\eta},\sqrt{\xi}; 2\tau)+ K^{(\Pe)}(\sqrt{\eta},-\sqrt{\xi}; 2\tau))
                                                 \nonumber \\
                                               & =\frac{p(\sqrt{\eta})q''(\sqrt{\xi})-\sqrt{\xi/\eta}p'(\sqrt{\eta})q'(\sqrt{\xi})+p''(\sqrt{\eta})q(\sqrt{\xi})-2\tau
                                                 p(\sqrt{\eta})q(\sqrt{\xi})}{\eta-\xi},
  \end{align}
  where the Pearcey functions $p$ and $q$ are defined by \eqref{eq:pearcey_integral} and we use that $p$ is even and $q$ is odd. From \eqref{eq:limit_kernel_mult}, \eqref{eq:Psi_HPeaecey} and \eqref{eq:Psi_HPeaecey_2}, we also have for $\xi, \eta\in(0,\infty)$
  \begin{equation} \label{eq:Hard_Pearcey_kernel_3}
    K^{(1, 0, 0,\frac{3}{2},\tau)}(\xi, \eta) =  K^{(\HPe,- \frac{1}{2})}(\eta,\xi;\tau), \quad K^{(-1, 1, 0,\frac{3}{2},\tau)}(\xi, \eta;\tau) =  K^{(\HPe, \frac{1}{2})}(\eta,\xi;\tau),
  \end{equation}
  where $K^{(\HPe,\alpha')}(\eta,\xi;\tau)$ is the hard-edge Pearcey kernel obtained in \cite[Equations (1.19) and (1.28)]{Kuijlaars-Martinez_Finkelshtein-Wielonsky11}. The relation \eqref{eq:Hard_Pearcey_kernel_3} is a special case of \cite[Proposition 1.2]{Liechty-Wang17}.
\end{rmk}

\subsection{Relation to other models}

\paragraph{Muttalib-Borodin ensemble}

The biorthogonal polynomials defined by \eqref{eq:biorthogonal_def} become the biorthogonal polynomials of the Muttalib-Borodin ensemble with parameter $\theta = 2$ and potential function $W(x)$ on $\realR_+$ if $\gamma = 0$ \cite[Equation (1.5)]{Claeys-Romano14}, \cite[Equation (1.13)]{Wang-Zhang21}. Hence, Theorem \ref{thm:main} implies limiting local statistics for the Muttalib-Borodin ensemble.
\begin{cor} \label{cor:MB}
  Let $n$ particles be in the Muttalib-Borodin ensemble of Laguerre type, such that they are distributed on $(0, \infty)$ and their joint probability density function is
  \begin{equation}
    \frac{1}{Z_n} \prod_{1 \leq i < j \leq n} (x_i - x_j)(x^2_i - x^2_j) \prod^n_{j = 1} x^{\alpha}_j e^{-n(V(x_j))}, \quad V(x) = \frac{x^4}{2} - tx^2 - 2ax, \quad \alpha > -1.
  \end{equation}
  Then as $n \to \infty$, the correlation kernel $K^{\MB}_n$ defined in \cite[Equation (13)]{Wang-Xu25} with $N = n$ and $V_t = V$ there, has the following limits for particles at $0$:
  \begin{enumerate}
  \item
    If $t$ and $a$ are parametrized by \eqref{eq:def_t_c_Pearcey} and $c_1$ is given in \eqref{eq:c_1_Pearcey}, then we have
    \begin{equation} \label{eq:MB_limit_Pearcey}
      \lim_{n \to \infty} \frac{1}{(c_1 n)^{3/4}} K^{\MB}_n(\frac{\xi}{(c_1 n)^{3/4}}, \frac{\eta}{(c_1n)^{3/4}}) = 2\xi K^{(0, \alpha, 0, \frac{3}{4},\tau)}(\xi^2, \eta^2). 
    \end{equation}
  \item
    If $t$ and $a$ are parametrized by \eqref{eq:def_t_c_critical} and \eqref{eq:c_multi_crit}, then we have
    \begin{equation} \label{eq:MB_limit_tri}
     \lim_{n \to \infty} \frac{e^{\frac{2}{3} n^{1/4}(\eta^2 - \xi^2)} }{ (3^{\frac{2}{3}} n)^{3/8}} K^{\MB}_{n}(\frac{\xi}{(3^{\frac{2}{3}} n)^{3/8}}, \frac{\eta}{(3^{\frac{2}{3}} n)^{3/8}})= 2\xi K^{(0, \alpha, -\frac{1}{4}, \frac{3}{2}\sigma, \tau)}(\xi^2, \eta^2).
    \end{equation}
  \end{enumerate}
\end{cor}
We note that the limit formula \eqref{eq:MB_limit_tri} is new, and \eqref{eq:MB_limit_Pearcey} is a special case of \cite[Theorem 1.6]{Wang-Xu25}. However, it is not trivial to see that the limiting kernel $2\xi K^{(0, \alpha, 0, \frac{3}{4},\tau)}(\xi^2, \eta^2)$ is the same as the kernel in \cite[Equation (21)]{Wang-Xu25}. This identification is justified in Section \ref{subsec:proof_MB}.

\paragraph{Universality of Wigner-type random matrices with external source}

Suppose $n$ is large, and $W = (w_{ij})^{2n}_{i, j = 1}$ is a Wigner-type Hermitian random matrix; that is, $w_{ij}$ ($i \leq j$) are independent random variables, with $\E w_{ij} = 0$, $\E \lvert \Re w_{ij} \rvert^2 = \E \lvert \Re w_{ij} \rvert^2 = 1/2$ if $i < j$ and $\E \lvert w_{ii} \rvert^2 = 1$. Let $f_{ij}(z)$ be the density function of $w_{ij}$, where $1 \leq i < j \leq 2n$ and $z \in \compC$ and $f_i(x)$ be the density function of $w_{ii}$ where $1 \leq i \leq 2n$ and $x \in \realR$. Let $A$ be defined in \eqref{eq:pdf_ext_source_intro} with $a = \sqrt{2n} + \bigO(1)$. We expect that if the Hermitian random matrix $X = (x_{ij})^{2n}_{i, j = 1}$ has the probability density proportional to
\begin{equation}
  \lvert \det(X) \rvert^{\alpha} \prod_{1 \leq i < j \leq 2n} f_{ij}(x_{ij}) \prod^n_{i =1} f_i(x_{ii} - a) \prod^{2n}_{j = n + 1} f_j(x_{jj} + a),
\end{equation}
then the local distribution around $0$ of eigenvalues of $X$, up to proper scaling, is approximated by the correlation kernel $K^{(\Pe, \alpha)}(\xi, \eta; 2\tau)$ where $\tau$ depends on $a - \sqrt{2n}$, under suitable assumption on the moments of $f_{ij}$ and $f_i$. We note that if $\alpha = 0$, $X$ can be simply expressed as $W + A$, and this universality result is proved in \cite{Erdos-Kruger-Schroder20}.

\paragraph{Similarity to the two-matrix model for the Ising model coupled with 2D gravity}

We note that the limiting correlation kernel at the multi-critical point in the Hermitian matrix model with external source is naturally associated with $\Psi^{(\gamma, \alpha, \rho, \sigma, \tau)}(\xi)$ whose large $\xi$-asymptotic behavior in \eqref{eq:RHP_Psi_infty} is given by (on $\compC_+$)
 \begin{equation}\Theta(\xi^2)=\diag\left( \rho \omega^j \xi^{\frac{8}{3}} + \sigma \omega^{2j} \xi^{\frac{4}{3}} + \tau \omega^{j} \xi^{\frac{2}{3}}\right)_{j=1}^3, \end{equation}
which involves only $\xi^{2k/3}$, $k=1,2,4$. The leading order of the derivative of the phase function, $5/3$, corresponds to the vanishing order of the limiting mean eigenvalue density \eqref{eq:higher_cusp} in the multi-critical regime. It is interesting that another $3 \times 3$ matrix-valued RH problem where the behavior at infinity involves exponential functions with exponent
\begin{equation}
  \diag\left(\frac{3}{7}\omega^{j-1}\xi^{\frac{7}{3}}+t_5\omega^{1-j}\xi^{\frac{5}{3}}+t_2\omega^{1-j}\xi^{\frac{2}{3}}+t_1\omega^{j-1}\xi^{\frac{1}{3}}\right)_{j=1}^3
\end{equation}
has been used quite recently in the study of the partition function of the two-matrix model with quartic interactions in the multi-critical regime where the limiting zero distribution of the associated polynomials vanishes at the speed $|x|^{4/3}$ near the origin \cite[Section 4]{Hayford24}. This two-matrix model is related to the Ising model coupled with 2D gravity \cite{Duits-Hayford-Lee25}.

\subsection{Organization of the paper}

We transform the random matrix model into biorthogonal ensembles in Section \ref{sec:bi_poly_multiple}, and then express the correlation kernel $K^{\ext}_{2n}$ in terms of solutions of Riemann-Hilbert problems in Section \ref{sec:RH}. The asymptotic analysis of the Riemann-Hilbert problems is carried out in Section \ref{sec:Deift-Zhou}, leading to the proof of Theorem \ref{thm:main} in Section \ref{sec:proof_main}. Theorem \ref{thm:RHP_unique_solvability} is proved in Section \ref{sec:solvability}, and Theorem \ref{prop:boussinesq} is subsequently proved in Section \ref{sec:boussinesq}. The proofs of Propositions \ref{pro:sol_Psi} and \ref{pro:sum_kernel_Psi} and Corollary \ref{cor:MB} are given in Section \ref{sec:remaining}.

\subsection*{Acknowledgements}

Dong Wang was partially supported by the National Natural Science Foundation of China under grant number 12271502, and the University of Chinese Academy of Sciences start-up grant 118900M043. Shuai-Xia Xu was partially supported by the National Natural Science Foundation of China under grant numbers 12431008, 12371257 and 11971492, and by the Guangdong Basic and Applied Basic Research Foundation (Grant No. 2022B1515020063). We thank Tom Claeys, Nathan Hayford, and Arno Kuijlaars for helpful discussions.

\section{Biorthogonal, polyorthogonal and multiple orthogonal polynomials} \label{sec:bi_poly_multiple}

\subsection{The generalized Muttalib-Borodin biorthogonal system}

We consider the biorthogonal system between polynomials $\{ p_n(x) = p^{W, \gamma}_n(x) \}$ and $\{ q_n(x^2)=q^{W, \gamma}_n(x^2) \}$, where $p_n$ and $q_n$ are degree $n$ monic polynomials defined by \eqref{eq:biorthogonal_def}. In general, this biorthogonal system may not be well-defined. However, in a special case, we have the well-definedness of the biorthogonal system:

\begin{prop} \label{prop:basic}
  Let $\What(x)$ be an even function that is non-negative and integrable on $\realR$. We further assume that $\What(x)$ vanishes fast at $\infty$ such that $\lim_{x \to \pm \infty} \lvert x \rvert^{-1} \log \What(x) = -\infty$. Under the following conditions
  \begin{enumerate}
  \item 
    $W(x) = \What(x) e^{\alta x}$ where $\alta > 0$, and
  \item
    $\lvert \gamma \rvert \leq 1$,
  \end{enumerate}
   the biorthogonal system given by \eqref{eq:biorthogonal_def} is well-defined, all $p_n$, $q_n$ are uniquely determined and $h^{W, \gamma}_n \neq 0$.
\end{prop}

\begin{proof}
  We will show that for each $n \geq 0$, there is a unique monic polynomial $p_n(x) = a_0 + a_1 x + \dotsb + a_n x^n$ with $a_n = 1$ such that
  \begin{equation} \label{eq:orth_cond_p_n}
    \int_{\realR} p_n(x) e^{\alta x} x^{2k} (1_{x \geq 0}(x) + \gamma \cdot 1_{x < 0}(x)) \What(x) dx = 0, \quad k = 0, 1, \dotsc, n - 1,
  \end{equation}
and a unique monic polynomial $q_n(x) = b_0 + b_1 x + \dotsb + b_n x^n$ with $b_n = 1$ such that 
  \begin{equation} \label{eq:orth_cond_q_n}
    \int_{\realR} x^k e^{\alta x} q_n(x^2) (1_{x \geq 0}(x) + \gamma \cdot 1_{x < 0}(x)) \What(x) dx = 0, \quad k = 0, 1, \dotsc, n - 1,
  \end{equation}
  and
  \begin{equation} \label{eq:orth_cond_p_n_q_n}
    \int_{\realR} p_n(x) e^{\alta x} q_n(x^2) (1_{x \geq 0}(x) + \gamma \cdot 1_{x < 0}(x)) \What(x) dx \neq 0.
  \end{equation}
  We note that
  \begin{multline}
    p_n(\sqrt{\xi}) e^{\alta \sqrt{\xi}} + \gamma p_n(-\sqrt{\xi}) e^{-\alta \sqrt{\xi}} = \\
    \left( \sum^{\lfloor n/2 \rfloor}_{j = 0} a_{2j} \xi^j \right) (e^{\alta \sqrt{\xi}} + \gamma e^{-\alta \sqrt{\xi}}) + \left( \sum^{\lfloor (n - 1)/2 \rfloor}_{j = 0} a_{2j + 1} \xi^j \right) \sqrt{\xi}(e^{\alta \sqrt{\xi}} - \gamma e^{-\alta \sqrt{\xi}}).
  \end{multline}
  After the change of variable $x^2 = \xi$, we find that the orthogonality condition \eqref{eq:orth_cond_p_n} becomes
  \begin{equation} \label{eq:p_orth_even}
    \int_{\realR_+} \left[ p_n(\sqrt{\xi}) e^{\alta \sqrt{\xi}} + \gamma p_n(-\sqrt{\xi}) e^{-\alta \sqrt{\xi}} \right] \xi^k \frac{\What(\sqrt{\xi})}{2 \sqrt{\xi}} d\xi = 0, \quad k = 0, 1, \dotsc, n - 1,
  \end{equation}
  the orthogonality condition \eqref{eq:orth_cond_q_n} becomes
  \begin{multline}
    \int_{\realR_+} \xi^j (e^{\alta \sqrt{\xi}} + \gamma e^{-\alta \sqrt{\xi}}) q_n(\xi) \frac{\What(\sqrt{\xi})}{2 \sqrt{\xi}} d\xi = 0, \quad \int_{\realR_+} \xi^k (e^{\alta \sqrt{\xi}} - \gamma e^{-\alta \sqrt{\xi}}) q_n(\xi) \frac{\What(\sqrt{\xi})}{2} d\xi = 0, 
  \end{multline}
  for $j = 0, 1, \dotsc, \left\lfloor \frac{n - 1}{2} \right\rfloor, \quad k = 0, 1, \dotsc, \left\lfloor \frac{n}{2} \right\rfloor - 1$, 
  and condition \eqref{eq:orth_cond_p_n_q_n} transforms into
  \begin{equation} \label{eq:p_orth_even_deg_n}
    \int_{\realR_+} \left[ p_n(\sqrt{\xi}) e^{\alta \sqrt{\xi}} + \gamma p_n(-\sqrt{\xi}) e^{-\alta \sqrt{\xi}} \right] q_n(\xi) \frac{\What(\sqrt{\xi})}{2 \sqrt{\xi}} d\xi \neq 0.
  \end{equation}

  We denote
  \begin{align} \label{eq:multiple_weights_u_i}
    u_1(\xi) = {}& e^{\alta \sqrt{\xi}} + \gamma e^{-\alta \sqrt{\xi}}, & u_2(\xi) = {}& \sqrt{\xi}(e^{\alta \sqrt{\xi}} - \gamma e^{-\alta \sqrt{\xi}}), & w_n(\xi) =
                                                                                                                                                            \begin{cases}
                                                                                                                                                              u_1(\xi) \xi^{\frac{n}{2}}, & n \text{ is even}, \\
                                                                                                                                                              u_2(\xi) \xi^{\frac{n - 1}{2}}, & n \text{ is odd}.
                                                                                                                                                            \end{cases}
  \end{align}
If we can show that $(u_1(\xi), u_2(\xi))$ is a weak AT-system on $(0, +\infty)$ in the sense of \cite[Section 4.4.4]{Nikishin-Sorokin91}, or equivalently, $(w_n(\xi))^{\infty}_{n = 0}$ is a Markov system on $(0, +\infty)$ in the sense of \cite[Section 4.4.2]{Nikishin-Sorokin91}, then for all $n \geq 0$, the determinants \cite[Section 4.4, Exercise 3]{Nikishin-Sorokin91}
\begin{equation}
  M_n := \det \left( m_{j, k} \right)^n_{j, k = 0} \neq 0, \quad \text{where} \quad m_{j, k} = \int_{\realR_+} w_j(\xi) \xi^k \frac{\What(\sqrt{\xi})}{2\sqrt{\xi}} d\xi.
\end{equation}
Then it is straightforward to check that (letting $M_{-1} = 1$)
\begin{align}
  p_n(\sqrt{\xi}) e^{\alta \sqrt{\xi}} + \gamma p_n(-\sqrt{\xi}) e^{-\alta \sqrt{\xi}} = {}& \frac{1}{M_{n - 1}} \det
                                                                                             \begin{pmatrix}
                                                                                               m_{0, 0} & m_{0, 1} & \dots & m_{0, n - 1} & w_0(\xi) \\
                                                                                               m_{1, 0} & m_{1, 1} & \dots & m_{1, n - 1} & w_1(\xi) \\
                                                                                               \vdots & \vdots & \dots & \vdots \\
                                                                                               m_{n, 0} &  m_{n, 1} & \dots &  m_{n, n - 1} &  w_n(\xi) \\
                                                                                             \end{pmatrix}, \\
  q_n(\xi) = {}& \frac{1}{M_{n - 1}} \det
                 \begin{pmatrix}
                   m_{0, 0} & m_{0, 1} & \dots & m_{0, n} \\
                   m_{1, 0} & m_{1, 1} & \dots & m_{1, n} \\
                   \vdots & \vdots & \dots & \vdots \\
                   m_{n - 1, 0} & m_{n - 1, 1} & \dots & m_{n - 1, n} \\
                   1 & \xi & \dots & \xi^n
                 \end{pmatrix},
\end{align}
and \eqref{eq:p_orth_even_deg_n} is also satisfied.
  
  To show that $(w_n(\xi))^{\infty}_{n = 0}$ is a Markov system on $(0, +\infty)$, we need to verify that for any $m + 1$ real numbers $\alpha_0, \alpha_1, \dotsc, \alpha_m$ that are not all zero, the linear combination
  \begin{equation}
    h_{\gamma}(\xi) := \alpha_0 w_0(\xi) + \alpha_1 w_1(\xi) + \dotsb + \alpha_m w_m(\xi)
  \end{equation}
  changes sign at no more than $m$ points on $(0, +\infty)$. Let
  \begin{equation}
    f(z) = \alpha_0 + \alpha_1 z + \alpha_2 z^2 + \dotsb + \alpha_m z^m.
  \end{equation}
  Suppose that the polynomial $f(z)$ has a zero of multiplicity $k$ at $0$ and zeros $\{ i y_j \}^l_{j = 1}$ on the positive imaginary axis $\{ iy : y > 0 \}$, counting multiplicities. (Here we allow $k = 0$ or $l = 0$.) We define
  \begin{equation}
    g(z) = f(z) z^{-k} \prod^l_{j = 1} (z^2 + y^2_j)^{-1},
  \end{equation}
  which is a polynomial of degree $m - k - 2l$. Then we have
  \begin{equation} \label{eq:h_in_f_g}
    h_{\gamma}(z^2) = \h_{\gamma}(z) z^k \prod^{l}_{j = 1} (z^2 + y^2_j), \quad \text{where} \quad \h_{\gamma}(z) = g(z) e^{\alta z} + (-1)^k \gamma g(-z) e^{-\alta z}.
  \end{equation}
  For $x\in(0, +\infty)$, $h_{\gamma}(x^2)$ changes sign if and only if $\tilde{h}_{\gamma}(x)$ does. Since $\h_{\gamma}(x)$ is continuous, it changes sign only at its zeros. However, not all zeros of $\h_{\gamma}(x)$ are sign-changing points, since some zeros may be ``permanance-knots'' of $\h_{\gamma}(x)$, in terms of \cite[Section 4.4]{Nikishin-Sorokin91}.

  We first consider the case $ \gamma \in(-1,1)$. Denote $\Omega = \{ z \in \compC : \lvert z \rvert < R \text{ and } \Re z > 0 \}$ for $R > 0$. If $R$ is large enough, we have $\lvert g(z) \rvert > 0$ for all $z \in \partial \Omega$, and the inequality
  \begin{equation}
    \lvert \gamma g(-z) e^{-\alta z} \rvert < \lvert g(z) e^{\alta z} \rvert \quad \text{on $\partial \Omega$}.
  \end{equation}
  Then by Rouch\'{e}'s theorem, in the region $\Omega$, the number of zeroes of $\h_{\gamma}(z) = g(z) e^{\alta z} +(-1)^k \gamma g(-z) e^{-\alta z}$ is equal to the number of zeroes of $g(z) e^{\alta z}$ there, which is no more than $m - k - 2l$. Especially, the number of zeroes of $\h_{\gamma}(x)$ on $(0, R)$ is no more than $m - k - 2l$. Letting $R \to +\infty$, we derive that the number of zeros of $\h_{\gamma}(x)$ on $(0, +\infty)$ is no more than $m - k - 2l$ by \eqref{eq:h_in_f_g}, and then conclude that $h_{\gamma}(x)$, as well as $h_{\gamma}(x^2)$, changes sign at no more than $m - k - 2l \leq m$ points there. Hence, we prove that $\{ w_n(\xi) \}^{\infty}_{n = 0}$ is a Markov system on $(0, +\infty)$ if $\gamma \in (-1, 1)$.

  We consider $\gamma = 1$ by contradiction. Suppose $\h_1(x) = g(x) e^{\alta x} + (-1)^k g(-x) e^{-\alta x}$ changes sign more than $m$ times on $(0, +\infty)$. Then there are $0 < x_0 < x_1 < \dotsb < x_{m + 1}$ such that $\h_1(x)$ has the property that it attains the same sign at $\{ x_{2i} \}_{0 \leq i \leq \lfloor (m + 1)/2 \rfloor}$, and the opposite sign at $\{ x_{2i + 1} \}_{0 \leq i \leq \lfloor m /2 \rfloor}$. Since $h_{\gamma}(x)$ depends on $\gamma$ continuously, if $\gamma \in (-1, 1)$ is very close to $1$, $\h_{\gamma}(x)$ has the same property, that is, $\h_{\gamma}(x)$ changes sign at least once in each interval $(x_i, x_{i + 1})$, $i = 0, 1, \dotsc, m$, which contradicts the results above that $\h_{\gamma}(x)$ has at most $m$ zeros on $(0, +\infty)$. In the same way, we have that $\h_{-1}(x)$ changes sign no more than $m$ times on $(0, +\infty)$. Hence, we also prove that $\{ w_n(\xi) \}^{\infty}_{n = 0}$ is a Markov system on $(0, +\infty)$ when $\gamma = \pm 1$.
\end{proof}

Throughout this section, we take $W$ to be defined by $\What$ and $\alta > 0$ as in Proposition \ref{prop:basic}. Here we recall the reproducing kernel $K^{W, \gamma}_m$ defined in \eqref{eq:kernel_MB} for the biorthogonal ensemble and $\Khat^{W, \gamma}_m$ defined in \eqref{eq:defn_Khat}. When $\gamma = 0$, we have that $\Khat^{W, 0}_m(x, y) = K^{W, 0}_m(x, y)$ for $x, y > 0$ and it is the correlation kernel of the Muttalib-Borodin ensemble.

\subsection{Polyorthogonal polynomials and multiple orthogonal polynomials} \label{subsec:poly_multiple}

We denote the multiple weights on $\realR_+$
\begin{align} \label{eq:defn_W_1_W_2}
  W_1(\xi) = {}& (e^{\alta \sqrt{\xi}} + \gamma e^{-\alta  \sqrt{\xi}}) \xi^{-\frac{1}{2}} \What(\sqrt{\xi}), & W_2(\xi) = {}& (e^{\alta \sqrt{\xi}} - \gamma e^{-\alta  \sqrt{\xi}}) \What(\sqrt{\xi}),
\end{align}
where $\alta$ and $\What$ are the same as in Proposition \ref{prop:basic}. Following \cite[The first definition in Section 4.3.4]{Nikishin-Sorokin91}, a non-zero polynomial $P_{k_1, k_2}(\xi)$ is called a polyorthogonal polynomial of type II with vector index $(k_1, k_2)$ with respect to the multiple weights $W_1(\xi)$ and $W_2(\xi)$, if
\begin{equation} \label{eq:polyorth_P}
  \int_{\realR_+} P_{k_1, k_2}(\xi) \xi^k W_j(x) dx = 0, \quad k = 0, 1, \dotsc, k_j - 1, \quad \text{for} \quad j = 1, 2,
\end{equation}
and $\deg P_{k_1, k_2}(\xi) \leq k_1 + k_2$. Following \cite[The second definition in Section 4.3.4]{Nikishin-Sorokin91}, if a non-trivial function $Q_{k_1, k_2}(\xi) = A_{k_1, k_2}(\xi) W_1(\xi) + B_{k_1, k_2}(\xi) W_2(\xi)$ satisfies that $A_{k_1, k_2}(\xi)$ is a polynomial with $\deg A_{k_1, k_2}(\xi) \leq k_1 - 1$, $B_{k_1, k_2}(\xi)$ is a polynomial with $\deg B_{k_1, k_2}(\xi) \leq k_2 - 1$, and
\begin{equation} \label{eq:inner_prod_Q}
  \int_{\realR_+} \xi^j Q_{k_1, k_2}(\xi) d\xi = 0, \quad j = 0, 1, \dotsc, k_1 + k_2 - 2,
\end{equation}
then $A_{k_1, k_2}(\xi)$ and $B_{k_1, k_2}(\xi)$ are called polyorthogonal polynomials of type I with respect to the multiple weights $W_1(\xi)$ and $W_2(\xi)$. \footnote{In the literature, the terms polyorthogonal polynomials and multiple orthogonal polynomials are usually used interchangeably. In our paper, we utilize these notions to distinguish the differences.}

We note that the polyorthogonal polynomials are not uniquely determined. Indeed, if $P_{k_1, k_2}(\xi)$ and $Q_{k_1, k_2}$ satisfy the definitions, then for any nonzero constant $c$, $c P_{k_1, k_2}(\xi)$ and $c Q_{k_1, k_2}$ also satisfy the definitions. In the spirit of \cite[Section 4]{Bleher-Kuijlaars04a}, if the polyorthogonal polynomial $P_{k_1, k_2}(\xi)$ additionally satisfies that $\deg P_{k_1, k_2} = k_1 + k_2$ and it is monic, then it is called a multiple orthogonal polynomial of type II. Similarly, if the polyorthogonal polynomials $A_{k_1, k_2}(\xi)$ and $B_{k_1, k_2}(\xi)$ satisfy
\begin{equation}
  \int_{\realR_+} \xi^{k_1 + k_2 - 1} Q_{k_1, k_2}(\xi) d\xi = 1,
\end{equation}
then $A_{k_1, k_2}(\xi)$ and $B_{k_1, k_2}(\xi)$ are called multiple orthogonal polynomials of type I.

Although we do not have the well-definedness of $P_{k_1, k_2}$ and $Q_{k_1, k_2}$ for all $(k_1, k_2)$, we have that for all integers $m$, the polynomials
\begin{align}
  P_{\lfloor \frac{m - 1}{2} \rfloor + 1, \lfloor \frac{m}{2} \rfloor}(\xi) = {}& q_m(\xi), \label{eq:multiple_P} \\
  Q_{\lfloor \frac{m - 1}{2} \rfloor + 1, \lfloor \frac{m}{2} \rfloor}(\xi) = {}& \frac{1}{2h_{m - 1}^{W, \gamma}} \left( p_{m - 1}(\sqrt{\xi}) e^{\alta  \sqrt{\xi}} + \gamma p_{m - 1}(-\sqrt{\xi}) e^{-\alta  \sqrt{\xi}} \right) \xi^{-\frac{1}{2}} \What(\sqrt{\xi}), \label{eq:multiple_Q}
\end{align}
are the unique multiple orthogonal polynomials with index $(\lfloor (m - 1)/2 \rfloor + 1, \lfloor m/2 \rfloor)$, where $q_m$ and $p_{m - 1}$ are the biorthogonal polynomials defined in \eqref{eq:biorthogonal_def} with respect to $W$ specified in Proposition \ref{prop:basic}. 
Hence, for the vector index $(k_1, k_2) = (k, k)$ or $(k + 1, k)$, multiple orthogonal polynomials $P_{k_1, k_2}$ and $Q_{k_1, k_2}$ are well-defined. For $(k_1, k_2) = (k, k + 1)$ or $(k + 2, k)$, we have polyorthogonal polynomials
\begin{align}
  P_{k, k + 1}(\xi) = {}& \langle q_{2k}(x^2), x^{2k + 1} \rangle q_{2k + 1}(\xi) - \langle q_{2k + 1}(x^2), x^{2k + 1} \rangle q_{2k}(\xi), \label{eq:defn_P_k,k+1} \\
  P_{k + 2, k}(\xi) = {}& \langle q_{2k + 1}(x^2), x^{2k + 2} \rangle q_{2k + 2}(\xi) - \langle q_{2k + 2}(x^2), x^{2k + 2} \rangle q_{2k + 1}(\xi), \label{eq:defn_P_k+2,k}
\end{align}
where $\langle \cdot , \cdot \rangle = \langle \cdot , \cdot \rangle_{W, \gamma}$ is the inner product defined in \eqref{eq:biorthogonal_def}. They may not be multiple orthogonal polynomials up to a constant multiple, since their degree may be less than $2k + 1$ and $2k + 2$, respectively. For $m = 2k + 1$ or $2k + 2$, let $p_m(\xi) = \xi^m + a^{(m)}_1 \xi^{m - 1} + \dotsb$, and define
\begin{equation} \label{eq:defn_tilde_p_m}
  \tilde{p}_m(\xi) = p_m(\xi) - a^{(m)}_1 p_{m - 1}(\xi).
\end{equation}
Then for $(k_1, k_2) = (k, k + 1)$ or $(k + 2, k)$, we have the polyorthogonal polynomials $A_{k_1, k_2}$ and $B_{k_1, k_2}$, in the form of $Q_{k_1, k_2}(\xi) = A_{k_1, k_2}(\xi) W_1(\xi) + B_{k_1, k_2}(\xi) W_2(\xi)$,
\begin{align}
  Q_{k, k + 1}(\xi) = {}& \left( \tilde{p}_{2k + 1}(\sqrt{\xi}) e^{\alta  \sqrt{\xi}} + \gamma \tilde{p}_{2k + 1}(-\sqrt{\xi}) e^{-\alta  \sqrt{\xi}} \right) \xi^{-\frac{1}{2}} \What(\sqrt{\xi}), \label{eq:defn_Q_k_k+1} \\
  Q_{k + 2, k}(\xi) = {}& \left( \tilde{p}_{2k + 2}(\sqrt{\xi}) e^{\alta  \sqrt{\xi}} + \gamma \tilde{p}_{2k + 2}(-\sqrt{\xi}) e^{-\alta  \sqrt{\xi}} \right) \xi^{-\frac{1}{2}} \What(\sqrt{\xi}). \label{eq:defn_Q_k+2_k}
\end{align}
They may not be multiple orthogonal polynomials up to a constant multiple, since the inner product between the right-hand side of \eqref{eq:defn_Q_k_k+1} (resp.~\eqref{eq:defn_Q_k+2_k}) with $\xi^{2k}$ (resp.~$\xi^{2k + 1}$) vanishes if $a_m = 0$ in the definition of $\tilde{p}_m$ for $m = 2k + 1$ (resp.~$m = 2k + 2$).

Later in this paper, we use \eqref{eq:multiple_P}, \eqref{eq:multiple_Q}, \eqref{eq:defn_P_k,k+1}, \eqref{eq:defn_P_k+2,k}, \eqref{eq:defn_Q_k_k+1} and \eqref{eq:defn_Q_k+2_k} as the definition of the polyorthogonal polynomials with indices $(k_1, k_2) = (k, k)$, $(k + 1, k)$, $(k, k + 1)$ and $(k + 2, k)$. We do not consider the well-definedness of the polyorthogonal polynomials for other indices.

Denote
\begin{align}
  \tilde{h}_{k_1, k_2} = {}& \int_{\realR_+} Q_{k_1, k_2}(\xi) \xi^{k_1 + k_2 - 1} d\xi, & h^{(j)}_{k_1, k_2} = {}& \int_{\realR_+} P_{k_1, k_2}(\xi) \xi^{k_j} W_j(\xi) d\xi, \quad j = 1, 2.
\end{align}
We have, by arguments above, with $h^{W, \gamma}_n$ defined in \eqref{eq:biorthogonal_def} and $W$ specified in Proposition \ref{prop:basic},
\begin{align} \label{eq:h^1_kk}
  h^{(1)}_{k, k} = {}& 2 h^{W, \gamma}_{2k}, & h^{(1)}_{k - 1, k} = {}&  -2 h^{W, \gamma}_{2k - 1} h^{W, \gamma}_{2k - 2}, & h^{(2)}_{k, k - 1} = {}& 2 h^{W, \gamma}_{2k - 1}, & h^{(2)}_{k + 1, k - 1} = {}& -2 h^{W, \gamma}_{2k} h^{W, \gamma}_{2k - 1}, 
\end{align}
\begin{equation} \label{eq:htilde_kk}
  \tilde{h}_{k, k} = \tilde{h}_{k + 1, k} = 1, 
\end{equation}
and so they are all non-zero. The quantities $h^{(1)}_{k + 1, k}$, $h^{(2)}_{k, k}$, $\tilde{h}_{k, k + 1}$ and $\tilde{h}_{k + 2, k}$ may vanish, but we have

\begin{equation}
  \begin{split}
    \frac{h^{(1)}_{k + 1, k}}{\tilde{h}_{k + 2, k}} = \frac{\langle q_{2k + 1}(x^2), x^{2k + 2} \rangle}{\langle \tilde{p}_{2k + 2}(x), x^{2(2k + 1)} \rangle} = {}& \frac{\langle q_{2k + 1}(x^2), p_{2k + 2}(x) \rangle - a^{(2k + 2)}_1\langle q_{2k + 1}(x^2), p_{2k + 1}(x) \rangle}{\langle p_{2k + 2}(x), x^{2(2k + 1)} \rangle - a^{(2k + 2)}_1 \langle p_{2k + 1}(x), x^{2(2k + 1)} \rangle} \\
    = {}& \frac{-a^{(2k + 2)}_1 h^{W, \gamma}_{2k + 1}}{-a^{(2k + 2)}_1 h^{W, \gamma}_{2k + 1}} = 1,
  \end{split}
\end{equation}
where the validity of the last identity holds for $a^{(2k + 2)}_1 = 0$ if we apply l'H\^{o}pital's rule. Similarly, we have
\begin{equation} \label{eq:tilde_h_kk}
  \frac{h^{(2)}_{k, k}}{\tilde{h}_{k, k + 1}} = 1.
\end{equation}

Like in the multiple orthogonal ensembles considered in \cite{Kuijlaars10a}, for $x, y > 0$, we consider the reproducing kernel of the polyorthogonal (multiple orthogonal) system. Upon using the expression \eqref{eq:multiple_P} and \eqref{eq:multiple_Q} of the multiple orthogonal polynomials, it turns out that the reproducing kernel is \cite[Equation (3.11)]{Bleher-Kuijlaars04a}
\begin{equation} \label{eq:K^mult_and_K^Wc}
  \sum^{m - 1}_{k = 0} P_{\lfloor \frac{k - 1}{2} \rfloor + 1, \lfloor \frac{k}{2} \rfloor}(y) Q_{\lfloor \frac{k}{2} \rfloor + 1, \lfloor \frac{k + 1}{2} \rfloor}(x) = \frac{1}{2\sqrt{x}} \Khat^{W, \gamma}_m(\sqrt{x}, \sqrt{y}),
\end{equation}
where $\Khat^{W, \gamma}$ is defined in \eqref{eq:defn_Khat}.

From a similar argument in \cite{Bleher-Kuijlaars04a}, we have the following Christoffel-Darboux formula:
\begin{equation} \label{eq:CDformula}  
  \begin{split}
    (y - x) \frac{1}{2\sqrt{x}} \Khat^{W, \gamma}_m(\sqrt{x}, \sqrt{y}) = {}& P_{\lfloor \frac{m - 1}{2} \rfloor + 1, \lfloor \frac{m}{2} \rfloor}(y) Q_{\lfloor \frac{m - 1}{2} \rfloor + 1, \lfloor \frac{m}{2} \rfloor}(x) \\
                                                                            & \hspace{-2cm} - \left( \frac{P_{\lfloor \frac{m - 1}{2} \rfloor, \lfloor \frac{m}{2} \rfloor}(y)}{h^{(1)}_{\lfloor \frac{m - 1}{2} \rfloor, \lfloor \frac{m}{2} \rfloor}} \right) 
                                                                              \times \left( \frac{h^{(1)}_{\lfloor \frac{m - 1}{2} \rfloor + 1, \lfloor \frac{m}{2} \rfloor}}{\tilde{h}_{\lfloor \frac{m - 1}{2} \rfloor + 2, \lfloor \frac{m}{2} \rfloor}} Q_{\lfloor \frac{m - 1}{2} \rfloor + 2, \lfloor \frac{m}{2} \rfloor}(x) \right) \\
                                                                            & \hspace{-2cm} - \left( \frac{P_{\lfloor \frac{m - 1}{2} \rfloor + 1, \lfloor \frac{m}{2} \rfloor - 1}(y)}{h^{(2)}_{\lfloor \frac{m - 1}{2} \rfloor + 1, \lfloor \frac{m}{2} \rfloor - 1}} \right) 
                                                                              \times \left( \frac{h^{(2)}_{\lfloor \frac{m - 1}{2} \rfloor + 1, \lfloor \frac{m}{2} \rfloor}}{\tilde{h}_{\lfloor \frac{m - 1}{2} \rfloor + 1, \lfloor \frac{m}{2} \rfloor + 1}} Q_{\lfloor \frac{m - 1}{2} \rfloor + 1, \lfloor \frac{m}{2} \rfloor + 1}(x) \right). 
  \end{split}
\end{equation}
Here we remark that all terms on the right-hand side of \eqref{eq:CDformula} are well-defined, due to the well-definedness of $P_{k_1, k_2}$ and $Q_{k_1, k_2}$ for $k_1 - k_2 = -1$, $0$, $1$ or $2$, and the non-vanishing property of quantities in \eqref{eq:h^1_kk}--\eqref{eq:tilde_h_kk}.

\subsection{Relation to Hermitian matrix model with external source} 

Denote the $2n$-dimensional diagonal matrix
\begin{equation} \label{eq:external_source}
  A = \diag(\underbrace{\alta, \dotsc, \alta}_{\text{$n$ entries}}, \underbrace{-\alta, \dotsc, -\alta}_{\text{$n$ entries}}), \quad \alta > 0.
\end{equation}
We consider the $2n$-dimensional random Hermitian matrix $M$ whose distribution is given by
\begin{equation} \label{eq:pdf_ext_source}
  \frac{1}{Z_n} \det(\What(M)) e^{\Tr(AM)}.
\end{equation}

By the general theory of Hermitian matrix model with external source (\cite[Equation (3.2)]{Zinn_Justin97} and \cite[Equation (3.4)]{Zinn_Justin98} for $A$ in \eqref{eq:pdf_ext_source} with distinct eigenvalues, and \cite[Equation (3.17)]{Bleher-Kuijlaars04a} for $A$ in \eqref{eq:pdf_ext_source} with eigenvalues with multiplicities), we have that the eigenvalue distribution of $M$ is a determinantal point process and also a biorthogonal ensemble in the sense of \cite[Equation (2.8)]{Borodin99}. Hence, the distribution of eigenvalues of $M$ is characterized by a correlation kernel $K^{\ext}_{2n}(x, y)$ that can be constructed as follows. Let $\{ \phi_i(x) \}^{2n}_{j = 1}$ be linearly independent functions whose linear span is $\Span \{ 1, x, \dotsc, x^{2n - 1} \}$, and $\{ \psi_k (x) \}^{2n}_{k = 1}$ be linearly independent functions whose linear span is $\Span \{ e^{\alta x}, x e^{\alta x}, \dotsc, x^{n - 1} e^{\alta x}, e^{-\alta x}, x e^{-\alta x}, \dotsc, x^{n - 1} e^{-\alta x} \}$. Let
\begin{equation} \label{eq:defn_M_and_inner}
  M = (m_{j, k})^{2n}_{j, k = 1}, \quad \text{where} \quad m_{j, k} = \langle \phi_j, \psi_k \rangle_{\ext}, \quad \text{and} \quad \langle f, g \rangle_{\ext} = \int_{\realR} f(x) g(x) \What(x) dx.
\end{equation}
Then one form of the correlation kernel is, as derived in \cite[Proposition 2.2]{Borodin99}
\begin{equation} \label{eq:kernel_for_K^ext}
  K^{\ext}_{2n}(x, y) = (\phi_1(y), \phi_2(y), \dotsc, \phi_{2n}(y)) (M^{-1})^{T} (\psi_1(x), \psi_2(x), \dotsc, \psi_{2n}(x))^T \What(x).
\end{equation}
We note that our correlation kernel $K^{\ext}_{2n}(x, y)$ is equivalent to $K_{2n}(y, x) (\What(x))^{1/2} (\What(y))^{-1/2}$, where $K_{2n}(x, y)$ is defined in \cite[Equation (4.3)]{Bleher-Kuijlaars04a} with $n_1 = n_2 = n$, $\mathfrak{a}_1 = \alta$, $\mathfrak{a}_2 = -\alta$, and $e^{-V(x)}$ there means $\What(x)$ in our paper. To simplify the formula \eqref{eq:kernel_for_K^ext}, we find a particular choice of $\{ \phi_i(x) \}^{2n}_{j = 1}$ and $\{ \psi_k(x) \}^{2n}_{k = 1}$ such that $\langle \phi_j, \psi_k \rangle_{\ext} = 0$ if $j \neq k$.

Since the weight function $\What(x)$ in the definition of $\langle f, g \rangle_{\ext}$ in \eqref{eq:defn_M_and_inner} is even, we restrict $\{ \phi_j(x) \}^n_{j = 1}$ to be in $\Span \{ 1, x^2, \dotsc, x^{2n - 2} \}$, and $\{ \psi_k(x)  \}^n_{k = 1}$ to be in $\Span(\{ x^{2k} \cosh \alta x : k = 0, 1, \dotsc, \lfloor (n - 1)/2 \rfloor \} \cup \{ x^{2k + 1} \sinh \alta x : k = 0, 1, \dotsc, \lfloor n/2 \rfloor - 1 \})$, such that they are all even functions. On the other hand, we restrict $\{ \phi_i(x) \}^{2n}_{j = n + 1}$ to be in $\Span \{ x, x^3, \dotsc, x^{2n - 1} \}$, and $\{ \psi_k \}^{2n}_{k = n + 1}$ to be in $\Span(\{ x^{2k} \sinh \alta x : k = 0, 1, \dotsc, \lfloor (n - 1)/2 \rfloor \} \cup \{ x^{2k + 1} \cosh \alta x : k = 0, 1, \dotsc, \lfloor n/2 \rfloor - 1 \})$, such that they are all odd functions.

Recall that $W(x)$ and $\What(x)$ are related by Proposition \ref{prop:basic}. Now we fix the notation
\begin{align} \label{eq:defn_Wodd}
  W(x) = {}& \What(x) e^{\alta x}, & W^{\odd}(x) = {}& \lvert x \rvert W(x).
\end{align}
Then by comparing the orthogonality satisfied by $\{ \phi_j \}$ and $\{\psi_k \}$, and the biorthogonal system defined in \eqref{eq:biorthogonal_def}, we have, for $k = 1, 2, \dotsc, n$, in terms of $p^{W, \gamma}_n$ and $q^{W, \gamma}_n$ in \eqref{eq:biorthogonal_def},
\begin{align}
  \phi_k(x) = {}& q^{W, 1}_{k - 1}(x^2), & \psi_k(x) = {}& p^{W, 1}_{k - 1}(x) e^{\alta x} + p^{W, 1}_{k - 1}(-x) e^{-\alta x}, \\
  \phi_{n + k}(x) = {}& x q^{W^{\odd}, -1}_{k - 1}(x^2), & \psi_{n + k}(x) = {}& p^{W^{\odd}, -1}_{k - 1}(x) e^{\alta x} - p^{W^{\odd}, -1}_{k - 1}(-x) e^{-\alta x}, \\
  m_{k, k} = {}& 2 h^{W, 1}_{k - 1}, & m_{n + k, n + k} = {}& 2 h^{W^{\odd}, -1}_{k - 1},
\end{align}
and all $m_{j, k} = 0$ if $j \neq k$. These formulas, together with kernels defined in \eqref{eq:kernel_MB} and \eqref{eq:defn_Khat}, imply the following proposition. 

\begin{prop} \label{prop:ext_in_biorth}
  The correlation kernel for the eigenvalues of the Hermitian matrix model with external source defined by \eqref{eq:pdf_ext_source} is given by
  \begin{equation}
    K^{\ext}_{2n}(x, y) = \frac{1}{2} \Khat^{W, 1}_n(x, y) + \frac{y}{2x} \Khat^{W^{\odd}, -1}_n(x, y),
  \end{equation}
  where $W$ and $\What$ are defined in \eqref{eq:defn_Wodd} and the kernel $\Khat$ is defined in \eqref{eq:defn_Khat}.
\end{prop}

\section{Riemann-Hilbert problems} \label{sec:RH}

In this section, we formulate vector Riemann-Hilbert (RH) problems for $p^{W, \gamma}_k$ and $ q^{W, \gamma}_k$ that are defined by the biorthogonality \eqref{eq:biorthogonal_def}, and matrix RH problems for $P_{k_1, k_2}$ and $Q_{k_1, k_2}$ that are defined by the polyorthogonality \eqref{eq:polyorth_P} and \eqref{eq:inner_prod_Q}. We specify $\What(x)$ in Proposition \ref{prop:basic} such that
\begin{equation} \label{eq:specified_What}
  \What(x) = \lvert x \rvert^{\alpha} e^{-n\Vhat(x)}, \quad \alpha > -1, \quad \text{$\Vhat(x)$ is analytic on $\realR$, and} \quad \Vhat(-x) = \Vhat(x),
\end{equation}
and then substitute $\alta$ in Proposition \ref{prop:basic} by $2na$, such that $W$, $W^{\odd}$, $W_1$ and $W_2$ in \eqref{eq:defn_Wodd} and \eqref{eq:defn_W_1_W_2} are specified as
\begin{align}
  W(x) = {}& \What(x) e^{2nax}, & W^{\odd}(x) = {}& \lvert x \rvert W(x), \label{eq:specified_W_Wodd} \\
  W_1(x) = {}& (e^{2na \sqrt{x}} + \gamma e^{2na \sqrt{x}}) x^{-\frac{1}{2}} \What(\sqrt{x}), & W_1(x) = {}& (e^{2na \sqrt{x}} + \gamma e^{2na \sqrt{x}}) x^{-\frac{1}{2}} \What(\sqrt{x}).  \label{eq:specified_W1_W2}
\end{align}
In this section, we use $p_k, q_k, h_k$ to mean $p^{W, \gamma}_k, q^{W, \gamma}_k, h^{W, \gamma}_k$, and use $\langle f, g \rangle$ to mean $\langle f, g \rangle_{W, \gamma}$ with fixed $W$ and $\gamma \in [-1, 1]$, unless otherwise stated. Below, we denote the right-half plane by $\halfH = \{z:  \arg z \in (-\pi/2, \pi/2) \}$.

\subsection{Vector RH problems for $p_k$ and $q_k$}

In this subsection we show that $p_k$ and $q_k$ satisfy vector RH problems. The idea goes back to the RH problems satisfied by orthogonal polynomials \cite{Fokas-Its-Kitaev91}, \cite{Fokas-Its-Kitaev92}. If $\gamma = 0$, the biorthogonal polynomials $p_k$ and $q_k$ are Muttalib-Borodin biorthogonal polynomials, and in \cite[Theorems 1.3 and 1.5]{Claeys-Romano14} it is shown that they satisfy vector RH problems; see also \cite{Wang-Zhang21}, \cite{Wang-Xu25}.

With $W_1$ and $W_2$ specified in \eqref{eq:specified_W1_W2}, we consider two vector RH problems as follows:

\begin{RHP} \label{RHP:defn_Y}
  $Y^{(k)} = (Y^{(k)}_1, Y^{(k)}_2)$ is analytic on $(\compC, \halfH \setminus \realR_+)$, such that
  \begin{enumerate}
  \item \label{enu:RHP:defn_Y_1}
    As $z \to \infty$ in $\compC$, $Y^{(k)}_1(z) = z^k + \bigO(z^{k - 1})$.
  \item \label{enu:RHP:defn_Y_2} 
    As $z \to \infty$ in $\halfH$, $Y^{(k)}_2(z) = \bigO(z^{-2(k + 1)})$.
  \item \label{enu:RHP:defn_Y_3}
    As $z \to 0$ in $\compC$, $Y^{(k)}_1(z) = \bigO(1)$, and as $z \to 0$ in $\halfH$,
    \begin{equation} \label{eq:bc_Y_general_c}
      Y^{(k)}_2(z) =
      \begin{cases}
        \bigO(1), & \alpha > 1, \\
        \bigO(\log z), & \alpha = 1, \\
        \bigO(z^{\alpha - 1}), & -1 < \alpha < 1.
      \end{cases}
    \end{equation}
  \item \label{enu:RHP:defn_Y_4}
    On $\realR_+$,
    \begin{equation}
      Y^{(k)}_{2, +}(x) - Y^{(k)}_{2, -}(x) = \frac{\lvert x \rvert^{\alpha}}{2x} \left( Y^{(k)}_1(x) e^{n(2ax - \Vhat(x))} + \gamma Y^{(k)}_1(-x) e^{n(-2ax - \Vhat(x))} \right).
    \end{equation}
  \item \label{enu:RHP:defn_Y_5}
    We have the boundary condition for $ x\in\realR_+$ 
    \begin{equation}
      Y^{(k)}_{2}(ix) =Y^{(k)}_{2}(-ix).
    \end{equation}
  \end{enumerate}
\end{RHP}

\begin{rmk} \label{rmk:c=-1_special}
  When $\gamma = -1$, the factor $(p_k(x) e^{2nax} + \gamma p_k(-x) e^{-2nax})$ has a zero at $x = 0$, and then the boundary condition \eqref{eq:bc_Y_general_c} is strengthened to
  \begin{equation}
    Y^{(k)}_2(z) =
    \begin{cases}
      \bigO(1), & \alpha > 0, \\
      \bigO(\log z), & \alpha = 0, \\
      \bigO(z^{\alpha}), & -1 < \alpha < 0.
    \end{cases}
  \end{equation}
\end{rmk}

\begin{RHP} \label{RHP:Ytilde}
  $\Y^{(k)} = (\Y^{(k)}_1, \Y^{(k)}_2)$ is analytic on $(\halfH, \compC \setminus \realR)$, such that
  \begin{enumerate}
  \item
    As $z \to \infty$ in $\halfH$, $\Y^{(k)}_1(z) = z^{2k} + \bigO(z^{2(k - 1)})$.
  \item 
    As $z \to \infty$ in $\compC$, $\Y^{(k)}_2(z) = \bigO(z^{-(k + 1)})$.
  \item
    As $z \to 0$ in $\halfH$, $\Y^{(k)}_1(z) = \bigO(1)$, and as $z \to 0$ in $\compC$,
    \begin{equation}
      \Y^{(k)}_2(z) =
      \begin{cases}
        \bigO(1), & \alpha > 0, \\
        \bigO(\log z), & \alpha = 0, \\
        \bigO(z^{\alpha}), & -1 < \alpha < 0.
      \end{cases}
    \end{equation}
  \item
    On $\realR_+$,
    \begin{equation}
      \Y^{(k)}_{2, +}(x) - \Y^{(k)}_{2, -}(x) = \Y^{(k)}_1(x) \lvert x \rvert^{\alpha} e^{n(2ax - \Vhat(x))},
    \end{equation}
    and on $\realR_-$, with $x > 0$,
    \begin{equation}
      \Y^{(k)}_{2, +}(-x) - \Y^{(k)}_{2, -}(-x) = \gamma \Y^{(k)}_1(x) \lvert x \rvert^{\alpha} e^{n(-2ax - \Vhat(x))}.
    \end{equation}
  \item
    We have the boundary condition for $ x\in\realR_+$ 
    \begin{equation}
      \Y^{(k)}_{1}(ix) =\Y^{(k)}_{1}(-ix).
    \end{equation}
  \end{enumerate}
\end{RHP}

To state the solutions of RH problems \ref{RHP:defn_Y} and \ref{RHP:Ytilde}, we define for $z\in \halfH \setminus \realR_+$ and a polynomial function $f$
\begin{equation} \label{eq:defn_Cp_k}
  \C f(z) = \frac{1}{2\pi i} \int_{\realR} \frac{1}{x^2 - z^2} f(x) (1_{x \geq 0}(x) + \gamma \cdot 1_{x < 0}(x)) \lvert x \rvert^{\alpha} e^{n(2ax - \Vhat(x))}  dx.
\end{equation}
Then, we have 
\begin{equation}
  \begin{split}
  \C f(z) = {}& \frac{1}{2\pi i} \int_{\realR} \frac{1}{x^2 - z^2} f(x) (1_{x \geq 0}(x) + \gamma \cdot 1_{x < 0}(x)) \lvert x \rvert^{\alpha} e^{n(2ax - \Vhat(x))}  dx \\
  = {}& \frac{1}{2\pi i} \int_{0}^{\infty} \frac{1}{x^2 - z^2} f(x) x^{\alpha} e^{n(2ax - \Vhat(x))}  dx + \frac{\gamma}{2\pi i} \int_{0}^{\infty} \frac{1}{x^2 - z^2} f(-x) x^{\alpha} e^{-n(2ax + \Vhat(x))}  dx \\
  = {}& \frac{1}{2\pi i} \int_{0}^{\infty} \frac{1}{x^2 - z^2} x^{\alpha} (f(x) e^{2nax} + \gamma f(-x) e^{-2nax}) e^{-n\Vhat(x)} dx.
  \end{split}
\end{equation}
Hence, on $\realR_+$, we have
\begin{equation} \label{eq:jump_Cp_k}
  \C f_+(x) - \C f_-(x) = \frac{\lvert x \rvert^{\alpha}}{2x} \left( f(x) e^{n(2ax - \Vhat(x))} + \gamma f(-x) e^{n(-2ax - \Vhat(x))} \right).
\end{equation}

We also define for $z \in \compC \setminus \realR$ and a polynomial function $f$
\begin{equation}
  C f(z) = \frac{1}{2\pi i} \int_{\realR} \frac{1}{x - z} f(x^2) (1_{x \geq 0}(x) + \gamma \cdot 1_{x < 0}(x)) \lvert x \rvert^{\alpha} e^{n(2ax - \Vhat(x))} dx.
\end{equation}
Hence, on $\realR$, we have
\begin{equation}
  C f_+(x) - C f_-(x) = f(x^2) (1_{x \geq 0}(x) + \gamma \cdot 1_{x < 0}(x)) \lvert x \rvert^{\alpha} e^{n(2ax - \Vhat(x))}.
\end{equation}
We have the following result:
\begin{prop}
  Let $p_k(z) = p^{W, \gamma}_k(z)$ and $q_k(z) = q^{W, \gamma}_k(z)$ be the biorthogonal polynomials defined in \eqref{eq:biorthogonal_def} with $W(x)$ specified in \eqref{eq:specified_What} and \eqref{eq:specified_W_Wodd}. Then we have that $(p_k(z), \C p_k(z))$ is the unique solution of RH problem \ref{RHP:defn_Y}, and $(q_k(z^2), C q_k(z))$ is the unique solution of RH problem \ref{RHP:Ytilde}.
\end{prop}

\begin{proof}
  The proof of this proposition is essentially the same as the proof of \cite[Theorems 1.3 and 1.5]{Claeys-Romano14}. We consider RH problem \ref{RHP:defn_Y}; RH problem \ref{RHP:Ytilde} is analogous.

  First we verify that $(p_k(z), \C p_k(z))$ is a solution of RH problem \ref{RHP:defn_Y}. Item \ref{enu:RHP:defn_Y_1} is clear, Items \ref{enu:RHP:defn_Y_3} and \ref{enu:RHP:defn_Y_5} follow from the definition \eqref{eq:defn_Cp_k} of $\C p_k(z)$, and Item \ref{enu:RHP:defn_Y_4} follows from \eqref{eq:jump_Cp_k}. Finally, Item \ref{enu:RHP:defn_Y_2} follows from the orthogonality \eqref{eq:biorthogonal_def} satisfied by $p_k(z)$.

  Next we show that RH problem \ref{RHP:defn_Y} has at most one solution. We note that $Y_1(z)$ is an entire function, and Item \ref{enu:RHP:defn_Y_1} entails that $Y_1$ is a monic polynomial of degree $k$. Then Items \ref{enu:RHP:defn_Y_3}, \ref{enu:RHP:defn_Y_4} and \ref{enu:RHP:defn_Y_5} imply that $Y_2(z)$ is related to $Y_1(z)$ as $Y_2(z) = \C Y_1(z)$, where $\C$ is the transform defined in \eqref{eq:defn_Cp_k}. Finally, Item \ref{enu:RHP:defn_Y_2} is equivalent to $Y_1(z)$ being orthogonal to $1, z^2, z^4, \dotsc, z^{2k - 2}$, where the orthogonality is defined in \eqref{eq:biorthogonal_def}. Since it is proved in Proposition \ref{prop:basic} that the biorthogonal polynomials are unique, we conclude that $Y_1(z) = p_k(z)$, and then $Y_2(z) = \C p_k(z)$.
\end{proof}

\subsection{Matrix RH problems for $P_{k_1, k_2}$ and $Q_{k_1, k_2}$}

In this subsection we show that $P_{k_1, k_2}$ and $Q_{k_1, k_2}$ satisfy matrix RH problems. The RH problems of multiple orthogonal polynomials have been extensively used in asymptotic analysis; see \cite{Geronimo-Kuijlaars-Van_Assche00}, \cite{Kuijlaars10} and \cite{Kuijlaars10a} for example.

With $W_1$ and $W_2$ specified in \eqref{eq:specified_W1_W2}, we consider two matrix RH problems as follows:
\begin{RHP} \label{RH:X_tilde}
  $\X^{(m)} = (\X^{(m)}_{ij})^2_{i, j = 0}$ is analytic on $\compC \setminus \realR_+$, such that
  \begin{enumerate}
  \item
    As $z \to \infty$,
    \begin{equation}
      \X^{(m)}(z) = (I + \bigO(z^{-1})) \times
      \begin{cases}
        \diag(z^{2k}, z^{-k}, z^{-k}), & m = 2k, \\
        \diag(z^{2k + 1}, z^{-k - 1}, z^{-k}), & m = 2k + 1.
      \end{cases}
    \end{equation}
    As $z \to 0$,
    \begin{align}
      \X^{(m)}_{i, 0}(z) = {}& \bigO(1), && \notag \\
      \X^{(m)}_{i, 1}(z) = {}&
                                                                    \begin{cases}
                                                                      \bigO(1), & \alpha > 1, \\
                                                                      \bigO(\log z), & \alpha = 1, \\
                                                                      \bigO(z^{\frac{\alpha - 1}{2}}), & -1 < \alpha < 1,
                                                                    \end{cases}
      & \X^{(m)}_{i, 2}(z) = {}&
                                 \begin{cases}
                                   \bigO(1), & \alpha > 0, \\
                                   \bigO(\log z), & \alpha = 0, \\
                                   \bigO(z^{\frac{\alpha}{2}}), & -1 < \alpha < 0.
                                 \end{cases}
    \end{align}
  \item
    On $\realR_+$,
    \begin{equation}
      \X^{(m)}_+(x) = \X^{(m)}_-(x) J_{\X}(x), \quad J_{\X}(x) =
      \begin{pmatrix}
        1 & W_1(x) & W_2(x) \\
        0 & 1 & 0 \\
        0 & 0 & 1
      \end{pmatrix}.
    \end{equation}
  \end{enumerate}
\end{RHP}

\begin{RHP} \label{RH:X}
  $X^{(m)} = (X^{(m)}_{ij})^2_{i, j = 0}$ is analytic on $\compC \setminus \realR_+$, such that
  \begin{enumerate}
  \item
    As $z \to \infty$,
    \begin{equation}
      X^{(m)}(z) = (I + \bigO(z^{-1})) \times
      \begin{cases}
        \diag(z^{-2k}, z^k, z^k), & m = 2k, \\
        \diag(z^{-2k - 1}, z^{k + 1}, z^k), & m = 2k + 1.
      \end{cases}
    \end{equation}
    As $z \to 0$,
    \begin{align}
      X^{(m)}_{i, 0}(z) = {}&
                              \begin{cases}
                                \bigO(1), & \alpha > 1, \\
                                \bigO(\log z), & \alpha = 1, \\
                                \bigO(z^{\frac{\alpha - 1}{2}}), & -1 < \alpha < 1,
                              \end{cases}
      & X^{(m)}_{i, 1}(z) = {}& \bigO(1), & X^{(m)}_{i, 2}(z) = {}& \bigO(1),
    \end{align}
  \item
    On $\realR_+$,
    \begin{equation}
      X^{(m)}_+(x) = X^{(m)}_-(x) J_X(x), \quad J_X(x) =
      \begin{pmatrix}
        1 & 0 & 0 \\
        -W_1(x) & 1 & 0 \\
        -W_2(x) & 0 & 1
      \end{pmatrix}.
    \end{equation}
  \end{enumerate}
\end{RHP}

To state the solutions of RH problems \ref{RH:X_tilde} and \ref{RH:X}, we define for any polynomial $f$, the transforms
\begin{align}
  \C_1 f(z) = {}& \frac{1}{2\pi i} \int_{\realR_+} \frac{1}{x - z} f(x) W_1(x) dx, & \C_2 f(z) = {}& \frac{1}{2\pi i} \int_{\realR_+} \frac{1}{x - z} f(x) W_2(x) dx,
\end{align}
and for any integrable function $g$, the transform
\begin{equation}
  C_0 g(z) = \frac{1}{2\pi i} \int_{\realR_+} \frac{1}{x - z} g(x) dx.
\end{equation}
Then $C_0 g(z)$ and $\C_i f(z)$ ($i = 1, 2$) are defined in $\compC \setminus \realR_+$, and on $\realR_+$,
\begin{align}
  C_0 g_+(x) - C_0 g_-(x) = {}& g(x), & \C_i f_+(x) - \C_i f_-(x) = {}& f(x) W_i(x), \quad i = 1, 2.
\end{align}

We have the following result:
\begin{prop} \label{prop:matrix_RHP_solution}
  Let $P_m$ and $Q_m$ be defined in Section \ref{subsec:poly_multiple} with $W_1$ and $W_2$ specified in \eqref{eq:specified_W1_W2}. The unique solution of RH problem \ref{RH:X_tilde} is
  \begin{multline} \label{eq:constr_Xtilde}
    \X^{(m)}(z) = (\X^{(m)}_{i, j}(z))^2_{i, j = 0} = (\C^{(m)})^{-1} \\
    \times
    \begin{pmatrix}
      P_{\lfloor \frac{m - 1}{2} \rfloor + 1, \lfloor \frac{m}{2} \rfloor}(z) & \C_1 P_{\lfloor \frac{m - 1}{2} \rfloor + 1, \lfloor \frac{m}{2} \rfloor}(z) & \C_2 P_{\lfloor \frac{m - 1}{2} \rfloor + 1, \lfloor \frac{m}{2} \rfloor}(z) \\
      P_{\lfloor \frac{m - 1}{2} \rfloor, \lfloor \frac{m}{2} \rfloor}(z) & \C_1 P_{\lfloor \frac{m - 1}{2} \rfloor, \lfloor \frac{m}{2} \rfloor}(z) & \C_2 P_{\lfloor \frac{m - 1}{2} \rfloor, \lfloor \frac{m}{2} \rfloor}(z) \\
      P_{\lfloor \frac{m - 1}{2} \rfloor + 1, \lfloor \frac{m}{2} \rfloor - 1}(z) &    \C_1 P_{\lfloor \frac{m - 1}{2} \rfloor + 1, \lfloor \frac{m}{2} \rfloor - 1}(z) & \C_2 P_{\lfloor \frac{m - 1}{2} \rfloor + 1, \lfloor \frac{m}{2} \rfloor - 1}(z)
    \end{pmatrix},
  \end{multline}
  where $\C^{(m)}$ is a constant diagonal matrix
  \begin{equation}
    \C^{(m)} = \diag \left( 1, -\frac{1}{2\pi i} h_{\lfloor \frac{m - 1}{2} \rfloor, \lfloor \frac{m}{2} \rfloor}^{(1)}, -\frac{1}{2\pi i} h_{\lfloor \frac{m - 1}{2} \rfloor + 1, \lfloor \frac{m}{2} \rfloor - 1}^{(2)} \right),
  \end{equation}
  and the unique solution of RH problem \ref{RH:X} is 
  \begin{multline} \label{eq:defn_X^(m)}
    X^{(m)}(z) = (X^{(m)}_{i, j}(z))^2_{i, j = 0} = (C^{(m)})^{-1} \\
    \times
    \begin{pmatrix}
      -C_0 Q_{\lfloor \frac{m - 1}{2} \rfloor + 1, \lfloor \frac{m}{2} \rfloor}(z) &  A_{\lfloor \frac{m - 1}{2} \rfloor + 1, \lfloor \frac{m}{2} \rfloor}(z) &   B_{\lfloor \frac{m - 1}{2} \rfloor + 1, \lfloor \frac{m}{2} \rfloor}(z) \\
      - C_0 Q_{\lfloor \frac{m - 1}{2} \rfloor + 2, \lfloor \frac{m}{2} \rfloor}(z) & A_{\lfloor \frac{m - 1}{2} \rfloor + 2, \lfloor \frac{m}{2} \rfloor}(z) & B_{\lfloor \frac{m - 1}{2} \rfloor + 2, \lfloor \frac{m}{2} \rfloor}(z) \\
      -  C_0 Q_{\lfloor \frac{m - 1}{2} \rfloor + 1, \lfloor \frac{m}{2} \rfloor + 1}(z) & A_{\lfloor \frac{m - 1}{2} \rfloor + 1, \lfloor \frac{m}{2} \rfloor + 1}(z) & B_{\lfloor \frac{m - 1}{2} \rfloor + 1, \lfloor \frac{m}{2} \rfloor + 1}(z)
    \end{pmatrix},
  \end{multline}
  where $C^{(m)}$ is a constant diagonal matrix
    \begin{align}
    C^{(m)} = \diag \left( \frac{1}{2\pi i},     C^{(m)}_{22}, C^{(m)}_{33}\right),
  \end{align}
  with $ C^{(m)}_{22}$ and $ C^{(m)}_{33}$ given respectively by the leading coefficients of $A_{\lfloor \frac{m - 1}{2} \rfloor + 2, \lfloor \frac{m}{2} \rfloor}$ and $B_{\lfloor \frac{m - 1}{2} \rfloor + 1, \lfloor \frac{m}{2} \rfloor + 1}$. 
\end{prop}

Before giving the proof of Proposition \ref{prop:matrix_RHP_solution}, we remark that all the diagonal entries of $\C^{(m)}$ and $C^{(m)}$ are nonzero, due to \eqref{eq:h^1_kk} and \eqref{eq:htilde_kk}:
\begin{align}
  \C^{(2k)}_{22} = {}& h^{(1)}_{k - 1, k}, &
  \C^{(2k)}_{33} = {}& h^{(2)}_{k, k - 1}, &
  \C^{(2k + 1)}_{22} = {}& h^{(1)}_{k, k}, & 
  \C^{(2k + 1)}_{33} = {}& h^{(2)}_{k + 1, k - 1}, \\
  C^{(2k)}_{22} = {}& \frac{1}{2h^{W, \gamma}_{2k}}, &
  C^{(2k)}_{33} = {}& 1, &
  C^{(2k + 1)}_{33} = {}& \frac{1}{2h^{W, \gamma}_{2k + 1}}, &
  C^{(2k + 1)}_{22} = {}& 1.
\end{align}
\begin{proof}[Proof of Proposition \ref{prop:matrix_RHP_solution}]
  RH problem \ref{RH:X} is equivalent to the RH problem in \cite[Theorem 2.1]{Geronimo-Kuijlaars-Van_Assche00} for type I multiple orthogonal polynomials and RH problem \ref{RH:X_tilde} is equivalent to the RH problem in \cite[Theorem 3.1]{Geronimo-Kuijlaars-Van_Assche00} for type II multiple orthogonal polynomials, such that the multiple weights there are $W_1$ and $W_2$, and the index there is $(n_1, n_2) = (\lfloor (m - 1)/2 \rfloor + 1, \lfloor m/2 \rfloor)$, which is $(k, k)$ if $m = 2k$ or $(k + 1, k)$ if $m = 2k+1$. Hence, Proposition \ref{prop:matrix_RHP_solution} is a consequence of \cite[Theorems 2.1 and 3.1]{Geronimo-Kuijlaars-Van_Assche00}, as long as we can show that the indices $(n_1, n_2)$, $(n_1 + 1, n_2)$, $(n_1, n_2 + 1)$, $(n_1 - 1, n_2)$ and $(n_1, n_2 - 1)$ are all ``normal'' in the sense of \cite[Section 1.1]{Geronimo-Kuijlaars-Van_Assche00}, or equivalently, $P_{k_1, k_2}$ and $Q_{k_1, k_2}$ are determined up to a constant multiple for $(k_1, k_2)$ to be any one of the five indices aforementioned, or equivalently, for $(k_1, k_2)$ such that $k_1 - k_2 = -1$, $0$, $1$ or $2$. The well-definedness of such $P_{k_1, k_2}$ and $Q_{k_1, k_2}$ is given in \eqref{eq:multiple_P}--\eqref{eq:defn_Q_k+2_k}.
\end{proof}

We note that from the relation between RH problems \ref{RH:X_tilde} and \ref{RH:X} and their unique solvability, we find, analogous to \cite[Theorem 4.1]{Geronimo-Kuijlaars-Van_Assche00},
\begin{equation} \label{eq:relation_X_Xtilde}
  X^{(m)}(z)^T = \X^{(m)}(z)^{-1}.
\end{equation}
Then, from \eqref{eq:CDformula}, the expressions for $\X^{(m)}(z)$ and $X^{(m)}(z)$ and \eqref{eq:relation_X_Xtilde}, we have 
\begin{equation}\label{eq: kernel_K_mult}
  \begin{split}
    \frac{1}{2\sqrt{x}} \Khat^{W, \gamma}_m(\sqrt{x}, \sqrt{y}) = {}& \frac{1}{2\pi i(y - x)}
                                                                      \begin{pmatrix}
                                                                        1 & 0 & 0
                                                                      \end{pmatrix} 
                                                                      (\X^{(m)})^T(y)  ((\X^{(m)})^T(x))^{-1}
                                                                      \begin{pmatrix}
                                                                        0 \\
                                                                        W_1(x) \\
                                                                        W_2(x)
                                                                      \end{pmatrix} \\
    = {}& \frac{1}{2\pi i(y - x)}
          \begin{pmatrix}
            1 & 0 & 0
          \end{pmatrix} 
          (X^{(m)})^{-1}(y) X^{(m)}(x)
          \begin{pmatrix}
            0 \\
            W_1(x) \\
            W_2(x)
          \end{pmatrix}.
  \end{split}
\end{equation}
\section{Asymptotic analysis of $Y^{(n + k)}$ with quartic weight function} \label{sec:Deift-Zhou}

In this section, we derive the double scaling limits of the correlation kernel by studying the asymptotics of the vector-valued Riemann-Hilbert problem for $Y^{(n + k)}$; see \cite{Wang-Zhang21} for a similar approach for the hard edge universality of the Muttalib-Borodin ensemble. In the Pearcey regime and the multi-critical regime considered in this section, we focus on the construction of the local parametrix near the origin by using the solution of the new model RH problem for $\Phi^{(\gamma, \alpha, \rho, \sigma, \tau)}(\xi)$, which is associated with the Boussinesq hierarchy.

For this purpose, we further specialize $W(x)$ and $\What(x)$ in \eqref{eq:specified_What} and \eqref{eq:specified_W_Wodd} by letting $\Vhat(x) = x^4/2 - t x^2$. Hence, we express $W(x)$ as
\begin{align} \label{eq:quartic_V}
 W(x) = {}& \lvert x \rvert^{\alpha} e^{-nV(x)}, & V(x) = \Vhat(x) - 2ax  = {}& \frac{x^4}{2} - t x^2 - 2ax, \quad \alpha > -1.
\end{align}
We consider two regimes: the Pearcey (critical) regime and the multi-critical regime. In the Pearcey regime, we parametrize the coefficient $t$ and $a$ by \eqref{eq:def_t_c_Pearcey} with parameter $\tau \in \realR$,
 and in the multi-critical regime, we parametrize the coefficient $t$ and $a$ by \eqref{eq:def_t_c_critical} and \eqref{eq:c_multi_crit}, with parameters $\tau, \sigma \in \realR$.
From the formulas \eqref{eq:def_t_c_Pearcey} and \eqref{eq:def_t_c_critical}, we have that for $t, a$ in both regimes,
\begin{equation} \label{eq:c_identity}
  10c^4 - 2tc^2 - ac = 1,
\end{equation}
and
\begin{align}
  2c^4 + ac - 1 =
  \begin{cases}
    \frac{4}{3} \sqrt{\frac{1 - 3c^4}{2n}} \tau, & \text{$t,a$ in the Pearcey regime}, \\
    \frac{2}{3} n^{-\frac{3}{4}} \tau, &  \text{$t, a$ in the multi-critical regime}.
  \end{cases}
\end{align}

\subsection{$\gfn$-functions} \label{subsec:G_and_tildeG}

In both the Pearcey regime and the multi-critical regime, we let
\begin{align} \label{eq:defn_J_c}
  J_c(s) = {}& c (s + 1)^{\frac{3}{2}} s^{-\frac{1}{2}}, && \text{and} & b = {}& 3^{\frac{3}{2}} 2^{-\frac{1}{2}} c,
\end{align}
where $J_c(s)$ is defined in \cite[Equation (1.25)]{Claeys-Romano14} and $b$ is defined in \cite[Equation (4.1)]{Claeys-Romano14}, both with $\theta = 2$. Here we recall the results in \cite[Section 4.1]{Claeys-Romano14} and set up the notation to be used later. Let $\gamma_1 \subseteq \compC_+$ be the curve connecting $-1$ to $1/2$ such that $J_c$ maps $\gamma_1$ to the interval $[0, b]$. Then let $\gamma_2 = \overline{\gamma_1} \subseteq \compC_-$ and $\gamma = \gamma_1 \cup \gamma_2$ be the closed contour oriented counterclockwise and $D$ be the region enclosed by $\gamma$. Recall that $J_c$ maps $\compC \setminus \overline{D}$ to $\compC \setminus [0, b]$ and $D \setminus [-1, 0]$ to $\halfH \setminus [0, b]$ univalently. Then we define $I_1(z)$ and $I_2(z)$ to be the inverse of $J_c(z)$ on $\compC \setminus \overline{D}$ and $D \setminus [-1, 0]$ respectively. 

With the potential function $V(x)$ given in \eqref{eq:quartic_V}, we have the auxiliary function
\begin{equation}
  U_c(s) = V'(J_c(s)) J_c(s) = 2c^4 \frac{(s + 1)^6}{s^2} - 2t c^2 \frac{(s + 1)^3}{s} - 2ac \frac{(s + 1)^{\frac{3}{2}}}{s^{\frac{1}{2}}}
\end{equation}
as in \cite[Lemma 4.2]{Claeys-Romano14}. Then, by \cite[Equation (4.12)]{Claeys-Romano14}, we have
\begin{equation} \label{eq:explicit_N(s)_critical}
  N(s) =
  \begin{cases}
    N_{\outside}(s) = \frac{-1}{2\pi i} \oint_{\gamma} \frac{U_c(\xi)}{\xi - s} d\xi + 1, & \text{outside $\gamma$}, \\
    N_{\inside}(s) = \frac{1}{2\pi i} \oint_{\gamma} \frac{U_c(\xi)}{\xi - s} d\xi - 1, & \text{inside $\gamma$},
  \end{cases}
\end{equation}
and the explicit formulas are, with the help of \eqref{eq:c_identity},
\begin{align}
  N_{\outside}(s) = {}& 2c^4\frac{6s + 1}{s^2} - 2tc^2 \frac{1}{s} + 2ac \left( \frac{3}{2} + s - \frac{(s + 1)^{\frac{3}{2}}}{s^{\frac{1}{2}}} \right) + 1 \notag \\
  = {}& -2a J_c(s) + 2c^4 \frac{s + 1}{s^2} + ac \frac{(2s + 1)(s + 1)}{s} + \frac{s + 1}{s}, \\
  N_{\inside}(s) = {}& U_c(s) - N_{\outside}(s).
\end{align}
We note that \eqref{eq:c_identity} is the specialization of \cite[Equation (4.3)]{Claeys-Romano14}.

Below we define the functions $G(z)$ on $\compC \setminus [0, b]$ and $\G(z)$ on $\halfH \setminus [0, b]$, and then functions $\gfn(z)$ on $\compC \setminus (-\infty, b]$ and $\gfntilde(z)$ on $\halfH \setminus [0, b]$. $G(z)$ and $\G(z)$ are uniquely determined by $N(s)$ as in \cite[Equation (4.17)]{Claeys-Romano14}
\begin{equation} \label{eq:N(s)}
  N(s) =
  \begin{cases}
    J_c(s) G(J_c(s)), & \text{outside $\gamma$}, \\
    J_c(s) \G(J_c(s)), & \text{inside $\gamma$},
  \end{cases}
\end{equation}
and $\gfn(z)$ and $\gfntilde(z)$ are uniquely determined by
\begin{align}
  \gfn'(z) = {}& G(z), &\gfn(z) - \log z = {}& \bigO(z^{-1}) \quad \text{as } z \to \infty \text{ in $\compC$}, \\
  \gfntilde'(z) = {}& \G(z), & \gfntilde(z) - 2\log z = {}& \bigO(z^{-1}) \quad \text{as } z \to \infty \text{ in $\halfH$}.
\end{align}

Next, we consider the limits of $G(z)$, $\G(z)$, $\gfn(z)$ and $\gfntilde(z)$ as $z \to 0$. We use \cite[Equations (3.26) and (3.27)]{Wang-Zhang21} ($\omega = e^{2\pi i/3}$)
\begin{align}
  I_1(z) = {}& -1 -
  \begin{cases}
    \omega^2 (\frac{z}{c})^{\frac{2}{3}} (1 + \frac{1}{3} \omega^2 (\frac{z}{c})^{\frac{2}{3}} - \frac{1}{81} (\frac{z}{c})^2 + \bigO(z^{\frac{8}{3}})), & \arg z \in (0, \pi), \\
    \omega (\frac{z}{c})^{\frac{2}{3}} (1 + \frac{1}{3} \omega (\frac{z}{c})^{\frac{2}{3}} - \frac{1}{81} (\frac{z}{c})^2 + \bigO(z^{\frac{8}{3}})), & \arg z \in (-\pi, 0),
  \end{cases} \label{eq:I_1_at_0} \\
  I_2(z) = {}& -1 -
  \begin{cases}
    \omega (\frac{z}{c})^{\frac{2}{3}} (1 + \frac{1}{3} \omega (\frac{z}{c})^{\frac{2}{3}} - \frac{1}{81} (\frac{z}{c})^2 + \bigO(z^{\frac{8}{3}})), & \arg z \in (0, \frac{\pi}{2}), \\
    \omega^2 (\frac{z}{c})^{\frac{2}{3}} (1 + \frac{1}{3} \omega^2 (\frac{z}{c})^{\frac{2}{3}} - \frac{1}{81} (\frac{z}{c})^2 + \bigO(z^{\frac{8}{3}})), & \arg z \in (-\frac{\pi}{2}, 0).
  \end{cases} \label{eq:I_2_at_0}
\end{align}
Substituting \eqref{eq:I_1_at_0} and \eqref{eq:I_2_at_0} into the expression of $G(z)$ and $\G(z)$ in \eqref{eq:N(s)}, we have
\begin{multline} \label{eq:asy_G_quadratic}
  G(z) = G(J_c(I_1(z))) = -2 a + \frac{z}{(I_1(z) + 1)^2} \left( \frac{2c^2}{I_1(z)} + \frac{a}{c}(2I_1(z) + 1) + \frac{1}{c^2} \right) = -2a + (\frac{z}{c})^{-\frac{1}{3}} \\
  \times
          \begin{cases}
            (- 2c^3 - a + \frac{1}{c}) \omega^2 + \frac{1}{3}(10 c^3 -4a - \frac{2}{c})\omega (\frac{z}{c})^{\frac{2}{3}} & \\
            + \frac{1}{3}(-10 c^3 + a + \frac{1}{c}) (\frac{z}{c})^{\frac{4}{3}} & \\
            + \frac{1}{81}(200 c^3 - 8a - \frac{10}{c}) \omega^2 (\frac{z}{c})^2 + \bigO(z^{\frac{8}{3}}), & \arg z \in (0, \pi), \\
            (- 2c^3 - a + \frac{1}{c}) \omega^{-2} + \frac{1}{3}(10 c^3 -4a - 2 \frac{1}{c})\omega^{-1} (\frac{z}{c})^{\frac{2}{3}} & \\
            + \frac{1}{3}(-10 c^3 + a + \frac{1}{c}) (\frac{z}{c})^{\frac{4}{3}} & \\
            + \frac{1}{81}(200 c^3 - 8a - \frac{10}{c}) \omega^{-2} (\frac{z}{c})^2 + \bigO(z^{\frac{8}{3}}), & \arg z \in (-\pi, 0),
          \end{cases}
\end{multline}
\begin{multline} \label{eq:asy_tildeG_quadratic}
  \G(z) = \G(J_c(I_2(z))) = V'(z) + 2a - (\frac{z}{c})^{-\frac{1}{3}} \\
  \times
  \begin{cases}
    (- 2c^3 - a + \frac{1}{c}) \omega^{-2} + \frac{1}{3}(10 c^3 -4a - \frac{2}{c})\omega^{-1} (\frac{z}{c})^{\frac{2}{3}} & \\
    + \frac{1}{3}(-10 c^3 + a + \frac{1}{c}) (\frac{z}{c})^{\frac{4}{3}} & \\
    + \frac{1}{81}(200 c^3 - 8a - \frac{10}{c}) \omega^{-2} (\frac{z}{c})^2 + \bigO(z^{\frac{8}{3}}), & \arg z \in (0, \frac{\pi}{2}), \\
    (- 2c^3 - a + \frac{1}{c}) \omega^2 + \frac{1}{3}(10 c^3 -4a - \frac{2}{c})\omega (\frac{z}{c})^{\frac{2}{3}} & \\
    + \frac{1}{3}(-10 c^3 + a + \frac{1}{c}) (\frac{z}{c})^{\frac{4}{3}} & \\
    + \frac{1}{81}(200 c^3 - 8a - \frac{10}{c}) \omega^2 (\frac{z}{c})^2 + \bigO(z^{\frac{8}{3}}), & \arg z \in (-\frac{\pi}{2}, 0),
  \end{cases}
\end{multline}
We define, for $x \in (0, b)$,
\begin{equation}
  \psi(x) = \frac{1}{2\pi i}(G_-(x) - G_+(x)) = \frac{1}{2\pi i}(\G_-(x) - \G_+(x)).
\end{equation}
Then, as $x \to 0$,
\begin{multline} \label{eq:asy_psi}
  \psi(x) = \frac{\sqrt{3}}{2\pi} \left[ (- 2c^3 - a + \frac{1}{c}) (\frac{x}{c})^{-\frac{1}{3}} - \frac{1}{3}(10 c^3 -4a - \frac{2}{c}) (\frac{x}{c})^{\frac{1}{3}} + \frac{1}{81}(200 c^3 - 8a - \frac{10}{c}) (\frac{x}{c})^{\frac{5}{3}} \right] \\
  + \bigO(x^{\frac{7}{3}}).
\end{multline}

From the limiting formulas of $G(z)$ and $\G(z)$, we derive that
\begin{align}
  \gfn(z) = {}& C \pm \pi i - 2a z +  \\
              &
                \begin{cases}
                  \frac{3}{2}(- 2c^4 - ac + 1) \omega^2 (\frac{z}{c})^{\frac{2}{3}} + \frac{1}{4}(10 c^4 -4ac - 2)\omega (\frac{z}{c})^{\frac{4}{3}} & \\
                  + \frac{1}{6}(-10 c^4 + ac + 1) (\frac{z}{c})^2 & \\
                  + \frac{1}{216}(200 c^4 - 8ac - 10) \omega^2 (\frac{z}{c})^{\frac{8}{3}} + \bigO(z^{\frac{10}{3}}), & \arg z \in (0, \pi), \\
                  \frac{3}{2}(- 2c^4 - ac + 1) \omega^{-2} (\frac{z}{c})^{\frac{2}{3}} + \frac{1}{4}(10 c^4 -4ac - 2)\omega^{-1} (\frac{z}{c})^{\frac{4}{3}} & \\
                  + \frac{1}{6}(-10 c^4 + ac + 1) (\frac{z}{c})^2 & \\
                  + \frac{1}{216}(200 c^4 - 8ac - 10) \omega^{-2} (\frac{z}{c})^{\frac{8}{3}} + \bigO(z^{\frac{10}{3}}), & \arg z \in (-\pi, 0),
                \end{cases} \label{eq:qfn_at_0_quadratic} \\
  \gfntilde(z) = {}& C' \pm \pi i + V(z) + 2az - \\
              &
                \begin{cases}
                  \frac{3}{2}(- 2c^4 - ac + 1) \omega^{-2} (\frac{z}{c})^{\frac{2}{3}} + \frac{1}{4}(10 c^4 -4ac - 2)\omega^{-1} (\frac{z}{c})^{\frac{4}{3}} & \\
                  + \frac{1}{6}(-10 c^4 + ac + 1) (\frac{z}{c})^2 & \\
                  + \frac{1}{216}(200 c^4 - 8ac - 10) \omega^{-2} (\frac{z}{c})^{\frac{8}{3}} + \bigO(z^{\frac{10}{3}}), & \arg z \in (0, \frac{\pi}{2}), \\
                  \frac{3}{2}(- 2c^4 - ac + 1) \omega^2 (\frac{z}{c})^{\frac{2}{3}} + \frac{1}{4}(10 c^4 -4ac - 2)\omega (\frac{z}{c})^{\frac{4}{3}} & \\
                  + \frac{1}{6}(-10 c^4 + ac + 1) (\frac{z}{c})^2 & \\
                  + \frac{1}{216}(200 c^4 - 8ac - 10) \omega^2 (\frac{z}{c})^{\frac{8}{3}} + \bigO(z^{\frac{10}{3}}), & \arg z \in (-\frac{\pi}{2}, 0),
                \end{cases} \label{eq:qfntilde_at_0_quadratic}
\end{align}
where $C$ and $C'$ are constants and $\pm$ depends on whether $z \in \compC_{\pm}$. If we denote $\ell = C + C'$, and define
\begin{equation} \label{eq:defn_phi}
  \phi(z) = \gfn(z) + \gfntilde(z) - V(z) - \ell,
\end{equation}
then we have the limiting formula
\begin{equation} \label{eq:phi_quadratic_crit}
  \phi(z) = \pm 2\pi i + \sqrt{3} i \times
  \begin{cases}
    -\frac{3}{2}(- 2c^4 - ac + 1) (\frac{z}{c})^{\frac{2}{3}} + \frac{1}{4}(10 c^4 -4ac - 2) (\frac{z}{c})^{\frac{4}{3}} & \\
    - \frac{1}{216}(200 c^4 - 8ac - 10) (\frac{z}{c})^{\frac{8}{3}} + \bigO(z^{\frac{10}{3}}), & \arg z \in (0, \frac{\pi}{2}), \\
    \frac{3}{2}(- 2c^4 - ac + 1) (\frac{z}{c})^{\frac{2}{3}} - \frac{1}{4}(10 c^4 -4ac - 2) (\frac{z}{c})^{\frac{4}{3}} & \\
    + \frac{1}{216}(200 c^4 - 8ac - 10) (\frac{z}{c})^{\frac{8}{3}} + \bigO(z^{\frac{10}{3}}), & \arg z \in (-\frac{\pi}{2}, 0).
  \end{cases}
\end{equation}

In the Pearcey regime, we have, as $x \to 0$,
\begin{equation}
   \psi(x) = \frac{\sqrt{3}}{2\pi c} \left[ -\frac{4}{3} \tau \sqrt{\frac{1 - 3c^4}{2n}} (\frac{x}{c})^{-\frac{1}{3}} + (2 + \bigO(n^{-\frac{1}{2}})) (1 - 3c^4) (\frac{x}{c})^{\frac{1}{3}} \right] + \bigO(x^{\frac{5}{3}}),
\end{equation}
\begin{align}
  \gfn(z) = {}& C \pm \pi i - 2a z + \sqrt{1 - 3c^4} \notag \\
              & \times
                \begin{cases}
                  -(2n)^{-\frac{1}{2}} 2\tau \omega^{-1} (\frac{z}{c})^{\frac{2}{3}} - (\frac{3}{2} \sqrt{1 - 3c^4} + \bigO(n^{-\frac{1}{2}})) \omega^{-2} (\frac{z}{c})^{\frac{4}{3}} + \bigO(z^2), & \arg z \in (0, \pi), \\
                  -(2n)^{-\frac{1}{2}} 2\tau \omega (\frac{z}{c})^{\frac{2}{3}} - (\frac{3}{2} \sqrt{1 - 3c^4} + \bigO(n^{-\frac{1}{2}})) \omega^2 (\frac{z}{c})^{\frac{4}{3}} + \bigO(z^2), & \arg z \in (-\pi, 0),
                \end{cases} \\
  \gfntilde(z) = {}& C' \pm \pi i + V(z) + 2az - \sqrt{1 - 3c^4} \notag \\
              & \times
                \begin{cases}
                  -(2n)^{-\frac{1}{2}} 2\tau \omega (\frac{z}{c})^{\frac{2}{3}} - (\frac{3}{2} \sqrt{1 - 3c^4} + \bigO(n^{-\frac{1}{2}})) \omega^2 (\frac{z}{c})^{\frac{4}{3}} + \bigO(z^2), & \arg z \in (0, \frac{\pi}{2}), \\
                  -(2n)^{-\frac{1}{2}} 2\tau \omega^{-1} (\frac{z}{c})^{\frac{2}{3}} - (\frac{3}{2} \sqrt{1 - 3c^4} + \bigO(n^{-\frac{1}{2}})) \omega^{-2} (\frac{z}{c})^{\frac{4}{3}} + \bigO(z^2), & \arg z \in (-\frac{\pi}{2}, 0),
                \end{cases}
\end{align}
\begin{multline}
  \phi(z) = \pm 2\pi i + \sqrt{3} \sqrt{1 - 3c^4} i \\
  \times
  \begin{cases}
    (2n)^{-\frac{1}{2}} 2\tau (\frac{z}{c})^{\frac{2}{3}} - (\frac{3}{2} \sqrt{1 - 3c^4} + \bigO(n^{-\frac{1}{2}})) (\frac{z}{c})^{\frac{4}{3}}  + \bigO(z^2), & \arg z \in (0, \frac{\pi}{2}), \\
    -(2n)^{-\frac{1}{2}} 2\tau (\frac{z}{c})^{\frac{2}{3}} + (\frac{3}{2} \sqrt{1 - 3c^4} + \bigO(n^{-\frac{1}{2}})) (\frac{z}{c})^{\frac{4}{3}}  + \bigO(z^2), & \arg z \in (-\frac{\pi}{2}, 0).
  \end{cases}
\end{multline}

In the multi-critical regime, we have, as $x \to 0$,
\begin{equation}
   \psi(x) = \frac{\sqrt{3}}{2\pi c} \left[ -\frac{2}{3} \tau n^{-\frac{3}{4}} (\frac{x}{c})^{-\frac{1}{3}} + (2\sigma n^{-\frac{1}{2}} + \bigO(n^{-\frac{3}{4}})) (\frac{x}{c})^{\frac{1}{3}} + (\frac{2}{3} + \bigO(n^{-\frac{1}{2}})) (\frac{x}{c})^{\frac{5}{3}} + \bigO(x^{\frac{7}{3}}) \right],
\end{equation}
\begin{multline}
  \gfn(z) = C \pm \pi i - 2a z + \\
  \begin{cases}
    -\tau n^{-\frac{3}{4}} \omega^{-1} (\frac{z}{c})^{\frac{2}{3}} - (\frac{3}{2} \sigma n^{-\frac{1}{2}} + \bigO(n^{-\frac{3}{4}})) \omega^{-2} (\frac{z}{c})^{\frac{4}{3}} & \\
    - (\frac{1}{3} + \bigO(n^{-\frac{1}{2}})) (\frac{z}{c})^2 + (\frac{1}{4} + \bigO(n^{-\frac{1}{2}})) \omega^{-1} (\frac{z}{c})^{\frac{8}{3}} + \bigO(z^{\frac{10}{3}}), & \arg z \in (0, \pi), \\
    -\tau n^{-\frac{3}{4}} \omega (\frac{z}{c})^{\frac{2}{3}} - (\frac{3}{2} \sigma n^{-\frac{1}{2}} + \bigO(n^{-\frac{3}{4}})) \omega^2 (\frac{z}{c})^{\frac{4}{3}} & \\
    - (\frac{1}{3} + \bigO(n^{-\frac{1}{2}})) (\frac{z}{c})^2 + (\frac{1}{4} + \bigO(n^{-\frac{1}{2}})) \omega (\frac{z}{c})^{\frac{8}{3}} + \bigO(z^{\frac{10}{3}}), & \arg z \in (-\pi, 0),
  \end{cases}
\end{multline}
\begin{multline}
  \gfntilde(z) = C' \pm \pi i + V(z) + 2az - \\
  \begin{cases}
    -\tau n^{-\frac{3}{4}} \omega (\frac{z}{c})^{\frac{2}{3}} - (\frac{3}{2} \sigma n^{-\frac{1}{2}} + \bigO(n^{-\frac{3}{4}})) \omega^2 (\frac{z}{c})^{\frac{4}{3}} & \\
    - (\frac{1}{3} + \bigO(n^{-\frac{1}{2}})) (\frac{z}{c})^2 + (\frac{1}{4} + \bigO(n^{-\frac{1}{2}})) \omega (\frac{z}{c})^{\frac{8}{3}} + \bigO(z^{\frac{10}{3}}), & \arg z \in (0, \frac{\pi}{2}), \\
    -\tau n^{-\frac{3}{4}} \omega^{-1} (\frac{z}{c})^{\frac{2}{3}} - (\frac{3}{2} \sigma n^{-\frac{1}{2}} + \bigO(n^{-\frac{3}{4}})) \omega^{-2} (\frac{z}{c})^{\frac{4}{3}} & \\
    - (\frac{1}{3} + \bigO(n^{-\frac{1}{2}})) (\frac{z}{c})^2 + (\frac{1}{4} + \bigO(n^{-\frac{1}{2}})) \omega^{-1} (\frac{z}{c})^{\frac{8}{3}} + \bigO(z^{\frac{10}{3}}), & \arg z \in (-\frac{\pi}{2}, 0),
  \end{cases}
\end{multline}
\begin{multline}
  \phi(z) = \pm 2\pi i + \sqrt{3} i \\
  \times
  \begin{cases}
   \tau n^{-\frac{3}{4}} (\frac{z}{c})^{\frac{2}{3}} - (\frac{3}{2} \sigma n^{-\frac{1}{2}} + \bigO(n^{-\frac{3}{4}})) (\frac{z}{c})^{\frac{4}{3}} - (\frac{1}{4} + \bigO(n^{-\frac{1}{2}})) (\frac{z}{c})^{\frac{8}{3}} + \bigO(z^{\frac{10}{3}}), & \arg z \in (0, \frac{\pi}{2}), \\
    -\tau n^{-\frac{3}{4}} (\frac{z}{c})^{\frac{2}{3}} + (\frac{3}{2} \sigma n^{-\frac{1}{2}} + \bigO(n^{-\frac{3}{4}})) (\frac{z}{c})^{\frac{4}{3}} + (\frac{1}{4} + \bigO(n^{-\frac{1}{2}})) (\frac{z}{c})^{\frac{8}{3}} + \bigO(z^{\frac{10}{3}}), & \arg z \in (-\frac{\pi}{2}, 0).
  \end{cases}
\end{multline}

Finally, we have limiting formulas for $G(z)$, $\G(z)$, $\gfn(z)$, and $\gfntilde(z)$ as $z \to b$. We only state the result that in both regimes, for $z \in D(b, \epsilon) \setminus [0, b]$ with $\epsilon > 0$ small enough,
\begin{equation} \label{eq:behaviour_phi_t_at_b}
  \phi(z) = -\frac{4\pi}{3} d_2(z - b)^{\frac{3}{2}}( 1 + \bigO(z - b)),
\end{equation}
where $d_2 > 0$ and the term $\bigO(z - b)$ is analytic in $D(b, \epsilon)$.

\begin{rmk} \label{rmk:positivity_of_psi}
  We find that in the Pearcey regime, if $\tau \leq 0$, then $\psi(x) > 0$ for all $x \in (0, b)$. Similarly, in the multi-critical regime, if $\tau \leq 0$ and $\sigma > 0$, then $\psi(x) > 0$ for all $x \in (0, b)$. Hence, if $\tau \leq 0$ in the Pearcey regime, or $\tau \leq 0$ and $\sigma > 0$ in the multi-critical regime, then $1_{(0, b)}\psi(x) dx$ is the equilibrium measure defined by the minimization problem in \cite[Equation (1.22)]{Claeys-Romano14} for the functional
  \begin{equation}
    \frac{1}{2} \iint \log \frac{1}{\lvert x - y \rvert} d\nu(x) d\nu(y) + \frac{1}{2} \iint \log \frac{1}{\lvert x^{\theta} - y^{\theta} \rvert} d\nu(x) d\nu(y) + \int V(x) d\nu(x)
  \end{equation}
  among all probability measures supported on $\realR_+$, and $\gfn$ and $\gfntilde$ are defined from the equilibrium measure by \cite[Equations (4.4) and (4.5)]{Claeys-Romano14}. In general, in either the Pearcey or the multi-critical regime, $1_{(0, b)}\psi(x) dx$ is not a probability measure. However, from the construction of $\psi(x)$, we have that in the Pearcey case with a fixed $\tau \in \realR$, there exists $C \geq 0$ depending on $\tau$ such that $\psi(x) > 0$ for $x \in (C n^{-3/4}, b)$. Similarly, in the multi-critical case with fixed $(\tau, \sigma) \in \realR^2$, there exists $C \geq 0$ depending on $(\tau, \sigma)$ such that $\psi(x) > 0$ for $x \in (C n^{-3/8}, b)$.
\end{rmk}

\subsection{Transforms $Y \to T \to S \to Q$} \label{sec:transform_YTSQ}

Let $k \in \intZ$ be a constant. From $Y^{(n + k)} = (Y^{(n + k)}_1, Y^{(n + k)}_2) = (p_{n + k}, \C p_{n + k})$ defined in RH problem \ref{RHP:defn_Y}, we let
\begin{align}
  T^{(n + k)}(z) = (T^{(n + k)}_1(z), T^{(n + k)}_2(z)) := (Y^{(n + k)}_1(z) e^{-n\gfn(z)}, Y^{(n + k)}_2(z) e^{n(\gfntilde(z) - \ell)}),
\end{align}
Then $T^{(n + k)}$ satisfies the following RH problem:
\begin{RHP}
  $T^{(n + k)} = (T^{(n + k)}_1, T^{(n + k)}_2)$ is analytic on $(\compC \setminus \realR_+, \halfH \setminus \realR_+)$ and satisfies the following:
  \begin{enumerate}
  \item
    As $z \to \infty$ in $\compC$, $T^{(n + k)}_1(z) = z^k + \bigO(z^{k - 1})$.
  \item
    As $z \to \infty$ in $\halfH$, $T^{(n + k)}_2(z) = \bigO(z^{-2(k + 1)})$.
  \item
    As $z \to 0$ in $\compC$, $T^{(n + k)}_1(z)$ has the same behaviour as $Y^{(n + k)}_1(z)$, and as $z \to 0$ in $\halfH$, $T^{(n + k)}_2(z)$ has the same behaviour as $Y^{(n + k)}_2(z)$.
  \item
    As $z \to b$, $T^{(n + k)}(z) = \bigO(1)$.
  \item
    On $\realR_+$,
    \begin{multline}
      T^{(n + k)}_+(x) =  T^{(n + k)}_-(x)
      \begin{pmatrix}
        e^{n(\gfn_-(x) - \gfn_+(x))} & \frac{1}{2x^{1 - \alpha}} e^{n(\gfn_-(x) + \gfntilde_+(x) - V(x) - \ell)} \\
        0 & e^{n(\gfntilde_+(x) - \gfntilde_-(x))}
      \end{pmatrix} \\
      + \gamma T^{(n + k)}_1(-x) \left( 0, \frac{1}{2x^{1 - \alpha}} e^{n(\gfntilde_+(x) + \gfn(-x) - V(-x) - \ell)} \right).
    \end{multline}
  \end{enumerate}
\end{RHP}

Next, we consider the contour $\Sigma_1 \subseteq \halfH_+$ that connects $0$ and $b$, and $\Sigma_2 = \overline{\Sigma_1}$. We assume that $\Sigma_1$ is flat and close to the interval $(0, b)$, whose precise shape is specified in Section \ref{subsec:local_parametrices} at its end points $0$ and $b$ and in Section \ref{subsubsec:final_trans_p} for the remaining part. We also consider the contour $\Sigma'_1 \subseteq (\halfH_+ \setminus \{ \text{upper lens} \})$ that connects $0$ and $b' \in (b, +\infty)$, and $\Sigma'_2 = \overline{\Sigma'_1}$. The precise shape of $\Sigma'_1$ is specified in Section \ref{subsec:local_parametrices} at its end point $0$ and in Section \ref{subsubsec:final_trans_p} for the remaining part. We denote $\Sigma = \Sigma_1 \cup \Sigma_2 \cup \Sigma'_1 \cup \Sigma'_2 \cup \realR_+$. We call the region enclosed by $\Sigma'_1$ and $\Sigma'_2$ the ``lens'', the region enclosed by $\Sigma_1$ and $\Sigma_2$ the ``inner lens'', and the part of the lens that is outside the inner lens the ``outer lens''. We call the intersection with $\compC_+$ (resp.~$\compC_-$) of the (inner/outer) lens the upper (resp.~lower) (inner/outer) lens. See Figure \ref{fig:outer_and_inner_lens} for the shape of $\Sigma$. (We note that $\Sigma'_1$ is leftward, and $\Sigma'_2$, $\Sigma_1$, $\Sigma_2$ are rightward.)

\begin{figure}[htb]
  \centering
  \includegraphics{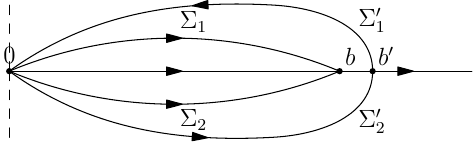}
  \caption{Shape of $\Sigma$ and the lens.}
  \label{fig:outer_and_inner_lens}
\end{figure}
We define
\begin{equation}
  S^{(n + k)}(z) =
  \begin{cases}
    T^{(n + k)}(z), & \text{outside the lens}, \\
    T^{(n + k)}(z) + \gamma T^{(n + k)}_1(-z) \left( e^{n(\gfn(-z) - \gfn(z) - 4az)}, 0 \right), & \text{outer lens}, \\
    T^{(n + k)}(z)
    \begin{pmatrix}
      1 & 0 \\
      2z^{1 - \alpha} e^{-n\phi(z)} & 1
    \end{pmatrix} & \\
    \phantom{T^{(n + k)}(z)}
    + \gamma T^{(n + k)}_1(-z) \left( e^{n(\gfn(-z) - \gfn(z) - 4az)}, 0 \right),
    & \text{lower inner lens}, \\
    T^{(n + k)}(z)
    \begin{pmatrix}
      1 & 0 \\
      -2z^{1 - \alpha} e^{-n\phi(z)} & 1
    \end{pmatrix} & \\
    \phantom{T^{(n + k)}(z)}
    + \gamma T^{(n + k)}_1(-z) \left( e^{n(\gfn(-z) - \gfn(z) - 4az)}, 0 \right),
    & \text{upper inner lens},
  \end{cases}
\end{equation}
where $\phi(z)$ is defined in \eqref{eq:defn_phi}.

We have that $S^{(n + k)}(z)$ satisfies the following RH problem:
\begin{RHP}
  $S^{(n + k)} = (S^{(n + k)}_1, S^{(n + k)}_2)$ is analytic on $(\compC \setminus \Sigma, \halfH \setminus \Sigma)$ and satisfies the following:
  \begin{enumerate}
  \item
    As $z \to \infty$ in $\compC$, and $z \to \infty$ in $\halfH$, $S^{(n + k)}(z)$ has the same behaviour as $T^{(n + k)}(z)$.
  \item
    As $z \to 0$ in $\compC \setminus \Sigma$,
    \begin{equation} \label{eq:S_1_at_0}
      S^{(n + k)}_1(z) =
      \begin{cases}
        \bigO(1), & \text{$z$ outside the inner lens or $-1 < \alpha < 1$}, \\
        \bigO(\log z), & \text{$z$ inside the inner lens and $\alpha = 1$}, \\
        \bigO(z^{1 - \alpha}), & \text{$z$ inside the inner lens and $\alpha > 1$},
      \end{cases}
    \end{equation}
    and as $z \to 0$ in $\halfH \setminus \Sigma$,
    \begin{equation} \label{eq:S_2_at_0}
      S^{(n + k)}_2(z) =
      \begin{cases}
        \bigO(1), & \alpha > 1, \\
        \bigO(\log z), & \alpha = 1, \\
        \bigO(z^{\alpha - 1}), & -1 < \alpha < 1.
      \end{cases}
    \end{equation}
  \item
    As $z \to b$, $S^{(n + k)}(z) = \bigO(1)$.
  \item
    On $\Sigma$,
    \begin{equation}
      S^{(n + k)}_+(z) =
      \begin{cases}
        S^{(n + k)}_-(z) + \gamma S^{(n + k)}_1(-z) \left( e^{n(\gfn(-z) - \gfn(z) - 4az)}, 0 \right), & z \in \Sigma'_1 \cup \Sigma'_2, \\
        S^{(n + k)}_-(z)
        \begin{pmatrix}
          1 & 0 \\
          2z^{1 - \alpha} e^{-n \phi(z)} & 1
        \end{pmatrix},
        & z \in \Sigma_1 \cup \Sigma_2, \\
        S^{(n + k)}_-(z)
        \begin{pmatrix}
          0 & \frac{1}{2z^{1 - \alpha}} \\
          -2z^{1 - \alpha} & 0
        \end{pmatrix},
        & z \in (0, b), \\
        S^{(n + k)}_-(z)
        \begin{pmatrix}
          1 & \frac{1}{2z^{1 - \alpha}} e^{n \phi(z)} \\
          0 & 1
        \end{pmatrix},
        & z \in (b, b'), \\
        S^{(n + k)}_-(z)
        \begin{pmatrix}
          1 & \frac{1}{2z^{1 - \alpha}} e^{n \phi(z)} \\
          0 & 1
        \end{pmatrix}
        & \\
        \quad + \gamma S^{(n + k)}_1(-z) \left( 0, \frac{1}{2z^{1 - \alpha}} e^{n(\gfntilde(z) + \gfn(-z) - V(-z) - \ell)} \right), & z \in (b', \infty),
      \end{cases}
    \end{equation}
  \end{enumerate}
\end{RHP}
Next, we need to construct the global parametrix $P^{(j, \infty)}$ for $j \in \intZ$, which satisfies the following RH problem:
\begin{RHP} \label{RH:P^infty}
  $P^{(j, \infty)} = (P^{(j, \infty)}_1, P^{(j, \infty)}_2)$ is analytic in $(\compC \setminus [0, b], \halfH \setminus [0, b])$ and satisfies the following:
  \begin{enumerate}
  \item
    For $x \in (0, b)$,
    \begin{equation}
      P^{(j, \infty)}_+(x) = P^{(j, \infty)}_-(x)
      \begin{pmatrix}
        0 & \frac{1}{2x^{1 - \alpha}} \\
        -2x^{1 - \alpha} & 0
      \end{pmatrix}.
    \end{equation}
  \item
    As $z \to \infty$ in $\compC$, $P^{(j, \infty)}_1(z) = z^j + \bigO(z^{j - 1})$.
  \item
    As $z \to \infty$ in $\halfH$, $P^{(j, \infty)}_2(z) = \bigO(z^{-2(j + 1)})$.
  \item
    For $x > 0$, we have the boundary condition $P^{(j, \infty)}_2(e^{\pi i/2} x) = P^{(j, \infty)}_2(e^{-\pi i/2} x)$.
  \end{enumerate}
\end{RHP}
The construction of $P^{(j, \infty)}$ is given in \cite[Section 3.3]{Wang-Zhang21} for the $j = 0$ case, and it is straightforward to generalize it to integer $j$. Here we give the construction and refer the reader to \cite[Section 3.3]{Wang-Zhang21} for details.
Let
\begin{equation}
  \mathcal{P}^{(j)}(s) = (cs)^j \times
  \begin{cases}
    \frac{s}{\sqrt{(s + 1)(s - s_b)}} \left( \frac{s + 1}{s} \right)^{\frac{1 - \alpha}{2}}, & s \in \compC \setminus \overline{D}, \\
    \frac{c^{\alpha - 1} s(s + 1)^{\alpha - 1}}{2\sqrt{(s + 1)(s - s_b)}}, & s \in D,
  \end{cases}
\end{equation}
where $c$ is the parameter for $t$ and $a$, and appears in the mapping \eqref{eq:defn_J_c}.  We let
\begin{align}
  P^{(j, \infty)}_1(z) = {}& \mathcal{P}^{(j)}(I_1(z)), & z \in \compC \setminus [0, b], \label{eq:defn_P^jinfty_1} \\
  P^{(j, \infty)}_2(z) = {}& \mathcal{P}^{(j)}(I_2(z)), & z \in \halfH \setminus [0, b]. \label{eq:defn_P^jinfty_2}
\end{align}
Then $P^{(j, \infty)}$, as defined above, is a solution of RH problem \ref{RH:P^infty}, and we will take it as our definition of $P^{(j, \infty)}$.

Like \cite[Proposition 3.6]{Wang-Zhang21}, we have that as $z \to 0$,
\begin{align}
  P^{(j, \infty)}_1(z) = {}& \frac{(-1)^j}{\sqrt{6}} c^{j + \frac{\alpha}{3}} \times
                          \begin{cases}
                            2 e^{\frac{\alpha}{3} \pi i} z^{-\frac{\alpha}{3}} (1 + \bigO(z^{\frac{2}{3}})), & \arg z \in (0, \pi), \\
                            2 e^{-\frac{\alpha}{3} \pi i} z^{-\frac{\alpha}{3}} (1 + \bigO(z^{\frac{2}{3}})), & \arg z \in (-\pi, 0),
                          \end{cases} \\
  P^{(j, \infty)}_2(z) = {}& \frac{(-1)^j}{\sqrt{6}} c^{j + \frac{\alpha}{3}} \times
                          \begin{cases}
                            e^{-\frac{\alpha}{3} \pi i} z^{\frac{2 \alpha}{3} - 1} (1 + \bigO(z^{\frac{2}{3}})), & \arg z \in (0, \frac{\pi}{2}), \\
                            -e^{\frac{\alpha}{3} \pi i} z^{\frac{2 \alpha}{3} - 1} (1 + \bigO(z^{\frac{2}{3}})), & \arg z \in (-\frac{\pi}{2}, 0).
                          \end{cases}
\end{align}

Then we define the function
\begin{equation}
  Q^{(n + k, j)}(z) = (Q^{(n + k, j)}_1(z), Q^{(n + k, j)}_2(z)) := \left( \frac{S^{(n + k)}_1(z)}{P^{(j, \infty)}_1(z)}, \frac{S^{(n + k)}_2(z)}{P^{(j, \infty)}_2(z)} \right),
\end{equation}
and it satisfies the following RH problem:
\begin{RHP}
  $Q^{(n + k, j)} = (Q^{(n + k, j)}_1, Q^{(n + k, j)}_2)$ is analytic on $(\compC \setminus \Sigma, \halfH \setminus \Sigma)$ and satisfies the following:
  \begin{enumerate}
  \item
    As $z \to \infty$ in $\compC$, $Q^{(n + k, j)}_1(z) = z^{k - j}(1 + \bigO(z^{-1}))$.
  \item
    As $z \to \infty$ in $\halfH$, $Q^{(n + k, j)}_2(z) = \bigO(z^{2(j - k - 1)})$.
  \item
    As $z \to 0$ in $\compC \setminus \Sigma$,
    \begin{equation}
      Q^{(n + k, j)}_1(z) =
      \begin{cases}
        \bigO(z^{\frac{\alpha}{3}}), & \text{$z$ outside the inner lens or $-1 < \alpha < 1$}, \\
        \bigO(z^{\frac{1}{3}} \log z), & \text{$z$ inside the inner lens and $\alpha = 1$}, \\
        \bigO(z^{1 - \frac{2\alpha}{3}}), & \text{$z$ inside the inner lens and $\alpha > 1$},
      \end{cases}
    \end{equation}
    and as $z \to 0$ in $\halfH \setminus \Sigma$,
    \begin{equation}
      Q^{(n + k, j)}_2(z) =
      \begin{cases}
        \bigO(z^{1 - \frac{2\alpha}{3}}), & \alpha > 1, \\
        \bigO(z^{\frac{1}{3}} \log z), & \alpha = 1, \\
        \bigO(z^{\frac{\alpha}{3}}), & -1 < \alpha < 1.
      \end{cases}
    \end{equation}
  \item
    As $z \to b$, $Q^{(n + k, j)}_1(z) = \bigO((z - b)^{1/4})$ and $Q^{(n + k, j)}_2(z) = \bigO((z - b)^{1/4})$.
  \item
    On $\Sigma$,
    \begin{multline}
      Q^{(n + k, j)}_+(z) = \\
      \begin{cases}
        Q^{(n + k, j)}_-(z) + \gamma Q^{(n + k, j)}_1(-z) \left( \frac{P^{(j, \infty)}_1(-z)}{P^{(j, \infty)}_1(z)} e^{n(\gfn(-z) - \gfn(z) - 4az)}, 0 \right), & z \in \Sigma'_1 \cup \Sigma'_2, \\
        Q^{(n + k, j)}_-(z)
        \begin{pmatrix}
          1 & 0 \\
          2z^{1 - \alpha} \frac{P^{(j, \infty)}_2(z)}{P^{(j, \infty)}_1(z)} e^{-n \phi(z)} & 1
        \end{pmatrix},
        & z \in \Sigma_1 \cup \Sigma_2, \\
        Q^{(n + k, j)}_-(z)
        \begin{pmatrix}
          0 & 1 \\
          1 & 0
        \end{pmatrix},
        & z \in (0, b), \\
        Q^{(n + k, j)}_-(z)
        \begin{pmatrix}
          1 & \frac{1}{2z^{1 - \alpha}} \frac{P^{(j, \infty)}_1(z)}{P^{(j, \infty)}_2(z)} e^{n \phi(z)} \\
          0 & 1
        \end{pmatrix},
        & z \in (b, b'), \\
        Q^{(n + k, j)}_-(z)
        \begin{pmatrix}
          1 & \frac{1}{2z^{1 - \alpha}} \frac{P^{(j, \infty)}_1(z)}{P^{(j, \infty)}_2(z)} e^{n \phi(z)} \\
          0 & 1
        \end{pmatrix}
        & \\
        \quad + \gamma Q^{(n + k, j)}_1(-z) \left( 0, \frac{1}{2z^{1 - \alpha}} \frac{P^{(j, \infty)}_1(-z)}{P^{(j, \infty)}_2(z)} e^{n(\gfntilde(z) + \gfn(-z) - V(-z) - \ell)} \right), & z \in (b', +\infty),
      \end{cases}
    \end{multline}
  \end{enumerate}
\end{RHP}

\subsection{Local parametrices} \label{subsec:local_parametrices}

\paragraph{Local parametrix at $b$}

Our construction of the local parametrix near the right endpoint $b$ is given by the Airy parametrix, analogous to \cite{Wang-Zhang21}. Below we cite corresponding formulas in \cite{Wang-Zhang21} for the convenience of the reader. Let $\epsilon$ be a small positive number. We define, for $z$ in a neighbourhood $D(b, \epsilon)$ of $b$, \cite[Equation (3.35)]{Wang-Zhang21}
\begin{equation} \label{def:fb_hard}
  f_b(z) = \left( -\frac{3}{4} \phi(z) \right)^{\frac{2}{3}}.
\end{equation}
Due to \eqref{eq:behaviour_phi_t_at_b}, it is a conformal mapping with $f_b(b) = 0$ and $f'_b(b) > 0$.

We define \cite[Equation (3.38)]{Wang-Zhang21}
\begin{equation} \label{eq:defn_P^b}
  P^{(j, b)}(z) = E^{(j, b)}(z) \Psi^{(\Ai)}(n^{\frac{2}{3}} f_b(z))
  \begin{pmatrix}
    e^{-\frac{n}{2} \phi(z)} g^{(j, b)}_1(z) & 0 \\
    0 & e^{\frac{n}{2} \phi(z)} g^{(j, b)}_2(z)
  \end{pmatrix},
\end{equation}
where $\Psi^{(\Ai)}$ is the well-known Airy parametrix as given in Appendix \ref{app:Airy}, and \cite[Equations (3.37) and (3.43)]{Wang-Zhang21}
\begin{align}
  g^{(j, b)}_1(z) = {}& \frac{z^{(1 - \alpha)/2}}{P^{(j, \infty)}_1(z)}, & g^{(j, b)}_2(z) = {}& \frac{z^{(\alpha - 1)/2}}{2 P^{(j, \infty)}_2(z)},
\end{align}
\begin{equation}
  E^{(j, b)}(z) = \frac{1}{\sqrt{2}}
  \begin{pmatrix}
    g^{(j, b)}_1(z) & 0 \\
    0 & g^{(j, b)}_2(z)
  \end{pmatrix}^{-1}
  e^{\frac{\pi i}{4} \sigma_3}
  \begin{pmatrix}
    1 & -1 \\
    1 & 1
  \end{pmatrix}
  \begin{pmatrix}
    n^{\frac{1}{6}} f_b(z)^{\frac{1}{4}} & 0 \\
    0 & n^{-\frac{1}{6}} f_b(z)^{-\frac{1}{4}}
  \end{pmatrix}.
\end{equation}
We have that $P^{(j, b)}(z)$ satisfies an RH problem analogous to \cite[RH problem 3.9]{Wang-Zhang21}.

Now we specify the shape of $\Sigma_1$ (and then $\Sigma_2 = \overline{\Sigma_1}$) in $D(b, \epsilon)$ as $f^{-1}_b(\{ e^{\frac{2\pi i}{3}} [0, +\infty) \}) \cap D(b, \epsilon)$. Then we define the vector-valued function $V^{(n + k, j, b)}$ by \cite[Equation (3.46)]{Wang-Zhang21}
\begin{equation} \label{eq:defn_V^b}
  V^{(n + k, j, b)}(z) = Q^{(n + k, j)}(z) P^{(j, b)}(z)^{-1}, \quad z \in D(b, \epsilon) \setminus \Sigma,
\end{equation}
which satisfies an RH problem analogous to \cite[RH problem 3.10]{Wang-Zhang21}.

\paragraph{Local parametrix at $0$}

Let $r_n > 0$ be a small constant specified in \eqref{def:rn} for the Pearcey regime and in \eqref{def:rn_crit} for the multi-critical regime. We require that in the region $D(0, r_n)$, $\Sigma_1$ and $\Sigma_2$ coincide with the rays $\{ \arg z = \pm \pi/8 \}$, and $\Sigma'_1$ and $\Sigma'_2$ coincide with the rays $\{ \arg z = \pm 3\pi/8 \}$. we define $U^{(n + k, j)} = (U^{(n + k, j)}_0, U^{(n + k, j)}_1, U^{(n + k, j)}_2)$ in the region $D(0, r_n)$ by
\begin{align}
  U^{(n + k, j)}_0(z^2) = {}& Q^{(n + k, j)}_2(z), & \arg z \in {}& (-\frac{\pi}{2}, \frac{\pi}{2}), \label{eq:defn_U(n+k):1} \\
  U^{(n + k, j)}_1(z^2) = {}&
                           Q^{(n + k, j)}_1(z), & \arg z \in {}& (-\frac{\pi}{2}, \frac{\pi}{2}) \setminus \{ \pm\frac{3\pi}{8}, \pm\frac{\pi}{8}, 0 \} \\
  U^{(n + k, j)}_2(z^2) = {}& Q^{(n + k, j)}_1(-z), & \arg z \in {}& (-\frac{\pi}{2}, \frac{\pi}{2}). \label{eq:defn_U(n+k):3}
\end{align}
Then $U^{(n + k, j)}$ satisfies the following RH problem:
\begin{RHP} \label{rhp:U}
  $U^{(n + k, j)} = (U^{(n + k, j)}_0, U^{(n + k, j)}_1, U^{(n + k, j)}_2)$ is defined and analytic in $D(0, r_n) \setminus \Gamma_{\Phi}$, where $\Gamma_{\Phi}$ is defined in RH problem \ref{RHP:model}, and satisfies the following:
  \begin{enumerate}
  \item
    On $\Sigma^{(0)} = D(0, r_n) \cap \Gamma_{\Phi}$, where all rays are oriented outward from $0$, we have
    \begin{multline} \label{eq:defn_J_U}
      U^{(n + k, j)}_+(z) = U^{(n + k, j)}_-(z) J^{(j)}_U(z), \\
      J^{(j)}_U(z) =
      \begin{cases}
        \begin{pmatrix}
          1 & 0 & 0 \\
          0 & 1 & 0 \\
          0 & -\gamma \frac{P^{(j, \infty)}_1(-\sqrt{z})}{P^{(j, \infty)}_1(\sqrt{z})} e^{n(\gfn(-\sqrt{z}) - \gfn(\sqrt{z}) - 4a\sqrt{z})} & 1
        \end{pmatrix},
        & \arg z = \frac{3\pi}{4}, \\
        \begin{pmatrix}
          1 & 0 & 0 \\
          0 & 1 & 0 \\
          0 & \gamma \frac{P^{(j, \infty)}_1(-\sqrt{z})}{P^{(j, \infty)}_1(\sqrt{z})} e^{n(\gfn(-\sqrt{z}) - \gfn(\sqrt{z}) - 4a\sqrt{z})} & 1
        \end{pmatrix},
        & \arg z = -\frac{3\pi}{4}, \\
        \begin{pmatrix}
          1 & 2z^{\frac{1 - \alpha}{2}} \frac{P^{(j, \infty)}_2(\sqrt{z})}{P^{(j, \infty)}_1(\sqrt{z})} e^{-n \phi(\sqrt{z})} & 0 \\
          0 & 1 & 0 \\
          0 & 0 & 1
        \end{pmatrix},
        & \arg z = \pm \frac{\pi}{4}, \\
        \begin{pmatrix}
          0 & 1 & 0 \\
          1 & 0 & 0 \\
          0 & 0 & 1
        \end{pmatrix},
        & z \in \realR_+, \\
        \begin{pmatrix}
          1 & 0 & 0 \\
          0 & 0 & 1 \\
          0 & 1 & 0
        \end{pmatrix},
        & z \in \realR_-.
      \end{cases}
    \end{multline}
  \item
    As $z \to 0$,
    \begin{align}
      U^{(n + k, j)}_0(z) = {}&
      \begin{cases}
        \bigO(z^{\frac{1}{2} - \frac{\alpha}{3}}), & \alpha > 1, \\
        \bigO(z^{\frac{1}{3}} \log z), & \alpha = 1, \\
        \bigO(z^{\frac{\alpha}{6}}), & -1 < \alpha < 1,
      \end{cases} \\
       U^{(n + k, j)}_1(z) = {}&
                    \begin{cases}
                      \bigO(z^{\frac{1}{2} - \frac{\alpha}{3}}), & \arg z \in (-\frac{\pi}{4}, 0) \cup (0, \frac{\pi}{4}) \text{ and } \alpha > 1, \\
                      \bigO(z^{\frac{1}{6}} \log z), & \arg z \in (-\frac{\pi}{4}, 0) \cup (0, \frac{\pi}{4}) \text{ and } \alpha = 1, \\
                      \bigO(z^{\frac{\alpha}{6}}), & \arg z \in (-\pi, -\frac{\pi}{4}) \cup (\frac{\pi}{4}, \pi) \text{ and } -1 < \alpha < 1,
                    \end{cases} \\
      U^{(n + k, j)}_2(z) = {}& \bigO(z^{\frac{\alpha}{6}}).
    \end{align}
  \end{enumerate}
\end{RHP}

As in \cite[Equations (3.104)--(3.109)]{Wang-Zhang21}, we define
\begin{equation} \label{def:M}
  M(z) = \diag(m_0(z), m_1(z), m_2(z)), \quad z \in \compC \setminus \realR,
\end{equation}
where
\begin{align}
  m_0(z) = {}&
               \begin{cases}
                 \gfntilde(z^{\frac{1}{2}}) - \gfntilde_+(0), & z \in \compC_+, \\
                 2\pi i + \gfntilde(z^{\frac{1}{2}}) - \gfntilde_+(0), & z \in \compC_-,
               \end{cases}
               \label{def:m0}\\
  m_1(z) = {}&
               \begin{cases}
                 -\gfn(z^{\frac{1}{2}}) + V(z^{\frac{1}{2}}) + \ell - \gfntilde_-(0), & z \in \compC_+, \\
                 -2\pi i - \gfn(z^{\frac{1}{2}}) + V(z^{\frac{1}{2}}) + \ell - \gfntilde_-(0), & z \in \compC_-,
               \end{cases} \\
  m_2(z) = {}&
               \begin{cases}
                 -\gfn(-z^{\frac{1}{2}}) + V(-z^{\frac{1}{2}}) + \ell - \gfntilde_-(0), & z \in \compC_-, \\
                 -2\pi i - \gfn(-z^{\frac{1}{2}}) +  V(-z^{\frac{1}{2}}) + \ell - \gfntilde_-(0), & z \in \compC_+,
               \end{cases}
               \label{def:mj}
\end{align}
and
\begin{equation} \label{def:N}
  N^{(j)}(z) = \diag (n^{(j)}_0(z), n^{(j)}_1(z), n^{(j)}_2(z)), \quad z\in\compC \setminus (-\infty,b^2],
\end{equation}
where
\begin{align}\label{def:ni}
  n^{(j)}_0(z) = {}& P^{(j, \infty)}_2(z^{\frac{1}{2}}), & n^{(j)}_1(z) = {}& P^{(j, \infty)}_1(z^{\frac{1}{2}}),  & n^{(j)}_2(z) = {}& P^{(j, \infty)}_1(-z^{\frac{1}{2}}).
\end{align}
We also define
\begin{equation} \label{eq:Npre}
  N^{(\pre)}(z) = \diag (n^{(\pre)}_0(z), n^{(\pre)}_1(z), n^{(\pre)}_2(z)), \quad z\in\compC \setminus \realR,
\end{equation}
where
\begin{equation}
  \begin{gathered}
    n^{(\pre)}_0(z) =
                        \begin{cases}
                          \frac{c^{\frac{\alpha}{3}}}{\sqrt{6}} e^{-\frac{\alpha}{3}\pi i} z^{\frac{\alpha}{3} - \frac{1}{2}}, & z \in \compC_+, \\
                          -\frac{c^{\frac{\alpha}{3}}}{\sqrt{6}} e^{\frac{\alpha}{3}\pi i} z^{\frac{\alpha}{3} - \frac{1}{2}}, & z \in \compC_-,
                        \end{cases} \quad
  n^{(\pre)}_1(z) =
                        \begin{cases}
                          \frac{2c^{\frac{\alpha}{3}}}{\sqrt{6}} e^{\frac{\alpha}{3}\pi i} z^{-\frac{\alpha}{6}}, & z \in \compC_+, \\
                          \frac{2c^{\frac{\alpha}{3}}}{\sqrt{6}} e^{-\frac{\alpha}{3}\pi i} z^{-\frac{\alpha}{6}}, & z \in \compC_-,
                        \end{cases} \\
  n^{(\pre)}_2(z) =  \frac{2c^{\frac{\alpha}{3}}}{\sqrt{6}} z^{-\frac{\alpha}{6}}.
  \end{gathered}
\end{equation}
Below we have separate discussions for the Pearcey regime and the multi-critical regime.

\paragraph{Pearcey regime}

In the Pearcey regime, we assume that $r_n$ shrinks with $n$ such that
\begin{equation}\label{def:rn}
  r_n=n^{-\frac{3}{5}}.
\end{equation}
For $z \in D(0,r_n^2)\setminus \Gamma_{\Phi}$, analogous to \cite[Equation (3.117)]{Wang-Zhang21}, we first introduce a $3 \times 3$ matrix-valued function
\begin{equation}\label{def:sfP0}
  \mathsf{P}^{(j, 0)}(z) = (-c)^j \diag((\rho n)^{-\frac{k}{2}})^2_{k = 0} \Phi^{(\gamma, \alpha, 0, \frac{3}{4}, \tau)}((\rho n)^{\frac{3}{2}} z) N^{(\pre)}(z) N^{(j)}(z)^{-1} e^{nM(z)},
\end{equation}
where
\begin{equation}
  \rho = 2c^{-\frac{4}{3}} (1 - 3c^4),
\end{equation}
$M(z)$, $N^{(j)}(z)$ and $N^{(\pre)}(z)$ are defined in \eqref{def:M}, \eqref{def:N} and \eqref{eq:Npre}, respectively, and $\Phi^{(\gamma, \alpha, 0, \frac{3}{4}, \tau)}$ is defined in RH problem \ref{RHP:model}.

\paragraph{Multi-critical regime}

In the multi-critical regime, we assume that $r_n$ shrinks with $n$ at a slower rate than in the Pearcey regime, such that
\begin{equation}\label{def:rn_crit}
  r_n=n^{-\frac{1}{3}}.
\end{equation}
For $z \in D(0,r_n^2)\setminus \Gamma_{\Phi}$, similar to \eqref{def:sfP0}, we define
\begin{equation} \label{def:sfP0_crit}
  \mathsf{P}^{(j, 0)}(z) = (-c)^j e^{\frac{2}{\sqrt{3}} nz} \diag((3^{\frac{2}{3}} n)^{-\frac{k}{4}})^2_{k = 0} \Phi^{(\gamma, \alpha, -\frac{1}{4}, \frac{3}{2}\sigma, \tau)}((3^{\frac{2}{3}} n)^{\frac{3}{4}} z) N^{(\pre)}(z) N^{(j)}(z)^{-1} e^{nM(z)},
\end{equation}
where $M(z)$, $N^{(j)}(z)$ and $N^{(\pre)}(z)$ are defined in \eqref{def:M}, \eqref{def:N} and \eqref{eq:Npre}, respectively, and $\Phi^{(\gamma, \alpha, -\frac{1}{4}, \frac{3}{2}\sigma, \tau)}$ is defined in RH problem \ref{RHP:model}.

We then have the following proposition regarding the RH problem for $\mathsf{P}^{(j, 0)}$ (defined in either \eqref{def:sfP0} or \eqref{def:sfP0_crit}), which is analogous to \cite[Proposition 3.15]{Wang-Zhang21}. Since the statement of the results are almost identical for both the Pearcey regime and the multi-critical regime, we combine them in one proposition.

\begin{prop}\label{RHP:sfP0}
 In either the Pearcey regime or the multi-critical regime, the function $\mathsf{P}^{(j, 0)}(z)$ defined in \eqref{def:sfP0} (for the Pearcey regime) or \eqref{def:sfP0_crit} (for the multi-critical regime) has the following properties.
  \begin{enumerate} 
  \item
    $ \mathsf{P}^{(j, 0)}(z)$ is analytic in $D(0,r_n^2) \setminus \Gamma_{\Phi}$, where $r_n$ is defined in \eqref{def:rn} for the Pearcey regime and in \eqref{def:rn_crit} for the multi-critical regime.
  \item
    For $z\in D(0,r_n^2) \cap \Gamma_{\Phi}$, we have
    \begin{equation}
      \mathsf{P}^{(j, 0)}_{+}(z) = \mathsf{P}^{(j, 0)}_{-}(z) J^{(j)}_U(z),
    \end{equation}
    where $J^{(j)}_U(z)$ is defined in \eqref{eq:defn_J_U}.
  \item
    For $z\in \partial D(0,r_n^2)$, we have, as $n\to \infty$,
    \begin{equation}\label{eq:sfP0asy}
      \mathsf{P}^{(j, 0)}(z) = \Upsilon(z) \Omega_{\pm} \times
      \begin{cases}
        (I+\bigO(n^{-\frac{1}{10}})), & \text{(Pearcey regime)}, \\
        (I+\bigO(n^{-\frac{1}{36}})), & \text{(multi-critical regime)},
      \end{cases}
    \end{equation}
    where the $\pm$ depends on $z \in \compC_{\pm}$, and $\Upsilon$ and $\Omega_{\pm}$ are defined in \eqref{def:Lpm}.
  \end{enumerate}
\end{prop}

With $\mathsf{P}^{(j, 0)}$ given in \eqref{def:sfP0}, we next define in either the Pearcey regime or the multi-critical regime, analogous to \cite[Equation (3.127)]{Wang-Zhang21},
\begin{equation}\label{def:P0}
  P^{(j, 0)}(z) = \Omega_{\pm}^{-1}\Upsilon(z)^{-1}\mathsf{P}^{(j, 0)}(z), \quad z \in \compC_{\pm} \cap D(0, r^2_n).
\end{equation}
where $\mathsf{P}^{(j, 0)}(z)$ is defined in \eqref{def:sfP0} (resp.~\eqref{def:sfP0_crit}) and $r_n$ is defined in \eqref{def:rn} (resp.~\eqref{def:rn_crit}) in the Pearcey regime (resp.~the multi-critical regime). In view of \eqref{def:P0}, Proposition \ref{RHP:sfP0}, and
\begin{equation}\label{eq:UpsilonOmegajump}
  \Upsilon_+(z) \Omega_+ =  \Upsilon_-(z) \Omega_- \times
  \begin{cases}
    \begin{pmatrix}
      0 & 1 & 0 \\
      1 & 0 & 0 \\
      0 & 0 & 1
    \end{pmatrix},
    & z >0, \\
    \begin{pmatrix}
      1 & 0 & 0 \\
      0 & 0 & 1 \\
      0 & 1 & 0
    \end{pmatrix},
    & z <0,
  \end{cases}
\end{equation}
the following RH problem for $P^{(j, 0)}$ is then immediate, see \cite[Proposition 3.16]{Wang-Zhang21}.
\begin{RHP}\label{rhp:P0}
    In either the Pearcey regime or the multi-critical regime, $P^{(j, 0)}(z)$ is analytic in $D(0,r_n^2) \setminus \Gamma_{\Phi}$, where $r_n$ is defined in \eqref{def:rn} for the Pearcey regime and in \eqref{def:rn_crit} for the multi-critical regime, and satisfies the following:
  \begin{enumerate}
  \item
    For $z\in D(0,r_n^2) \cap \Gamma_{\Phi}$, we have
    \begin{equation}
      P^{(j, 0)}_+(z) =
      \begin{cases}
        P^{(j, 0)}_-(z) J^{(j)}_U(z), & \arg z = \pm \pi/4 \text{ or } \pm 3\pi/4, \\
        \begin{pmatrix}
          0 & 1 & 0 \\
          1 & 0 & 0 \\
          0 & 0 & 1
          \end{pmatrix}
        P^{(j, 0)}_-(z) J^{(j)}_U(z), & z \in (0, r_n^2), \\
        \begin{pmatrix}
          1 & 0 & 0 \\
          0 & 0 & 1 \\
          0 & 1 & 0
        \end{pmatrix}
        P^{(j, 0)}_-(z) J^{(j)}_U(z), & z \in (-r_n^2, 0),
      \end{cases}
    \end{equation}
    where $J^{(j)}_U(z)$ is defined in \eqref{eq:defn_J_U}.
  \item
    For $z\in \partial D(0,r_n^2)$, we have, as $n\to \infty$,
    \begin{equation}\label{eq:P0asy}
      P^{(j, 0)}(z) =
      \begin{cases}
        I+\bigO(n^{-\frac{1}{10}}), & \text{(Pearcey regime)}, \\
        I+\bigO(n^{-\frac{1}{36}}), & \text{(multi-critical regime)}.
      \end{cases}
    \end{equation}
  \end{enumerate}
\end{RHP}

Finally, we set in either the Pearcey regime or the multi-critical regime, analogous to \cite[Equation (3.130)]{Wang-Zhang21} and \cite[Equation (191)]{Wang-Xu25}
\begin{equation} \label{def:V0}
  V^{(n + k, j, 0)}(z) = U^{(n + k, j)}(z) P^{(j, 0)}(z)^{-1}, \quad z \in D(0,r_n^2) \setminus \Gamma_{\Phi},
\end{equation}
where $U^{(n + k, j)}$ is defined in \eqref{eq:defn_U(n+k):1}--\eqref{eq:defn_U(n+k):3} and satisfies RH problem \ref{rhp:U}. Then, by a similar argument as \cite[Proposition 3.17]{Wang-Zhang21},  we have the following proposition.
\begin{prop} \label{rhp:V0}
  In either the Pearcey regime or the multi-critical regime, the function $V^{(n + k, j, 0)}(z)$ defined in \eqref{def:V0} has the following properties.
  \begin{enumerate} 
  \item \label{enu:rhp:V0:1}
    $V^{(n + k, j, 0)}=(V^{(n + k, j, 0)}_0, V^{(n + k, j, 0)}_1, V^{(n + k, j, 0)}_2)$ is analytic in $D(0,r_n^2) \setminus \mathbb{R}$, where $r_n$ is defined in \eqref{def:rn} for the Pearcey regime and in \eqref{def:rn_crit} for the multi-critical regime.
  \item \label{enu:rhp:V0:2}
    For $z\in (-r_n^2,r_n^2)\setminus \{0\}$, we have
    \begin{equation} \label{eq:enu:rhp:V0:2}
      V^{(n + k, j, 0)}_+(z) = V^{(n + k, j, 0)}_-(z) \times
      \begin{cases}
        \begin{pmatrix}
          0 & 1 & 0 \\
          1 & 0 & 0 \\
          0 & 0 & 1
        \end{pmatrix},
        & z \in (0, r_n^2), \\
        \begin{pmatrix}
          1 & 0 & 0 \\
          0 & 0 & 1 \\
          0 & 1 & 0
        \end{pmatrix},
        & z \in (-r_n^2, 0).
      \end{cases}
    \end{equation}
  \item \label{enu:rhp:V0:3}
    For $z\in \partial D(0,r_n^2)$, we have, as $n \to \infty$,
    \begin{equation}\label{eq:V0asy}
      V^{(n + k, j, 0)}(z) = U^{(n + k, j)}(z) \times
      \begin{cases}
        (I+\bigO(n^{-\frac{1}{10}})), & \text{(Pearcey regime)}, \\
        (I+\bigO(n^{-\frac{1}{36}})), & \text{(multi-critical regime)}.
      \end{cases}
    \end{equation}
  \item
    As $z \to 0$, we have
    \begin{equation}\label{eq:V0kzero}
      V^{(n + k, j, 0)}_k(z)=\bigO(1), \qquad k=0, 1, 2.
    \end{equation}
  \end{enumerate}
\end{prop}

\subsection{Final transformation} \label{subsubsec:final_trans_p}


The arguments in this subsubsection apply to both the Pearcey regime and the multi-critical regime. We note that the constant $r_n$ has different definitions in the two regimes: \eqref{def:rn} in the Pearcey regime and \eqref{def:rn_crit} in the multi-critical regime.

Before performing the final transformation, we complete the description of the shapes of $\Sigma_1$ and $\Sigma'_1$. Then $\Sigma_2 = \overline{\Sigma_1}$ and $\Sigma'_2 = \overline{\Sigma'_1}$ are also fixed.

\begin{itemize}
\item
Let $\sigma_0$ and $\sigma_b$ be the intersections of $\Sigma_1$ with $\partial D(0, r_n)$ and $\partial D(b, \epsilon)$, respectively. We let the undetermined part of $\Sigma_1$ be a contour $\Sigma^R_1$ connecting $\sigma_0$ and $\sigma_b$ such that $\Re \phi(z) > 0$ for $z \in \Sigma^R_1$. We may choose $\Sigma^R_1$ to be close to the real axis since, on $(0, b)$, $\Re \phi_+(x) = 0$ and $\phi'_+(x) = -2\pi i \psi(x)$, and $\psi(x)$ is (almost) positive on $(0, b)$ as discussed in Remark \ref{rmk:positivity_of_psi}.
\item
Let $\sigma'_0$ be the intersection of $\Sigma'_1$ with $\partial D(0, r_n)$. We choose $b' \in (b, +\infty)$ and let the undetermined part of $\Sigma'_1$ be a contour $\Sigma^{R, '}_1$ connecting $\sigma'_0$ and $b'$ such that $\Re (\gfn(-z) - \gfn(z) - 4az) < 0$ on $\Sigma^{R, '}_1$ and $\Re (\gfntilde(x) + \gfn(-x) - V(-x) - \ell) < 0$ on $(b', + \infty)$. A possible construction of $\Sigma^{R, '}_1$ is the combination of the line segment from $\sigma'_0$ to $\sigma''_1 = iC + \Re(\sigma'_0)$, where $C$ is a sufficiently large constant, and the arc from $\sigma''_1$ to $\lvert \sigma''_1 \rvert$, where the arc is part of the circle $\{ \lvert z \rvert = \lvert \sigma''_1 \rvert \}$.
\end{itemize}
We set, similar to \cite[Equation 3.137)]{Wang-Zhang21} and \cite[Equation (195)]{Wang-Xu25}
\begin{equation}\label{def:sigmaR}
 \Sigma^{R}:=[0,b]\cup [b+\epsilon, \infty) \cup \partial D(0,r_n) \cup \partial D(b,\epsilon) \cup \Sigma_1^R \cup \Sigma_2^R \cup \Sigma^{R, '}_1 \cup \Sigma^{R, '}_2,
\end{equation}
where $\Sigma^R_1$ and $\Sigma^{R, '}_1$ are specified above, and $\Sigma^R_2 = \overline{\Sigma^R_1}$ and $\Sigma^{R, '}_2 = \overline{\Sigma^{R, '}_1}$.
We also divide the open disk $D(0,r_n)$ into $2$ parts by setting, as in \cite[Equation 3.139)]{Wang-Zhang21} and \cite[Equation (196)]{Wang-Xu25}
\begin{align}\label{def:Wk}
  W_1 = {}& \{z\in D(0,r_n)\setminus\{0\} \mid  \arg z \in (-\frac{\pi}{2}, \frac{\pi}{2}) \}, & W_2 = {}& \{z\in D(0,r_n)\setminus\{0\} \mid  \arg (-z) \in (-\frac{\pi}{2}, \frac{\pi}{2}) \}.
\end{align}
We then define a $1\times 2$ vector-valued function $R^{(n + k, j)}(z) = (R^{(n + k, j)}_1(z), R^{(n + k, j)}_2(z))$ such that $R^{(n + k, j)}_1(z)$ is analytic in $\compC \setminus \Sigma^R$ and $R^{(n + k, j)}_2(z)$ is analytic in $\halfH \setminus \Sigma^R$, analogous to \cite[Equations (3.141) and (3.142)]{Wang-Zhang21} and \cite[Equations (197) and (198)]{Wang-Xu25}:
\begin{align}
  R^{(n + k, j)}_1(z) = {}&
  \begin{cases}
    Q^{(n + k, j)}_1(z), & \hbox{$z\in \mathbb{C} \setminus \{D(b,\epsilon) \cup D(0,r_n) \cup \Sigma^R \}$,} \\
    V_1^{(n + k, j, b)}(z), & \hbox{$z\in D(b,\epsilon) \setminus [b-\epsilon, b]$,} \\
    V_1^{(n + k, j, 0)}(z^2), & \hbox{$z\in W_1 \setminus [0,r_n]$,} \\
    V_2^{(n + k, j, 0)}(z^2), & \hbox{$z\in W_2$,}
  \end{cases} \label{def:R1} \\
  R^{(n + k, j)}_2(z) = {}&
  \begin{cases}
    Q^{(n + k, j)}_2(z), & \hbox{$z\in \halfH \setminus \{D(b,\epsilon) \cup D(0,r_n) \cup \Sigma^R \}$,} \\
    V_2^{(n + k, j, b)}(z), & \hbox{$z\in D(b,\epsilon) \setminus [b-\epsilon, b]$,} \\
    V_0^{(n + k, j, 0)}(z^2), & \hbox{$z \in W_1 \setminus [0,r_n]$.}
  \end{cases} \label{def:R2}
\end{align}
Recall the function $J_c$ defined in \eqref{eq:defn_J_c} and its inverse functions $I_1$ and $I_2$ on $D \setminus [-1, 0]$ and $\compC \setminus \overline{D}$ respectively, see the paragraph below \eqref{eq:defn_J_c}. We define, as in \cite[Equation 3.146)]{Wang-Zhang21} and \cite[Equation (199)]{Wang-Xu25}
\begin{equation} \label{eq:scalar_R_defn}
  \tR^{(n + k, j)}(s) =
  \begin{cases}
    R^{(n + k, j)}_1(J_c(s)), & \text{$s \in \compC \setminus \overline{D}$ and $s \notin I_1(\Sigma^R)$,} \\
    R^{(n + k, j)}_2(J_c(s)), & \text{$s \in D \setminus[-1, 0]$ and $s \notin I_2(\Sigma^R)$.}
  \end{cases}
\end{equation}
We have that $\tR^{(n + k, j)}(s)$ satisfies a shifted RH problem stated below, which is analogous to \cite[RH problem 3.19]{Wang-Zhang21}.

We are now at the stage of describing the RH problem for $\tR$. For that purpose, we define
\begin{equation}\label{def:SigmaRi}
  \begin{aligned}
    \SigmaR^{(0)} := {}& I_1(\Sigma^{R, '}_1 \cup \Sigma^{R, '}_2) \subseteq \compC \setminus \overline{D}, & \SigmaR^{(0 '')} := {}& I_1((-\Sigma^{R, '}_1) \cup (-\Sigma^{R, '}_2)) \subseteq \compC \setminus \overline{D}, \\
    \SigmaR^{(1)} := {}& I_1(\Sigma^R_1 \cup \Sigma^R_2) \subseteq \compC \setminus \overline{D}, & \SigmaR^{(1')} := {}& I_2(\Sigma^R_1 \cup \Sigma^R_2) \subseteq D, \\
    \SigmaR^{(2)} := {}& I_1((b + \epsilon, b')) \subseteq \compC \setminus \overline{D}, & \SigmaR^{(2')} := {}& I_2((b + \epsilon, b')) \subseteq D, \\
    \SigmaR^{(\twoandhalf)} := {}& I_1((b', +\infty)) \subseteq \compC \setminus \overline{D}, & \SigmaR^{(\twoandhalf')} := {}& I_2((b', +\infty)) \subseteq D, \\
    \SigmaR^{(\twoandhalf'')} := {}& I_1((-\infty, -b')) \subseteq \compC \setminus \overline{D}, && \\
    \SigmaR^{(3)} := {}& I_1(\partial D(b,\epsilon)) \subseteq \compC \setminus \overline{D}, & \SigmaR^{(3')} := {}& I_2(\partial D(b,\epsilon)) \subseteq D, \\
    \SigmaR^{(4)} := {}& I_2(\Gamma_1) \subseteq D, & \SigmaR^{(5)}_k := {}& I_1(\Gamma_k) \subseteq \compC \setminus \overline{D}, \qquad k = 1, 2,
  \end{aligned}
\end{equation}
and set
\begin{equation}\label{def:SigmaR}
  \SigmaR = \SigmaR^{(5)}_1 \cup \SigmaR^{(5)}_2 \cup \left( \bigcup \SigmaR^{(*)} \right), \quad * = 0, 0'', 1, 1', 2, 2', \twoandhalf, \twoandhalf', \twoandhalf'', 3, 3', 4,
\end{equation}
which is the union of the curves in Figure \ref{fig:jump_R_scalar}. We also define the following functions on each curve constituting $\SigmaR$:
\begin{align}
  J^{(j)}_{\SigmaR^{(0)}}(s) = {}& \frac{P^{(j, \infty)}_1(-z)}{P^{(j, \infty)}_1(z)} e^{n(\gfn(-z) - \gfn(z) - 4az)}, & s \in {}& \SigmaR^{(0)},\label{def:Jsigma0} \\
  \intertext{with $z = J_c(s) \in \Sigma^{R, '}_1 \cup \Sigma^{R, '}_2$,}
  J^{(j)}_{\SigmaR^{(1)}}(s) = {}& 2z^{1 - \alpha} \frac{P^{(j, \infty)}_2(z)}{P^{(j, \infty)}_1(z)} e^{-n\phi(z)}, & s \in {}& \SigmaR^{(1)}, \label{def:Jsigma1} \\
  \intertext{with $z = J_c(s) \in \Sigma^R_1 \cup \Sigma^R_2$,}
  J^{(j)}_{\SigmaR^{(2')}}(s) = {}& \frac{1}{2 z^{1 - \alpha}} \frac{P^{(j, \infty)}_1(z)}{P^{(j, \infty)}_2(z)} e^{n\phi(z)}, & s \in {}& \SigmaR^{(2')}, \label{def:Jsigma2'} \\
  \intertext{with $z = J_c(s) \in (b + \epsilon, b')$,}
  J^{(j)}_{\SigmaR^{(\twoandhalf')}}(s) = {}& \frac{1}{2 z^{1 - \alpha}} \frac{P^{(j, \infty)}_1(-z)}{P^{(j, \infty)}_2(z)} e^{n(\gfntilde(z) + \gfn(-z) - V(-z) - \ell)}, & s \in {}& \SigmaR^{(\twoandhalf')}, \label{def:Jsigma2half'} \\
  \intertext{with $z = J_c(s) \in (b', +\infty)$,}
  J^{(j), 1}_{\SigmaR^{(3)}}(s) = {}& P^{(j, b)}_{11}(z) - 1, \quad J^{(j), 2}_{\SigmaR^{(3)}}(s) = P^{(j, b)}_{21}(z), & s \in {}& \SigmaR^{(3)}, \label{def:Jsigma3} \\
  \intertext{with $z = J_c(s) \in \partial D(b,\epsilon)$,}
  J^{(j), 1}_{\SigmaR^{(3')}}(s) = {}& P^{(j, b)}_{22}(z) - 1, \quad J^{(j), 2}_{\SigmaR^{(3')}}(s) = P^{(j, b)}_{12}(z), & s \in {}& \SigmaR^{(3')},
  \label{def:Jsigma32}\\
  \intertext{with $z = J_c(s) \in \partial D(b,\epsilon)$,}
  J^{(j), 0}_{\SigmaR^{(4)}}(s) = {}& P^{(j, 0)}_{00}(z^2) - 1, \quad J^{(j), j}_{\SigmaR^{(4)}}(s) = P^{(j, 0)}_{j0}(z^2), & s \in {}& \SigmaR^{(4)}, \label{def:Jsigma4} \\
  \intertext{with $z = J_c(s) \in \Gamma_1$ and $l=1, 2$,}
  J^{(j), k}_{\SigmaR^{(5)}_k}(s) = {}& P^{(j, 0)}_{kk}(z^2) - 1, \quad J^{(j), l}_{\SigmaR^{(5)}_k}(s) = P^{(j, 0)}_{lk}(z^2), \quad l\neq k, & s \in {}& \SigmaR^{(5)}_k, \label{def:Jsigma5}
\end{align}
where $z = J_c(s) \in \Gamma_k$, $k=1, 2$, and $j=0, 1, 2$. In \eqref{def:Jsigma0}, \eqref{def:Jsigma1}, \eqref{def:Jsigma2'} and \eqref{def:Jsigma2half'}, $P^{(j, \infty)}_1(z)$ and $P^{(j, \infty)}_2(z)$ are defined in \eqref{eq:defn_P^jinfty_1} and \eqref{eq:defn_P^jinfty_2}. In \eqref{def:Jsigma3} and \eqref{def:Jsigma32}, $P^{(j, b)}(z)=(P^{(j, b)}_{jk}(z))_{j,k=1}^2$ is defined in \eqref{eq:defn_P^b}. In \eqref{def:Jsigma4} and \eqref{def:Jsigma5}, $P^{(j, 0)}(z)=(P^{(j, 0)}_{jk}(z))_{j,k=1}^2$ is defined in \eqref{def:P0}. With the aid of these functions, we further define an operator $\DeltaR^{(j)}$ that acts on functions defined on $\SigmaR$ by
\begin{multline}\label{def:DeltaR}
  \DeltaR^{(j)} f(s) = \\
  \begin{cases}
    c J^{(j)}_{\SigmaR^{(0)}}(s) f(s'), & \text{$s \in \SigmaR^{(0)}$ and $s' = I_1(-J_c(s))$}, \\
    J^{(j)}_{\SigmaR^{(1)}}(s) f(\s), & \text{$s \in \SigmaR^{(1)}$ and $\s = I_2(J_c(s))$}, \\
    J^{(j)}_{\SigmaR^{(2')}}(s) f(\s), & \text{$s \in \SigmaR^{(2')}$ and $\s = I_1(J_c(s))$}, \\
    J^{(j)}_{\SigmaR^{(\twoandhalf')}}(s) f(\s) + J^{(j)}_{\SigmaR^{(\twoandhalf')}}(s) f(s'), & \text{$s \in \SigmaR^{(\twoandhalf')}$, $\s = I_1(J_c(s))$ and $s' = I_1(-J_c(s))$}, \\
    J^{(j), 1}_{\SigmaR^{(3)}}(s) f(s) + c J^{(j), 2}_{\SigmaR^{(3)}}(s) f(\s), & \text{$s \in \SigmaR^{(3)}$ and $\s = I_2(J_c(s))$}, \\
    J^{(j), 1}_{\SigmaR^{(3')}}(s) f(s) + J^{(j), 2}_{\SigmaR^{(3')}}(s) f(\s), & \text{$s \in \SigmaR^{(3')}$ and $\s = I_1(J_c(s))$}, \\
    J^{(j), 0}_{\SigmaR^{(4)}}(s) f(s) + \sum\limits ^2_{l = 1} J^{(j), l}_{\SigmaR^{(4)}}(s) f(\s_l), & \text{$s \in \SigmaR^{(4)}$ and $\s_l = I_1(J_c(s) e^{(l - 1)\pi i}) \in I_1(\Gamma_l)$}, \\
    \sum\limits^2_{l = 0} J^{(j), l}_{\SigmaR^{(5)}_k}(s) f(\s_l), & \text{$s \in \SigmaR^{(5)}_k$ and  $\s_0 = I_2(J_c(s) e^{(1 - k)\pi i}) \in D$,} \\
    & \text{  $\s_l = I_1(J_c(s) e^{(l - k)\pi i}) \in I_1(\Gamma_l) \subseteq \compC \setminus \overline{D}$} \\
                                & \text{for $l = 1, 2$,  such that $\s_k = s$}, \\
    0, & s \in \SigmaR^{(0'')}\cup\SigmaR^{(1')} \cup \SigmaR^{(2)} \cup \SigmaR^{(\twoandhalf)} \cup \SigmaR^{(\twoandhalf'')}, \\
  \end{cases}
\end{multline}
where $f$ is a complex-valued function defined on $\SigmaR$. Hence, we can define a scalar shifted RH problem as follows.

\begin{figure}[htb]
  \centering
  \includegraphics{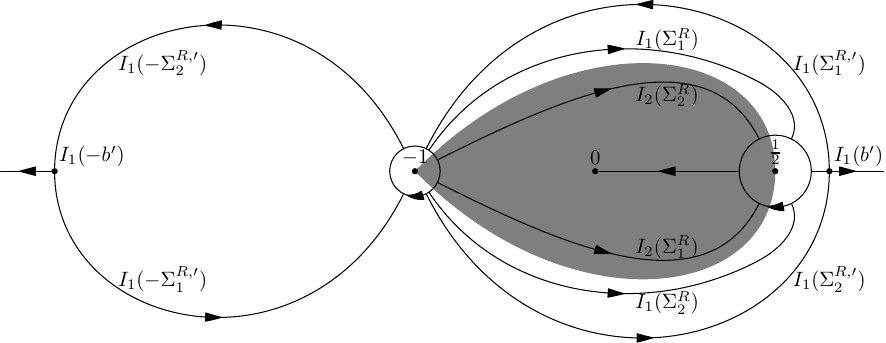}
  \caption{The contour $\SigmaR$ for the RH problem for $\tR^{(n + k, j)}$. The curves $I_1(\Sigma^R_1)$, $I_1(\Sigma^R_2)$, $I_2(\Sigma^R_1)$, $I_2(\Sigma^R_2)$, $I_1(-\Sigma^{R, '}_1)$ and $I_1(-\Sigma^{R, '}_2)$ are labelled.}
  \label{fig:jump_R_scalar}
\end{figure}

\begin{RHP} \label{rhp:tR}
  $\tR^{(n + k, j)}(s)$ is analytic in $\compC \setminus \SigmaR$, where the contour $\SigmaR$ is defined in \eqref{def:SigmaR}, and satisfies the following:
  \begin{enumerate}
  \item \label{enu:rhp:tR:2}
    For $s\in \SigmaR$, we have
    \begin{equation}
      \tR^{(n + k, j)}_+(s) - \tR^{(n + k, j)}_-(s) = \DeltaR^{(j)} \tR^{(n + k, j)}_-(s),
    \end{equation}
    where $\DeltaR^{(j)}$ is the operator defined in \eqref{def:DeltaR}.
  \item
    As $s \to \infty$, we have
    \begin{equation}
      \tR^{(n + k, j)}(s)= (cs)^{k - j}(1+\bigO(s^{-1})).
    \end{equation}
  \item
    As $s \to 0$, we have $\tR^{(n + k, j)}(s)=\bigO(s^{k - j})$.
  \end{enumerate}
\end{RHP}

By arguments analogous to \cite[Propositions 3.20 and 3.21, and Lemma 3.22]{Wang-Zhang21},  we derive the following estimate for $\tR^{(n + k, k)}(s)$.
\begin{lemma} \label{lem:tRest}
  As $n\to \infty$, for any $k \in \intZ$, we have
  \begin{equation} \label{eq:esttR}
    \tR^{(n + k, k)}(s)=1+\bigO(n^{-C_R}), 
  \end{equation}
  where $C_R = \frac{1}{10}$ in the Pearcey regime and $C_R =\frac{1}{36}$ in the multi-critical regime.  The error term is uniform for $\lvert s+1 \rvert < \epsilon r_n^{2/3}$ or $\lvert s \rvert < \epsilon$, where $\epsilon$ is a small positive constant and $r_n$ is given in \eqref{def:rn}.
\end{lemma}

The proof of the lemma is the same as that of \cite[Proposition 3.22]{Wang-Zhang21}, which is detailed in \cite[Section 3.7]{Wang-Zhang21}, and the idea is as follows: First, we show that $\tR^{(n + k, k)}(s)$ is the unique solution to this RH problem, as in \cite[Proposition 3.20]{Wang-Zhang21}. Then, by a small norm argument, we show that the solution to the RH problem that $\tR^{(n + k, k)}(s)$ satisfies is close to $1$ in the sense of \cite[Proposition 3.21]{Wang-Zhang21}. Finally, we derive the lemma as the counterpart of \cite[Proposition 3.22]{Wang-Zhang21}. Since the proof is lengthy and identical to that in \cite[Section 3.7]{Wang-Zhang21} for a slightly different setting, we omit it.
We note that the essential ingredient of the small norm argument is that as $n \to \infty$, $\lvert J^{(j), *}_*(s) \rvert = \bigO(n^{-C_R})$ uniformly, where $J^{(j), *}_*(s)$ stands for any of the functions defined in \eqref{def:Jsigma0} -- \eqref{def:Jsigma5}.

If $k \neq j$, the argument in the proof of \cite[Proposition 3.22]{Wang-Zhang21} does not apply directly because the asymptotics of $\tR^{(n + k, j)}(s)$ at $\infty$ and $0$ are different. However, from the definition, $\tR^{(n + k, j)}(s)$ can be derived from $\tR^{(n + k, k)}(s)$. For example, with $\epsilon > 0$ as in Lemma \ref{lem:tRest},
\begin{equation}
  \tR^{(n + k, j)}(s) = (-c)^{k - j} \tR^{(n + k, k)}(s), \quad \text{if $\lvert s+1 \rvert < \epsilon r_n^{2/3}$}.
\end{equation}

Let
\begin{equation} \label{eq:defn_C_10}
  \hat{\tR}^{(n + k, j)}(s) = \frac{1}{c(s + 1)} \left( \tR^{(n + k + 1, j)}(s) + c^{(n + k, j)}_{10} \tR^{(n + k, j)}(s) \right), \quad c^{(n + k, j)}_{10} = -\frac{\tR^{(n + k + 1, j)}(-1)}{\tR^{(n + k, j)}(-1)}.
\end{equation}
Then we find that $\hat{\tR}^{(n + k, k)}(s)$ satisfies a shifted RH problem similar to RH problem \ref{rhp:tR}:

\begin{RHP} \label{rhp:tR_hat}
  $\hat{\tR}^{(n + k, j)}(s)$ is analytic in $\compC \setminus \SigmaR$, where the contour $\SigmaR$ is defined in \eqref{def:SigmaR}, and satisfies the following:
  \begin{enumerate}
  \item
    For $s\in \SigmaR$, we have
    \begin{equation}
      \hat{\tR}^{(n + k, j)}_+(s) - \hat{\tR}^{(n + k, j)}_-(s) = \DeltaRhat^{(j)} \tR^{(n + k, j)}_-(s),
    \end{equation}
    where $\DeltaRhat^{(j)}$ is the operator defined as
    \begin{multline}
      \DeltaRhat^{(j)} f(s) = \frac{1}{c(s + 1)} \\
      \begin{cases}
        c J^{(j)}_{\SigmaR^{(0)}}(s) \hat{f}(s'), & \text{$s \in \SigmaR^{(0)}$ and $s' = I_1(-J_c(s))$}, \\
        J^{(j)}_{\SigmaR^{(1)}}(s) \hat{f}(\s), & \text{$s \in \SigmaR^{(1)}$ and $\s = I_2(J_c(s))$}, \\
        J^{(j)}_{\SigmaR^{(2')}}(s) \hat{f}(\s), & \text{$s \in \SigmaR^{(2')}$ and $\s = I_1(J_c(s))$}, \\
        J^{(j)}_{\SigmaR^{(\twoandhalf')}}(s) \hat{f}(\s) + J^{(j)}_{\SigmaR^{(\twoandhalf')}}(s) \hat{f}(s'), & \text{$s \in \SigmaR^{(\twoandhalf')}$, $\s = I_1(J_c(s))$ and $s' = I_1(-J_c(s))$}, \\
        J^{(j), 1}_{\SigmaR^{(3)}}(s) \hat{f}(s) + c J^{(j), 2}_{\SigmaR^{(3)}}(s) \hat{f}(\s), & \text{$s \in \SigmaR^{(3)}$ and $\s = I_2(J_c(s))$}, \\
        J^{(j), 1}_{\SigmaR^{(3')}}(s) \hat{f}(s) + J^{(j), 2}_{\SigmaR^{(3')}}(s) \hat{f}(\s), & \text{$s \in \SigmaR^{(3')}$ and $\s = I_1(J_c(s))$}, \\
        J^{(j), 0}_{\SigmaR^{(4)}}(s) \hat{f}(s) + \sum\limits ^2_{l = 1} J^{(j), l}_{\SigmaR^{(4)}}(s) \hat{f}(\s_l), & \text{$s \in \SigmaR^{(4)}$ and $\s_l = I_1(J_c(s) e^{(l - 1)\pi i}) \in I_1(\Gamma_l)$}, \\
        \sum\limits^2_{l = 0} J^{(j), l}_{\SigmaR^{(5)}_k}(s) \hat{f}(\s_l), & \text{$s \in \SigmaR^{(5)}_k$ and  $\s_0 = I_2(J_c(s) e^{(1 - k)\pi i}) \in D$,} \\
                                            & \text{$\s_l = I_1(J_c(s) e^{(l - k)\pi i}) \in I_1(\Gamma_l) \subseteq \compC \setminus \overline{D}$} \\
                                            & \text{for $l = 1, 2$,  such that $\s_k = s$}, \\
        0, & s \in \SigmaR^{(0'')}\cup\SigmaR^{(1')} \cup \SigmaR^{(2)} \cup \SigmaR^{(\twoandhalf)} \cup \SigmaR^{(\twoandhalf'')}, \\
      \end{cases}
    \end{multline}
    where $f$ is a complex-valued function defined on $\SigmaR$ and $\hat{f}(s) = c(s + 1)f(s)$.
  \item
    As $s \to \infty$, we have
    \begin{equation}
      \hat{\tR}^{(n + k, j)}(s)= (cs)^{k - j} (1+\bigO(s^{-1})).
    \end{equation}
  \item \label{enu:rhp:tR:4}
    As $s \to 0$, we have $\hat{\tR}^{(n + k, k)}(s)=\bigO(s^{k - j})$.
  \end{enumerate}
\end{RHP}

Then by the same argument as in the proof of Lemma \ref{lem:tRest}, we have
\begin{lemma} \label{lem:tRest_hat}
  As $n\to \infty$, for any $k \in \intZ$, we have
  \begin{equation} \label{eq:esttR_hat}
    \hat{\tR}^{(n + k, k)}(s)=1+\bigO(n^{-C_R}),
  \end{equation}
  uniformly for $\lvert s+1 \rvert < \epsilon r_n^{2/3}$ or $\lvert s \rvert < \epsilon$, where $\epsilon$ is a small positive constant, $r_n$ is given in \eqref{def:rn} for the Pearcey regime and in \eqref{def:rn_crit} for the multi-critical regime, and $C_R$ is the same as in \eqref{eq:esttR}.
\end{lemma}

Finally, we define
\begin{equation}
  \doublehat{\tR}^{(n + k, j)}(s) = \frac{1}{c(s + 1)} \left( \hat{\tR}^{(n + k + 1, j)}(s) - \hat{\tR}^{(n + k, j)}(s) \frac{\hat{\tR}^{(n + k + 1, j)}(-1)}{\hat{\tR}^{(n + k, j)}(-1)} \right),
\end{equation}
and note that there exist constants $c^{(n + k, j)}_{20}$ and $c^{(n + k, j)}_{21}$ such that
\begin{equation} \label{eq:defn_C_20_21}
  \doublehat{\tR}^{(n + k, j)}(s) = \frac{1}{c^2(s + 1)^2} \left( \tR^{(n + k + 2, j)}(s) + c^{(n + k, j)}_{21} \tR^{(n + k + 1, j)}(s) + c^{(n + k, j)}_{20} \tR^{(n + k, j)}(s) \right).
\end{equation}
Then we find that $\doublehat{\tR}^{(n + k, k)}(s)$ satisfies a shifted RH problem similar to RH problems \ref{rhp:tR} and \ref{rhp:tR_hat}:

\begin{RHP}
  $\doublehat{\tR}^{(n + k, j)}(s)$ is analytic in $\compC \setminus \SigmaR$, where the contour $\SigmaR$ is defined in \eqref{def:SigmaR}, and satisfies the following:
  \begin{enumerate}
  \item
    For $s\in \SigmaR$, we have
    \begin{equation}
      \doublehat{\tR}^{(n + k, j)}_+(s) - \doublehat{\tR}^{(n + k, j)}_-(s) = \DeltaRdoublehat^{(j)} \tR^{(n + k, j)}_-(s),
    \end{equation}
    where $\DeltaRdoublehat^{(j)}$ is the operator defined as
    \begin{equation}
      \DeltaRdoublehat^{(j)} f(s) = \frac{1}{c^2(s + 1)^2} \DeltaR(c^2(s + 1)^2 f(s)).
    \end{equation}
  \item
    As $s \to \infty$, we have
    \begin{equation}
      \doublehat{\tR}^{(n + k, j)}(s)= (cs)^{k - j} (1+\bigO(s^{-1})).
    \end{equation}
  \item
    As $s \to 0$, we have $\doublehat{\tR}^{(n + k, k)}(s)=\bigO(s^{k - j})$.
  \end{enumerate}
\end{RHP}
Similar to Lemmas \ref{lem:tRest} and \ref{lem:tRest_hat}, we have
\begin{lemma} \label{lem:tRest_doublehat}
  As $n\to \infty$, for any $k \in \intZ$, we have
  \begin{equation} \label{eq:esttR_doublehat}
    \doublehat{\tR}^{(n + k, k)}(s)=1+\bigO(n^{-C_R}),
  \end{equation}
  uniformly for $\lvert s+1 \rvert < \epsilon r_n^{2/3}$ or $\lvert s \rvert < \epsilon$, where $\epsilon$ is a small positive constant, $r_n$ is given in \eqref{def:rn} for the Pearcey regime and in \eqref{def:rn_crit} for the multi-critical regime, and $C_R$ is the same as in \eqref{eq:esttR}.
\end{lemma}

\section{Proof of Theorem \ref{thm:main}} \label{sec:proof_main}

\subsection{Asymptotics of $Y^{(n + k)}$, \texorpdfstring{$\hat{Y}^{(n + k)}$}{Y^hat}, \texorpdfstring{$\doublehat{Y}^{(n + k)}$}{Y^doublehat}, and $X^{(n)}$}

In this subsection, we consider the Pearcey regime and the multi-critical regime at the same time, unless otherwise specified. Recall that $r_n$ is defined in \eqref{def:rn} and \eqref{def:rn_crit} in the two regimes separately, and $C_R$ is defined in \eqref{eq:esttR} in the two regimes.

Recall the vector-valued function $R^{(n + k, j)}(z)$ in \eqref{def:R1} and \eqref{def:R2}, $c^{(n + k, j)}_{10}$ defined by \eqref{eq:defn_C_10}, and $c^{(n + k, j)}_{21}, c^{(n + k, j)}_{20}$ defined in \eqref{eq:defn_C_20_21}. We define
\begin{align}
  \hat{R}^{(n + k, j)}(z) = {}& R^{(n + k + 1, j)}(z) + c^{(n + k, j)}_{10} R^{(n + k, j)}(z), \\
  \doublehat{R}^{(n + k, j)}(z) = {}& R^{(n + k + 2, j)}(z) + c^{(n + k, j)}_{21} R^{(n + k + 1, j)}(z) + c^{(n + k, j)}_{20} R^{(n + k, j)}(z).
\end{align}
Then let $\epsilon > 0$ be a small positive constant, and we define, for $0 < \lvert s \rvert < \epsilon r^{2/3}_n$ which is $\epsilon n^{-2/5}$ in the Pearcey regime or $\epsilon n^{-2/9}$ in the multi-critical regime,
\begin{equation}
  \choice{v}(s) =
  \begin{cases}
    \choice{R}^{(n - 1, -1)}_2(e^{-\frac{\pi i}{2}} s^{\frac{3}{2}}), & \arg s \in (0, \frac{\pi}{3}), \\
    \choice{R}^{(n - 1, -1)}_2(e^{\frac{\pi i}{2}} s^{\frac{3}{2}}), & \arg s \in (-\frac{\pi}{3}, 0), \\
    \choice{R}^{(n - 1, -1)}_1(-(-s)^{\frac{3}{2}}), & \arg (-s) \in (-\frac{2\pi}{3}, \frac{2\pi}{3}),
  \end{cases}
  \quad \text{where $\bullet = \hat{\,}$, $\doublehat{\,}$ or empty}. 
\end{equation}
We have that $R^{(n + k, j)}(z)$ is expressed in terms of $\tR^{(n + k, j)}(s)$ as in \eqref{eq:scalar_R_defn}, $\hat{R}^{(n + k, j)}(z)$ is expressed in terms of $\hat{\tR}^{(n + k, j)}(s)$ as in \eqref{eq:defn_C_10}, and $\doublehat{R}^{(n + k, j)}(z)$ is expressed in terms of $\doublehat{\tR}^{(n + k, j)}(s)$ as in \eqref{eq:defn_C_20_21}. By the limit formulas \eqref{eq:esttR}, \eqref{eq:esttR_hat} and \eqref{eq:esttR_doublehat} for $\tR^{(n + k, k)}(s)$, $\hat{\tR}^{(n + k, k)}(s)$ and $\doublehat{\tR}^{(n + k, k)}(s)$, respectively, we have that, with the parameter $c \in (0, 3^{-1/4})$ in the Pearcey regime and $c$ defined in \eqref{eq:c_multi_crit} in the multi-critical regime,
\begin{align}
  v(s) = {}& 1 + \bigO(n^{-C_R}), & \hat{v}(s) = {}& c^{\frac{1}{3}} s (1 + \bigO(n^{-C_R})), & \doublehat{v}(s) = {}&  c^{\frac{2}{3}} s^2 (1 + \bigO(n^{-C_R})).
\end{align}
Now we write 

\begin{equation}
  \choice{v}_k(s) = (-1)^k \frac{1}{2\pi i} \oint_{\lvert \zeta \rvert = \epsilon r^{2/3}_n} \frac{\choice{v}(\zeta)}{\zeta^{k + 1}(1 - s/\zeta^3)} d\zeta, \quad k = 0, 1, 2, \quad \text{$\bullet = \hat{\,}$, $\doublehat{\,}$ or empty}.
\end{equation}
We note that $v(s)$, $\hat{v}(s)$ and $\doublehat{v}(s)$ are analytic in the region $D(0, \epsilon r^{2/3}_n)$, such that
\begin{align}
  v(s) = {}& \sum_{k=0}^{\infty}c_ks^k, & \hat{v}(s) = {}& s\sum_{k=0}^{\infty}\hat{c}_ks^k, & \doublehat{v}(s) = {}& s^2\sum_{k=0}^{\infty}\doublehat{c}_ks^k,
\end{align}
and the coefficients have the estimates
\begin{align}
  c_0 = {}& 1 + \bigO(n^{-C_R}), & \hat{c}_0 = {}& c^{\frac{1}{3}} + \bigO(n^{-C_R}), & \doublehat{c}_0 = {}& c^{\frac{2}{3}} + \bigO(n^{-C_R}),
\end{align}
\begin{equation}
  \choice{c}_k =
  \begin{cases}
    \bigO(n^{\frac{2k-1/2}{5}}), & \text{(Pearcey regime)}, \\
    \bigO(n^{\frac{2k-1/4}{9}}), & \text{(multi-critical regime)},
  \end{cases}
  \quad \text{$k \geq 1$ and $\bullet = \hat{\,}$, $\doublehat{\,}$ or empty}.
\end{equation}
Hence,
\begin{align}
  v_0(z) = {}& c_0 + c_3 z + \bigO(z^2), & v_1(z) = {}& c_1 + c_4 z + \bigO(z^2), & v_2(z) = {}& c_2 + c_5 z + \bigO(z^2), \label{eq:v_0v_1v_2} \\
  \hat{v}_0(z) = {}& \hat{c}_2 z + \bigO(z^2), & \hat{v}_1(z) = {}& \hat{c}_0 + \hat{c}_3 z + \bigO(z^2), & \hat{v}_2(z) = {}& \hat{c}_1 + \hat{c}_4 z + \bigO(z^2), \\
  \doublehat{v}_0(z) = {}& \doublehat{c}_1 z + \bigO(z^2), & \doublehat{v}_1(z) = {}& \doublehat{c}_2 z + \bigO(z^2), & \doublehat{v}_2(z) = {}& \doublehat{c}_0 + \doublehat{c}_3 z + \bigO(z^2). \label{eq:v_0v_1v_2_doublehat}
\end{align}
Let $\choice{U}(z) = (\choice{U}_0(z), \choice{U}_1(z), \choice{U}_2(z))$ be
\begin{equation}
  \choice{U}_k(z) = \sum^2_{i = 0}  \choice{v}_i(-z) \mathsf{P}^{(-1, 0)}_{ik}(z), \quad \text{$k = 0,1,2$ and $\bullet = \hat{\,}$, $\doublehat{\,}$ or empty},
\end{equation}
where $\mathsf{P}^{(-1, 0)}(z)$ is defined in \eqref{def:sfP0} in the Pearcey regime and in \eqref{def:sfP0_crit} in the multi-critical regime. By the relation between $\mathsf{P}^{(j, 0)}(z)$ and $P^{(j, 0)}(z)$ given in \eqref{def:P0}, and the relation \eqref{def:R1} and \eqref{def:R2} between $R^{(n + k, j)}$ and $V^{(n + k, j)}$, we have that
\begin{equation}
  U_k(z) = \sum^2_{i = 0} V^{(n - 1, -1)}_i(-z) \mathsf{P}^{(-1, 0)}_{ik}(z) = U^{(n - 1, -1)}_k(z),
\end{equation}
and more generally,
\begin{equation}\label{eq:U-R}
  \begin{pmatrix}
    U_0(z) & U_1(z) & U_2(z) \\
    \hat{U}_0(z) & \hat{U}_1(z) & \hat{U}_2(z) \\
    \doublehat{U}_0(z) & \doublehat{U}_1(z) & \doublehat{U}_2(z)
  \end{pmatrix}
  = 
 R(z)
  \mathsf{P}^{(-1, 0)}(z), \quad R(z)= \begin{pmatrix}
    v_0(z) & v_1(z) & v_2(z) \\
    \hat{v}_0(z) & \hat{v}_1(z) & \hat{v}_2(z) \\
    \doublehat{v}_0(z) & \doublehat{v}_1(z) & \doublehat{v}_2(z)
  \end{pmatrix}.
\end{equation}

Then $U(z)$ is related to $Y^{(n - 1, -1)}(z)$ by \eqref{eq:defn_U(n+k):1}--\eqref{eq:defn_U(n+k):3} via $Q^{(n - 1, -1)}(z)$, and $\hat{U}(z)$ and $\doublehat{U}(z)$ are related to $\hat{Y}^{(n - 1, -1)}(z)$ and $\doublehat{Y}^{(n - 1, -1)}(z)$ in the same way. If we consider the region $\{ z \in D(0, r^2_n) : \arg z \in (0, \pi/2) \}$, we have ($\bullet = \hat{\,}$, $\doublehat{\,}$ or empty)
\begin{multline}
  (\choice{Y}^{(n - 1)}_1(\sqrt{z}), \choice{Y}^{(n - 1)}_1(-\sqrt{z}), \choice{Y}^{(n - 1)}_2(\sqrt{z})) = (\choice{U}_0(z), \choice{U}_1(z), \choice{U}_2(z)) \\
  \begin{pmatrix}
    2z^{\frac{1 - \alpha}{2}} e^{-n(\gfntilde(\sqrt{z}) - V(\sqrt{z}) - \ell)} P^{(-1, \infty)}_2(\sqrt{z}) & 0 & \!\!\!\!\! e^{-n(\gfntilde(\sqrt{z}) - \ell)} P^{(-1, \infty)}_2(\sqrt{z}) \\
    e^{n\gfn(\sqrt{z})} P^{(-1, \infty)}_1(\sqrt{z}) & 0 & \!\!\!\!\! 0 \\
    -c e^{n(\gfn(-\sqrt{z}) - 4a\sqrt{z})} P^{(-1, \infty)}_1(-\sqrt{z}) & e^{n\gfn(-\sqrt{z})} P^{(-1, \infty)}_1(-\sqrt{z}) & \!\!\!\!\! 0
  \end{pmatrix},
\end{multline}
where  $P^{(-1, \infty)}_1(\sqrt{z})$ and $P^{(-1, \infty)}_2(\sqrt{z})$ are defined in \eqref{eq:defn_P^jinfty_1} and \eqref{eq:defn_P^jinfty_2}.

We define the vector-valued functions
\begin{align}
  \hat{Y}^{(n + k)}(z) = {}& Y^{(n + k + 1)}(z) + c^{(n + k, k)}_{10} Y^{(n + k)}(z), \\
  \doublehat{Y}^{(n + k)}(z) = {}& Y^{(n + k + 2)}(z) + c^{(n + k, k)}_{21} Y^{(n + k + 1)}(z) + c^{(n + k, k)}_{20} Y^{(n + k)}(z).
\end{align}
Recall the matrix-valued function $X^{(n)}(z)$ defined in \eqref{eq:defn_X^(m)} and we specify $\Vhat(x)$ in \eqref{eq:specified_What} as $x^4/2 - tx^2$, the same as in Section \ref{eq:quartic_V}. Then $X^{(n)}(z)$ satisfies
\begin{align}
    X^{(n)}(z) = {}& \Chat^{(n)}
    \begin{pmatrix}
      Y^{(n - 1)}_1(\sqrt{z}) & Y^{(n - 1)}_1(-\sqrt{z}) & Y^{(n - 1)}_2(\sqrt{z}) \\
      \hat{Y}^{(n - 1)}_1(\sqrt{z}) & \hat{Y}^{(n - 1)}_1(-\sqrt{z}) & \hat{Y}^{(n - 1)}_2(\sqrt{z}) \\
      \doublehat{Y}^{(n - 1)}_1(\sqrt{z}) & \doublehat{Y}^{(n - 1)}_1(-\sqrt{z}) & \doublehat{Y}^{(n - 1)}_2(\sqrt{z})
    \end{pmatrix}
    \begin{pmatrix}
      0 & \frac{1}{2} & \frac{1}{2\sqrt{z}} \\
      0 & \frac{1}{2} & -\frac{1}{2\sqrt{z}} \\
      -2 & 0 & 0
    \end{pmatrix} \nonumber \\
    = {}& \Chat^{(n)} R(z)
          \mathsf{P}^{(-1, 0)}(z)
          \begin{pmatrix}
            0 & 0 & 1 \\
            1 & 0 & 0 \\
            0 & 1 & 0
          \end{pmatrix} \nonumber \\
                   & \times
                     \begin{pmatrix}
                       e^{n\gfn(\sqrt{z})} P^{(-1, \infty)}_1(\sqrt{z}) & 0 & 0 \\
                       0 & e^{n\gfn(-\sqrt{z})} P^{(-1, \infty)}_1(-\sqrt{z}) & 0 \\
                       0 & 0 & e^{-n(\gfntilde(\sqrt{z}) - \ell)} P^{(-1, \infty)}_2(\sqrt{z})
                     \end{pmatrix} \nonumber \\
                   & \times
                     \begin{pmatrix}
                       1 & 0 & 0 \\
                       -\gamma e^{- 4na\sqrt{z}} & 1 & 0 \\
                       2z^{\frac{1 - \alpha}{2}} e^{nV(\sqrt{z})} & 0 & 1
                     \end{pmatrix}
                     \begin{pmatrix}
                       0 & \frac{1}{2} & \frac{1}{2\sqrt{z}} \\
                       0 & \frac{1}{2} & -\frac{1}{2\sqrt{z}} \\
                       -2  & 0 & 0
                     \end{pmatrix},
\end{align}
where $R(z)$ is defined in \eqref{eq:U-R}, $\Chat^{(n)}$ is an invertible constant matrix that can be expressed in terms of the $C^{(n)}$ defined in \eqref{eq:defn_X^(m)}, $c^{(n - 1, -1)}_{10}$, $c^{(n - 1, -1)}_{20}$ and $c^{(n - 1, -1)}_{21}$ defined in \eqref{eq:defn_C_10} and \eqref{eq:defn_C_20_21}, and $a^{(n + 1)}_1$ defined in \eqref{eq:defn_tilde_p_m}. Hence, if we take the analytic continuation of $W_1(z)$ and $W_2(z)$ in \eqref{eq:defn_W_1_W_2} in the region $\{ z \in D(0, r^2_n) : \arg z \in (0, \pi/2) \}$,
\begin{equation} \label{eq:X^n_right_vector}
  \begin{split}
    & X^{(n)}(z)
      \begin{pmatrix}
        0 \\
        W_1(z) \\
        W_2(z)
      \end{pmatrix} \\
    = {}& \Chat^{(n)}R(z)
          \mathsf{P}^{(-1, 0)}(z)
          \begin{pmatrix}
            2 e^{n(-\gfntilde(\sqrt{z}) + \ell)} P^{(-1, \infty)}_2(\sqrt{z}) \\
            z^{\frac{\alpha - 1}{2}} e^{n(\gfn(\sqrt{z}) - V(\sqrt{z}))} P^{(-1, \infty)}_1(\sqrt{z}) \\
            0
          \end{pmatrix} \\
    = {}& (-c)^{-1} (-1)^n e^{n(\ell - \Re \gfntilde_+(0))} \frac{2}{\sqrt{6}} c^{\frac{\alpha}{3}} z^{\frac{\alpha}{3} - \frac{1}{2}} \Chat^{(n)}
    R(z)\\
    & \times
      \left\{
      \begin{aligned}
        & \diag((\rho n)^{-\frac{k}{2}})^2_{k = 0} \Phi^{(\gamma, \alpha, 0, \frac{3}{4}, \tau)}((\rho n)^{\frac{3}{2}} z)
          \begin{pmatrix}
            -e^{-\frac{\alpha}{3} \pi i} \\
            -e^{\frac{\alpha}{3} \pi i} \\
            0
          \end{pmatrix}, && \text{(Pearcey regime)} \\
        & e^{\frac{2}{\sqrt{3}} n z} \diag((3^{\frac{2}{3}} n)^{-\frac{k}{4}})^2_{k = 0} \Phi^{(\gamma, \alpha, -\frac{1}{4}, \frac{3}{2}\sigma, \tau)}((3^{\frac{2}{3}} n)^{\frac{3}{4}} z)
          \begin{pmatrix}
            -e^{-\frac{\alpha}{3} \pi i} \\
            -e^{\frac{\alpha}{3} \pi i} \\
            0
          \end{pmatrix}, && \text{(multi-critical regime)}. \\
      \end{aligned}
      \right.
  \end{split}
\end{equation}
We also have
\begin{equation}
  \begin{split}
    X^{(n)}(z)^{-1} = {}&
                          \begin{pmatrix}
                            0 & 0 & -\frac{1}{2} \\
                            1 & 1 & 0 \\
                            \sqrt{z} & -\sqrt{z} & 0
                          \end{pmatrix}
                          \begin{pmatrix}
                            1 & 0 & 0 \\
                            \gamma e^{-4na \sqrt{z}} & 1 & 0 \\
                            -2z^{\frac{1 - \alpha}{2}} e^{nV(\sqrt{z})} & 0 & 1
                          \end{pmatrix} \\
                        & \times
                          \begin{pmatrix}
                            \frac{e^{-n\gfn(\sqrt{z})}}{P^{(-1, \infty)}_1(\sqrt{z})} & 0 & 0 \\
                            0 & \frac{e^{-n\gfn(-\sqrt{z})}}{P^{(-1, \infty)}_1(-\sqrt{z})} & 0 \\
                            0 & 0 & \frac{e^{n(\gfntilde(\sqrt{z}) - \ell)}}{P^{(-1, \infty)}_2(\sqrt{z})}
                          \end{pmatrix} \\
                        & \times
                          \begin{pmatrix}
                            0 & 1 & 0 \\
                            0 & 0 & 1 \\
                            1 & 0 & 0
                          \end{pmatrix}
                          \mathsf{P}^{(-1, 0)}(z)^{-1}R(z)^{-1}
                          (\Chat^{(n)})^{-1},
  \end{split}
\end{equation}
which implies 
\begin{equation}
  \begin{split}
    &
      \begin{pmatrix}
        1 & 0 & 0
      \end{pmatrix} 
      X^{(n)}(z)^{-1} \\
    = {}& \left( -\frac{1}{2} \frac{e^{n(\gfntilde(\sqrt{z}) - \ell)}}{P^{(-1, \infty)}_2(\sqrt{z})}, -z^{\frac{1 - \alpha}{2}} \frac{e^{-n(\gfn(\sqrt{z}) - V(\sqrt{z}))}}{P^{(-1, \infty)}_1(\sqrt{z})}, 0 \right)
          \mathsf{P}^{(-1, 0)}(z)^{-1}R(z)^{-1}(\Chat^{(n)})^{-1}\\
    = {}& (-c) (-1)^n e^{n(\Re \gfntilde_+(0) - \ell)} \frac{\sqrt{6}}{2} c^{-\frac{\alpha}{3}} z^{\frac{1}{2} - \frac{\alpha}{3}} (e^{\frac{\alpha}{3} \pi i}, -e^{-\frac{\alpha}{3} \pi i}, 0) \\
    & \times
      \left\{
      \begin{aligned}
        &\Phi^{(\gamma, \alpha, 0, \frac{3}{4}, \tau)}((\rho n)^{\frac{3}{2}} z)^{-1} \diag((\rho n)^{\frac{k}{2}})^2_{k = 0}R(z)^{-1} (\Chat^{(n)})^{-1}, && \text{(Pearcey regime)} \\
        &e^{-\frac{2}{\sqrt{3}} nz} \Phi^{(\gamma, \alpha, -\frac{1}{4}, \frac{3}{2}\sigma, \tau)}((3^{\frac{2}{3}} n)^{\frac{3}{4}} z)^{-1} \diag((3^{\frac{2}{3}} n)^{\frac{k}{4}})^2_{k = 0}R(z)^{-1} (\Chat^{(n)})^{-1}, && \text{(multi-critical regime)}.
      \end{aligned}\right.
  \end{split}
\end{equation}
Finally, for $x, y$ in the region $\{ z \in D(0, r^2_n) : \arg z \in (0, \pi/2) \}$, we have
\begin{equation}
  \begin{split}
    & 
      \begin{pmatrix}
        1 & 0 & 0
      \end{pmatrix} 
      X^{(n)}(y)^{-1} X^{(n)}(x)
      \begin{pmatrix}
        0 \\
        W_1(x) \\
        W_2(x)
      \end{pmatrix} \\
    = {}& 
          \left( -\frac{1}{2} \frac{e^{n(\gfntilde(\sqrt{y}) - \ell)}}{P^{(-1, \infty)}_2(\sqrt{y})}, -y^{\frac{1 - \alpha}{2}} \frac{e^{-n(\gfn(\sqrt{y}) - V(\sqrt{y}))}}{P^{(-1, \infty)}_1(\sqrt{y})}, 0 \right)
          \mathsf{P}^{(-1, 0)}(y)^{-1}R(y)^{-1}\\
    & \times
    R(x)  \mathsf{P}^{(-1, 0)}(x)
      \begin{pmatrix}
        2 e^{n(-\gfntilde(\sqrt{x}) + \ell)} P^{(-1, \infty)}_2(\sqrt{x}) \\
        x^{\frac{\alpha - 1}{2}} e^{n(\gfn(\sqrt{x}) - V(\sqrt{x}))} P^{(-1, \infty)}_1(\sqrt{x}) \\
        0
      \end{pmatrix} \\
   & = {}- \left( \frac{x}{y} \right)^{\frac{\alpha}{3} - \frac{1}{2}} (e^{\frac{\alpha}{3} \pi i}, -e^{-\frac{\alpha}{3} \pi i}, 0)\times \\
    &  \left\{
      \begin{aligned}
        & \Phi^{(\gamma, \alpha, 0, \frac{3}{4}, \tau)}((\rho n)^{\frac{3}{2}} y)^{-1} \diag((\rho n)^{\frac{k}{2}})^2_{k = 0}
          R(y)^{-1}R(x) \diag((\rho n)^{-\frac{k}{2}})^2_{k = 0} \\
        & \hspace{4cm}
          \times \Phi^{(\gamma, \alpha, 0, \frac{3}{4}, \tau)}((\rho n)^{\frac{3}{2}} x)
          \begin{pmatrix}
            e^{-\frac{\alpha}{3} \pi i} \\
            e^{\frac{\alpha}{3} \pi i} \\
            0
          \end{pmatrix},
          \quad \text{(Pearcey regime)} \\
        &
         \Phi^{(\gamma, \alpha, -\frac{1}{4}, \frac{3}{2}\sigma, \tau)}((3^{\frac{2}{3}} n)^{\frac{3}{4}} y)^{-1} \diag((3^{\frac{2}{3}} n)^{\frac{k}{4}})^2_{k = 0}
          \frac{R(y)^{-1}R(x) }{e^{\frac{2}{\sqrt{3}} n(y - x)} }\diag((3^{\frac{2}{3}} n)^{-\frac{k}{4}})^2_{k = 0} \\
        & \hspace{4cm}
          \times \Phi^{(\gamma, \alpha, -\frac{1}{4}, \frac{3}{2}\sigma, \tau)}((3^{\frac{2}{3}} n)^{\frac{3}{4}} x)
          \begin{pmatrix}
            e^{-\frac{\alpha}{3} \pi i} \\
            e^{\frac{\alpha}{3} \pi i} \\
            0
          \end{pmatrix}, 
          \quad \text{(multi-critical regime)}
      \end{aligned}
      \right.
  \end{split}
\end{equation}
We note that this result extends to $x, y \in (0, r^2_n)$ by taking straightforward analytic continuation.

\subsection{Proof of Theorem \ref{thm:main}}

Before we prove Theorem \ref{thm:main}, we first show that the limiting kernel defined by \eqref{eq:limit_K^mult_entries} is equivalent to \eqref{eq:limit_kernel_mult}. To see this, we recall the definition of $\tilde{\Phi}^{(\gamma, \alpha, \rho, \sigma, \tau)}(\xi)$ in the beginning of Section \ref{subsec:main_results}, and find that for $\xi$ in the sector $\{ \arg \xi \in (0, \pi/4) \}$, 
\begin{equation}\label{eq:def_tildePhi}
  \tilde{\Phi}^{(\gamma, \alpha, \rho, \sigma, \tau)}(\xi) = \xi^{ \frac{\alpha}{3}-\frac{1}{2} } \Phi^{(\gamma, \alpha, \rho, \sigma, \tau)}(\xi)
  \begin{pmatrix}
    1 & e^{-\frac{2\alpha}{3} \pi i} & 0 \\
    0 & 1 & 0 \\
    0 & 0 & 1
  \end{pmatrix}.
\end{equation}
This implies the equivalence between \eqref{eq:limit_K^mult_entries} and \eqref{eq:limit_kernel_mult}. Therefore we only need to prove Theorem \ref{thm:main} with $K^{(\gamma, \alpha, \rho, \sigma, \tau)}$ defined by \eqref{eq:limit_K^mult_entries}.

Now we consider the asymptotics in the Pearcey regime. We note that if $z = \bigO(n^{-3/2})$ and $\arg z \in (0, \pi/2)$, we have by \eqref{eq:v_0v_1v_2}--\eqref{eq:v_0v_1v_2_doublehat}
\begin{multline}
  \diag((\rho n)^{\frac{k}{2}})^2_{k = 0}R(z)
  \diag((\rho n)^{-\frac{k}{2}})^2_{k = 0} = 
  \begin{pmatrix}
    c_0 + \bigO(n^{-\frac{2}{5}}) & \bigO(n^{-\frac{1}{5}}) & \bigO(n^{-\frac{3}{10}}) \\
    \bigO(n^{-\frac{3}{10}}) & \hat{c}_0 + \bigO(n^{-\frac{2}{5}}) & \bigO(n^{-\frac{1}{5}}) \\
    \bigO(n^{-\frac{1}{5}}) & \bigO(n^{-\frac{3}{10}}) & \doublehat{c}_0 + \bigO(n^{-\frac{2}{5}})
  \end{pmatrix},
\end{multline}
where $R(z)$ is defined in \eqref{eq:U-R}.
This estimate extends to $z \in \realR_+$ by analytic continuation. Hence, we have that if $x, y = \bigO(n^{-3/2})$ and $x, y \in \realR_+$,
\begin{equation}
  \diag((\rho n)^{\frac{k}{2}})^2_{k = 0}R(y)^{-1}R(x)
  \diag((\rho n)^{-\frac{k}{2}})^2_{k = 0} 
  = I + \bigO(n^{-\frac{1}{5}}),
\end{equation}
and furthermore, if $\xi = (\rho n)^{3/2} x$ and $\eta = (\rho n)^{3/2} y$ are fixed positive numbers,
\begin{multline}
  \begin{pmatrix}
    1 & 0 & 0
  \end{pmatrix} 
  X^{(n)}(y)^{-1} X^{(n)}(x)
  \begin{pmatrix}
    0 \\
    W_1(x) \\
    W_2(x)
  \end{pmatrix} = \\
  \left( \frac{\xi}{\eta} \right)^{\frac{\alpha}{3} - \frac{1}{2}} (e^{\frac{\alpha}{3} \pi i}, -e^{-\frac{\alpha}{3} \pi i}, 0) \Phi^{(\gamma, \alpha, 0, \frac{3}{4}, \tau)}(\eta)^{-1} \Phi^{(\gamma, \alpha, 0, \frac{3}{4}, \tau)}(\xi)
  \begin{pmatrix}
    -e^{-\frac{\alpha}{3} \pi i} \\
    -e^{\frac{\alpha}{3} \pi i} \\
    0
  \end{pmatrix}
  + \bigO(n^{-\frac{1}{5}}).
\end{multline}
Thus, we prove Theorem \ref{thm:main} for the Pearcey regime.

Next, we consider the asymptotics in the multi-critical regime. We note that if $z = \bigO(n^{-3/4})$ and $\arg z \in (0, \pi/2)$, we have by \eqref{eq:v_0v_1v_2}--\eqref{eq:v_0v_1v_2_doublehat}
\begin{multline}
  \diag((3^{\frac{2}{3}} n)^{\frac{k}{4}})^2_{k = 0}R(z)
  \diag((3^{\frac{2}{3}} n)^{-\frac{k}{4}})^2_{k = 0} = 
  \begin{pmatrix}
    c_0 + \bigO(n^{-\frac{1}{9}}) & \bigO(n^{-\frac{1}{18}}) & \bigO(n^{-\frac{1}{12}}) \\
    \bigO(n^{-\frac{1}{12}}) & \hat{c}_0 + \bigO(n^{-\frac{1}{9}}) & \bigO(n^{-\frac{1}{18}}) \\
    \bigO(n^{-\frac{1}{18}}) & \bigO(n^{-\frac{1}{12}}) & \doublehat{c}_0 + \bigO(n^{-\frac{1}{9}})
  \end{pmatrix},
\end{multline}
and this estimate extends to $z \in (0, r^2_n)$ by analytic continuation. Hence, we have that if $x, y = \bigO(n^{-3/4})$ and $x, y \in \realR$,
\begin{equation}
  \diag((3^{\frac{2}{3}} n)^{\frac{k}{4}})^2_{k = 0}R(y)^{-1}R(x)
  \diag((3^{\frac{2}{3}} n)^{-\frac{k}{4}})^2_{k = 0} 
  = I + \bigO(n^{-\frac{1}{18}}),
\end{equation}
and furthermore, if $\xi = (3^{\frac{2}{3}} n)^{3/4} x$ and $\eta = (3^{\frac{2}{3}} n)^{3/4} y$ are fixed positive numbers,
\begin{multline}
  e^{\frac{2}{\sqrt{3}} n(y - x)} 
  \begin{pmatrix}
    1 & 0 & 0
  \end{pmatrix} 
  X^{(n)}(y)^{-1} X^{(n)}(x)
  \begin{pmatrix}
    0 \\
    W_1(x) \\
    W_2(x)
  \end{pmatrix}
  = \\
  \left( \frac{\xi}{\eta} \right)^{\frac{\alpha}{3} - \frac{1}{2}} (e^{\frac{\alpha}{3} \pi i}, -e^{-\frac{\alpha}{3} \pi i}, 0) \Phi^{(\gamma, \alpha, -\frac{1}{4}, \frac{3}{2}\sigma, \tau)}(\eta)^{-1} \Phi^{(\gamma, \alpha, -\frac{1}{4}, \frac{3}{2}\sigma, \tau)}(\xi)
  \begin{pmatrix}
    -e^{-\frac{\alpha}{3} \pi i} \\
    -e^{\frac{\alpha}{3} \pi i} \\
    0
  \end{pmatrix}
  + \bigO(n^{-\frac{1}{18}}).
\end{multline}
Thus, we prove Theorem \ref{thm:main} for the multi-critical regime.

\section{Proof of Theorem \ref{thm:RHP_unique_solvability}} \label{sec:solvability}

In this section, we prove Theorem \ref{thm:RHP_unique_solvability}. This is achieved by proving a vanishing lemma, which asserts that every solution of the associated homogeneous RH problem is trivial. Throughout this section, we assume that the parameters $\gamma, \alpha, \rho, \sigma, \tau$ satisfy the conditions in Theorem \ref{thm:RHP_unique_solvability}.

\subsection{Vanishing lemma}

\begin{lemma} (vanishing lemma) \label{lem:vanishing}
  Let $\Phi_0(\xi) = \Phi^{(\gamma, \alpha, \rho, \sigma, \tau)}_0(\xi)$ be a $3 \times 3$ vector-valued function on $\compC \setminus \Gamma_{\Phi}$, where $\Gamma_{\Phi}$ is the same as in RH problem \ref{RHP:model}. Suppose $\Phi_0$ satisfies the homogeneous version of RH problem \ref{RHP:model} such that the asymptotic behaviour at infinity is changed from \eqref{eq:expansion_at_infty_Phi} into
  \begin{equation}
    \Phi_0(\xi) = \bigO(\xi^{-1}) \Upsilon(\xi) \Omega_{\pm} e^{-\Theta(\xi)},
  \end{equation}
  and all other conditions in RH problem \ref{RHP:model} are unchanged. Then $\Phi_0$ is identically $0$, under the condition that $\alpha > -1$, $\gamma \in [-1, 1]$, and either of Conditions \ref{enu:thm:RHP_unique_solvability:1} or \ref{enu:thm:RHP_unique_solvability:2} in Theorem \ref{thm:RHP_unique_solvability} is satisfied.
\end{lemma}
\begin{proof}
We prove Lemma \ref{lem:vanishing} by contradiction. 
First, we show that we only need to consider real-analytic solutions of the homogeneous version of RH problem \ref{RHP:model}. Let $H(z)$ be a nontrivial solution of the homogeneous version of RH problem \ref{RHP:model}. Observe that $\overline{H(\bar{z})}$ is also a solution of this RH problem. Consequently, both $H(z) + \overline{H(\bar{z})}$ and $i(H(z) - \overline{H(\bar{z})})$ are solutions, and at least one of them must be nontrivial. By selecting a nontrivial one as $\Phi_0$, we ensure that $\Phi_0$ satisfies the symmetry relation 
\begin{equation} \label{eq:Symmetry_Psi_0}
  \Phi_0(z) = \overline{\Phi_0(\bar{z})}.
\end{equation}

Next, we show that every row of $\Phi_0$ is identically zero by contradiction. Suppose $(\phi_0, \phi_1, \phi_2)$ is one of the three rows of $\Phi_0$ and it is not identically zero.

Let $m_0$ be an integer such that $m_0 > \max(1, \alpha)$. We define the function $F(z)$ as follows
\begin{multline} \label{eq:defn_F(z)}
  F(z) = F^{(\pre)}(z) \times
  \begin{cases}
    1, & \lvert z \rvert > 1, \\
    z^{m_0}, & \lvert z \rvert < 1,
  \end{cases} \\
  \text{where} \quad
  F^{(\pre)}(z) = e^{\rho \omega z^4 + \sigma \omega^2 z^2 + \tau \omega z} \times
  \begin{cases}
    \phi_0(z^3), & \arg z \in (0, \frac{2\pi}{3}), \\
    \phi_2(z^3), & \arg z \in (-\pi, -\frac{\pi}{3}), \\
    \phi_1(z^3), & \arg z \in (\frac{2\pi}{3},\pi) \cup (-\frac{\pi}{3},0).
  \end{cases}
\end{multline}
It follows from the definition that $F(z)$ is analytic on $\compC \setminus \Gamma_F$, where
\begin{equation}
  \Gamma_F = \{ \lvert z \rvert = 1 \} \cup \{ \arg z = -\frac{\pi}{12}, -\frac{\pi}{4}, \frac{3\pi}{4}, \frac{11\pi}{12} \}.
\end{equation}
It is straightforward to see that $F(z) = \bigO(z^{-1})$ as $z \to \infty$, and $F(z)$ is bounded as $z \to 0$. Moreover, if we denote
\begin{equation}
  J_F(z) = F_+(z) - F_-(z), \quad z \in \Gamma_F,
\end{equation}
where $\Gamma_F$ is oriented such that the unit circle is counterclockwise and all rays are outward, then $J_F(z)$ can be expressed in terms of $\phi_0$, $\phi_1$ and $\phi_2$. We find that $J_F(z)$ is bounded on $\Gamma_F$, and it vanishes exponentially fast as $z \to \infty$ along the rays of $\Gamma_F$. Hence, we derive that
\begin{equation}
  F(z) = \frac{1}{2\pi i} \int_{\Gamma_F} \frac{J_F(w)}{w - z} dw,
\end{equation}
which implies that $F(z)$ has the asymptotic expansion at infinity:
\begin{equation} \label{eq:asy_exp_F}
  F(z) = \sum^{\infty}_{k = 1} c_k \omega^{-k} z^{-k}.
\end{equation}
Since we assume that $(\phi_0, \phi_1, \phi_2)$ is not identically $0$, we have that $F(z)$ is not identically $0$. We show below by contradiction that the coefficients $c_k$ in \eqref{eq:asy_exp_F} are not all zero. Otherwise, for all $m \geq 0$,
\begin{equation}
  \frac{1}{2\pi i} \int_{\Gamma_F} w^m J_F(w) dw = 0.
\end{equation}
We consider the value of $F(z)$ where $z = r$ or $z = e^{2\pi i/3} r$ with $r \in (1, +\infty)$. We have that for any $m$,
\begin{equation}
  \begin{split}
    F(z) = {}& \frac{1}{2\pi i} \int_{\Gamma_F} \frac{J_F(w)}{w - z} dw \\
    = {}& \frac{1}{2\pi i} \int_{\Gamma_F} J_F(w) \left( \frac{1}{w - z} - \frac{1 - (w/z)^m}{w - z} \right) dw \\
    = {}& \frac{1}{z^m} \frac{1}{2\pi i} \int_{\Gamma_F} J_F(w) \frac{w^m}{w - z} dw.
  \end{split}
\end{equation}
Then we show that $F(r)$ and $F(e^{2\pi i/3} r)$ vanish faster than any exponential function as $r \to \infty$. To be precise, we note that there exist $\epsilon, M > 0$ such that on $\Gamma_F$, 
\begin{equation}
  \lvert J_F(w) \rvert \leq M e^{-\epsilon \lvert w \rvert^2}.
\end{equation}
Then for large enough $r$, with $\epsilon', M' > 0$ that depend only on $\epsilon$ and $M$,
\begin{equation} \label{eq:F(z)_est}
  \lvert F(e^{\phi i} r) \rvert \leq \frac{M'}{r^m} \int^{\infty}_0 e^{-\epsilon' x^2} x^m dx, \quad \phi = 0 \text{ or } \frac{2\pi}{3}.
\end{equation}
We take $m$ to be the integer such that $m \leq 2 \epsilon' r^2 < m+1$. Then from \eqref{eq:F(z)_est}, we derive that 
\begin{equation}\label{eq:Expnential_Decay}
  \lvert F(e^{\phi i} r) \rvert \leq M'' e^{-\epsilon' r^2}, \quad \phi = 0 \text{ or } \frac{2\pi}{3},
\end{equation}
for some constant $M''$. Now consider the function
\begin{equation}
  \Fhat(z) = z^{\frac{2}{3} m_0} F^{(\pre)}(e^{\frac{\pi i}{3}} z^{\frac{2}{3}}), \quad \arg z \in (-\frac{\pi}{2}, \frac{\pi}{2}).
\end{equation}
It is clear that $\Fhat(z)$ is analytic in the half-plane $\{ \Re z > 0 \}$, and is continuous up to the boundary $\{ \Re z = 0 \}$. Moreover, $\lvert \Fhat(z) \rvert$ is bounded by $\lvert z \rvert^{2m_0/3}$ as $z \to \infty$, and also for $z = iy$ on the boundary, $\lvert \Fhat(iy) \rvert$ is bounded by $M''e^{-\epsilon' y^{4/3}} \max(1, \lvert y \rvert^{2 m_0/3})$. Hence, we can apply Carlson's theorem as stated in \cite[Section XIII.13]{Reed-Simon78} to show that $\Fhat(z)$ vanishes identically in $\{ \Re z \geq 0 \}$. This implies that $\phi_1(z)$ is identically zero, and then $(\phi_0(z), \phi_1(z), \phi_2(z))$ is identically zero, violating our assumption.
Thus, the coefficients $c_k$ in \eqref{eq:asy_exp_F} are not all zero. 

Below we assume that $K_0 = \min \{ k \geq 1 : c_k \neq 0 \}$. We define
\begin{equation}
  f(z) =
  \begin{cases}
    \phi_1(z), & \arg z \in (\frac{\pi}{4}, \frac{3\pi}{4}) \cup (-\frac{3\pi}{4}, -\frac{\pi}{4}), \\
    \phi_1(z) + e^{-\frac{2\alpha}{3} \pi i} \phi_0(z), & \arg z \in (0, \frac{\pi}{4}), \\
    \phi_1(z) + e^{\frac{2\alpha}{3} \pi i} \phi_0(z), & \arg z \in (-\frac{\pi}{4}, 0), \\
    \phi_1(z) + \gamma e^{-\frac{\alpha}{3} \pi i} \phi_2(z), & \arg z \in (\frac{3\pi}{4}, \pi), \\
    \phi_1(z) + \gamma e^{\frac{\alpha}{3} \pi i} \phi_2(z), & \arg z \in (-\pi, -\frac{3\pi}{4}),
  \end{cases}
\end{equation}
and then we have (with $\realR$ oriented from $-\infty$ to $+\infty$)
\begin{equation} \label{eq:f_jump}
 f_+(x) =
  \begin{cases}
    e^{-\frac{2\alpha}{3} \pi i}f_{-}(x), & x\in (0, +\infty), \\
    \gamma e^{-\frac{\alpha}{3} \pi i} f_{-}(x)+(1-\gamma^2)\phi_{2,-}(x), & x\in (-\infty,0).
  \end{cases}
\end{equation}
We also have that for $x \in (0, +\infty)$:
\begin{align}
  \phi_{0, +}(x) + e^{\frac{2\alpha}{3} \pi i}\phi_{0, -}(x) = {}& f_{-}(x), & \phi_{0, +}(-x) = {}& \phi_{0, -}(-x), \label{eq:phi_0_jump} \\
  \phi_{2, +}(-x) + \gamma e^{\frac{\alpha}{3} \pi i}\phi_{2, -}(-x) = {}& f_{-}(-x), & \phi_{2, +}(x) = {}& \phi_{2, -}(x).
\end{align}
For later use, we derive that for $x \in (-\infty,0)$:
\begin{equation} \label{eq:jump_on_real}
  \begin{split}
    & (f\phi_2)_+(x) + (f\phi_2)_-(x) \\
    = {}& f_{+}(x)(f_{-}(x) - \gamma e^{\frac{\alpha}{3} \pi i}\phi_{2, -}(x)) + (f\phi_2)_-(x)\\
    = {}& f_{+}(x) f_{-}(x) - \gamma e^{\frac{\alpha}{3} \pi i}\phi_{2, -}(x)(\gamma e^{-\frac{\alpha}{3} \pi i}f_{-}(x)+(1-\gamma^2)\phi_{2,-}(x)) + (f\phi_2)_-(x)\\
    = {}& f_{+}(x) f_{-}(x) + (1-\gamma^2) \phi_{2, -}(x)(f_{-}(x) - \gamma e^{\frac{\alpha}{3} \pi i}\phi_{2, -}(x)) \\
    = {}& f_{+}(x) f_{-}(x) + (1-\gamma^2) \phi_{2, +}(x)\phi_{2, -}(x).
  \end{split}
\end{equation}
We have, from the asymptotic expansion of $F(z)$ at infinity,
\begin{align}
  \phi_0(z) = {}& M_{0}(z) \times
                  \begin{cases}
                    e^{-\rho \omega \xi^{\frac{4}{3}} - \sigma \omega^2 z^{\frac{2}{3}} - \tau \omega z^{\frac{1}{3}}}, & \arg z \in (0, \pi), \\
                    e^{-\rho \omega^2 \xi^{\frac{4}{3}} - \sigma \omega z^{\frac{2}{3}} - \tau \omega^2 z^{\frac{1}{3}}}, & \arg z \in (-\pi, 0),
                  \end{cases} \label{eq:vanishing_g_1_infty} \\
  f(z) = {}& M_1(z) \times
                      \begin{cases}
                        e^{\rho \omega^2 \xi^{\frac{4}{3}} - \sigma \omega z^{\frac{2}{3}} - \tau \omega^2 z^{\frac{1}{3}}}, & \arg z \in (0, \pi), \\
                        e^{\rho \omega \xi^{\frac{4}{3}} - \sigma \omega^2 z^{\frac{2}{3}} - \tau \omega z^{\frac{1}{3}}}, & \arg z \in (-\pi, 0),
                      \end{cases} \label{eq:vanishing_f_infty} \\
  \phi_2(z) = {}& M_{2}(z) e^{-\rho \xi^{\frac{4}{3}} - \sigma z^{\frac{2}{3}} - \tau z^{\frac{1}{3}}}, \label{eq:vanishing_g_2_infty} 
\end{align}
where the terms $M_0(z)$, $M_{1}(z)$ and $M_{2}(z)$ are determined by $F(z)$ in \eqref{eq:defn_F(z)} and satisfy the following properties:
\begin{align}
 M_{0}(z) = {}& \sum^{\infty}_{k = K_0} c_k (-1)^k (-z)^{-\frac{k}{3}}, \quad \arg (-z) \in (-\pi + \epsilon, \pi - \epsilon), \label{eq:asy_exp_N_g_0} \\
  M_1(z) = {}&
               \begin{cases}
                 \sum^{\infty}_{k = K_0} c_k \omega^k z^{-\frac{k}{3}}, & \arg z \in (\epsilon, \pi - \epsilon), \\
                 \sum^{\infty}_{k = K_0} c_k \omega^{-k} z^{-\frac{k}{3}}, & \arg z \in (-\pi + \epsilon, -\epsilon), 
               \end{cases} \label{eq:asy_exp_N_f} \\
  M_{2}(z) = {}& \sum^{\infty}_{k = K_0} c_k z^{-\frac{k}{3}}, \quad \arg z \in (-\pi + \epsilon, \pi - \epsilon), \label{eq:asy_exp_N_g_2}
\end{align}
and $M_i(z) = \bigO(z^{-K_0/3})$ as $z \to \infty$ in any direction for $i = 0, 1, 2$. Then we have that
\begin{equation} \label{eq:est_M_i}
  f(z) \phi_0(z) \phi_2(z) = \bigO(z^{-K_0})
\end{equation}
holds uniformly as $\lvert z \rvert \to \infty$.
Similarly, from the behavior of $\Phi_0$ near the origin, if $\alpha$ is not an integer, we have as $z \to 0$:
\begin{multline} \label{eq:phi_0_at_origin_generic}
  \phi_0(z) = N_{0}(z) (-z)^{\frac{1}{2}-\frac{\alpha }{3}} - \frac{1-\gamma}{2\sin(\frac{\alpha\pi}{2})}N_{1}(z) (-z)^{\frac{1}{2}+\frac{\alpha }{6}} +\frac{1+\gamma}{2\cos(\frac{\alpha\pi}{2})}N_{2}(z) (-z)^{\frac{\alpha }{6}}, \\
  \arg (-z)\in(-\pi,\pi),
\end{multline}
\begin{align}
  f(z) = {}& 
                      \begin{cases}
                        (1-\gamma) e^{-\frac{\alpha\pi}{3}i}N_{1}(z) z^{\frac{1}{2}+\frac{\alpha }{6}}+(1+\gamma) e^{-\frac{\alpha\pi}{3}i}N_{2}(z) z^{\frac{\alpha }{6}}, & \arg z\in(0,\pi), \\
                        (1-\gamma) e^{\frac{\alpha\pi}{3}i}N_{1}(z) z^{\frac{1}{2}+\frac{\alpha }{6}}+(1+\gamma) e^{\frac{\alpha\pi}{3}i}N_{2}(z) z^{\frac{\alpha }{6}}, & \arg z\in(-\pi,0),
                      \end{cases} \\
  \phi_2(z) = {}& -N_{1}(z) z^{\frac{1}{2}+\frac{\alpha }{6}}+N_{2}(z) z^{\frac{\alpha }{6}}, \quad \arg z\in(-\pi,\pi), \label{eq:phi_2_at_origin_generic}
\end{align}
where $N_{0}(z)$, $N_{1}(z)$ and $N_{2}(z)$ are analytic functions defined in a neighborhood of the origin. If $\alpha = 2k$ is an even integer, \eqref{eq:phi_0_at_origin_generic} is substituted by
\begin{multline} \label{eq:phi_0_at_origin_even}
  \phi_0(z) = N_{0}(z) (-z)^{\frac{1}{2}-\frac{\alpha }{3}} - (-1)^k\frac{1-\gamma}{2\pi}N_{1}(z) (-z)^{\frac{1}{2}+\frac{\alpha }{6}} \log(-z) +\frac{1+\gamma}{2\cos(\frac{\alpha\pi}{2})}N_{2}(z) (-z)^{\frac{\alpha }{6}}, 
\end{multline}
for $\arg (-z)\in(-\pi,\pi)$, and if $\alpha = 2k + 1$ is an odd integer, \eqref{eq:phi_0_at_origin_generic} is substituted by
\begin{multline} \label{eq:phi_0_at_origin_odd}
  \phi_0(z) = N_{0}(z) (-z)^{\frac{1}{2}-\frac{\alpha }{3}} - \frac{1-\gamma}{2\sin(\frac{\alpha\pi}{2})}N_{1}(z) (-z)^{\frac{1}{2}+\frac{\alpha }{6}} - (-1)^k \frac{1+\gamma}{2\pi}N_{2}(z) (-z)^{\frac{\alpha }{6}} \log(-z),
\end{multline}
for $\arg (-z)\in(-\pi,\pi)$. 
For all $\alpha > -1$ and $\gamma \in [-1, 0]$, we have
\begin{align}
  f(z) \phi_0(z) \phi_1(z) = {}&
  \begin{cases}
    \bigO(z^{\frac{1}{2}}), & \alpha > 1, \\
    \bigO(z^{\frac{1}{2}}\log z), & \alpha = 1, \\
    \bigO(z^{\frac{\alpha}{2}}), & -1< \alpha < 1,
  \end{cases}
  & \text{if $\gamma \neq -1$ and $N_2(0) \neq 0$}, \\
  f(z) \phi_0(z) \phi_1(z) = {}&
  \begin{cases}
    \bigO(z), & \alpha > 0, \\
    \bigO(z \log z), & \alpha = 0, \\
    \bigO(z^{1 + \frac{\alpha}{2}}), & -1< \alpha < 0,
  \end{cases}
  & \text{if $\gamma = -1$ or $N_2(0) = 0$}. \label{eq:est_gamma=-1_N2(0)=0}
\end{align}

From \eqref{eq:Symmetry_Psi_0}, we have that $\phi_2(x) \in \realR$ and $\phi_0(-x) \in \realR$ for all $x \in (0, +\infty)$. As a result, the functions $N_{0}(x)$, $N_{1}(x)$ and $N_{2}(x)$ are real for $x\in\realR$. Furthermore, we have $f_{+}(x) f_{-}(x)=|f_{+}(x)|^2$ and $\phi_{2,+}(x)\phi_{2,-}(x) =|\phi_{2,+}(x)|^2$.

From the asymptotic behaviour of $\phi_0(z)$ and $\phi_2(z)$ given in \eqref{eq:vanishing_g_1_infty}, \eqref{eq:asy_exp_N_g_0}, \eqref{eq:vanishing_g_2_infty} and \eqref{eq:asy_exp_N_g_2} as $z \to \infty$ and the limits of $\phi_0(z)$ and $\phi_2(z)$ given in \eqref{eq:phi_0_at_origin_generic}, \eqref{eq:phi_0_at_origin_even}, \eqref{eq:phi_0_at_origin_odd} and \eqref{eq:phi_2_at_origin_generic}, we have that $\phi_0(z)$ has finitely many zeros on $(-\infty, 0)$ and $\phi_2(z)$ has finitely many zeros on $(0, +\infty)$. Let $a_1, \dotsc, a_m$ be zeros of $\phi_0(z)$ on $(-\infty, 0)$ and $b_1, \dotsc, b_n$ be zeros of $\phi_2(z)$ on $(0, +\infty)$. Let $k = 0$ or $1$. Denote $\Gamma_R$ the contour $R e^{i\theta}$ with $R$ a large enough positive number and $\theta\in(0,\pi)\cup(-\pi,0)$, and orient $\Gamma_R$ from left to right, in both the upper curve and the lower curve. We consider the contour integral
\begin{equation} \label{eq:vanishing_integral_on_circle}
  -\oint_{\Gamma_R} f(z) \frac{\phi_0(z)}{\prod^m_{i = 1}(z - a_i)} \frac{\phi_2(z)}{\prod^n_{i = 1}(z - b_i)} z^{ - k} dz,
\end{equation} 
where $\Gamma_R = \{ R e^{i\theta} : \theta \in (0, \pi) \cup (-\pi, 0) \}$. 
As $R \to \infty$, the integrand in \eqref{eq:vanishing_integral_on_circle} is $\bigO(z^{-(K_0 + k + m + n)})$ as $z \to \infty$, where $K_0 \geq 1$ is the same as in \eqref{eq:est_M_i}. Hence, as $R \to \infty$, the contour integral is bounded, and it vanishes if $K_0 + k + m + n > 1$.

Deforming the contour $\Gamma_R$ into two flat contours, one from $-R + i\epsilon$ to $R + i\epsilon$ and the other from $-R - i\epsilon$ to $R - i\epsilon$ such that $\epsilon \to 0_+$, and using \eqref{eq:f_jump}, \eqref{eq:phi_0_jump} and \eqref{eq:jump_on_real}, we have that the contour integral in \eqref{eq:vanishing_integral_on_circle} is equal to $I_1(R) + I_2(R)$, where 
\begin{equation} \label{eq:vanishing_integral_on_real}
     I_1(R) =  \int^R_0| f_{+}(x) |^2 h(x) x^{-k} dx, \quad  I_2(R) = \int^0_{-R}   \left( |f_{+}(x) |^2+(1-\gamma^2)|\phi_{2,+}(x) |^2\right) h(x) x^{-k} dx,
\end{equation}
with
\begin{equation}
  h(x) =
  \begin{cases}
    \frac{\phi_2(x)}{\prod^n_{i = 1} (x - b_i)\prod^m_{i = 1} (x - a_i)}, & x \in (0, +\infty), \\
    \frac{\phi_0(x)}{\prod^m_{i = 1} (x - a_i)\prod^n_{i = 1} (x - b_i)}, & x \in (-\infty, 0),
  \end{cases}
\end{equation}
as long as $I_1(R)$ and $I_2(R)$ are well-defined. We note that $h(x)$ does not change sign on $(0, +\infty)$ and $(-\infty, 0)$, respectively.
Below we consider the limits of $I_1(R)$ and $I_2(R)$ as $R \to \infty$, with properly chosen $k = 0$ or $1$, in two cases. In either case, we show that by computing the integrals in \eqref{eq:vanishing_integral_on_real}, we derive that $\lim_{R \to \infty} I_1(R) + I_2(R) \neq 0$. On the other hand, in either case, we also show $\lim_{R \to \infty} I_1(R) + I_2(R) = 0$ from the vanishing property of \eqref{eq:vanishing_integral_on_circle}. Hence, we derive the contradiction that finally proves Lemma \ref{lem:vanishing}.

\paragraph{Case 1: $\alpha \geq 1$, or $\gamma = -1$, or $N_2(0) = 0$}

In this case, $I_1(R)$ and $I_2(R)$ in \eqref{eq:vanishing_integral_on_real} are well-defined for either $k = 0$ or $k = 1$, since the integrand is $\bigO(x^{-1/2} \log x)$ as $x \to 0$. We take $k = 1$ if $h(x)$ changes sign at $0$, and $k = 0$ if $h(x)$ does not change sign at $0$. Then the integrand in \eqref{eq:vanishing_integral_on_real} is sign-constant at $0$.  

If $k = 1$, we see that the integrand in \eqref{eq:vanishing_integral_on_circle} decays faster than $z^{-1}$ since $K_0 + 1 + m + n > 1$. Consequently, the integral \eqref{eq:vanishing_integral_on_circle} vanishes as $R \to \infty$, which implies that $\lim_{R \to \infty} I_1(R) + I_2(R) = 0$. On the other hand, since the integrands for $I_1(R)$ and $I_2(R)$ are of the same sign, and both of them vanish faster than $x^{-1}$, we have that $\lim_{R \to \infty} I_1(R) + I_2(R)$ exists and is nonzero. We thus derive the contradiction.

If $k = 0$, the same argument as above can be applied to derive the contradiction, if $K_0 + m + n > 1$. We show below that this condition is always satisfied; that is, it is not possible for $m = n = 0$ and $K_0 = 1$ to hold simultaneously. If $K_0 = 1$ and $m = n = 0$, that is, the coefficient $c_1$ in \eqref{eq:asy_exp_N_g_0}, \eqref{eq:asy_exp_N_f} and \eqref{eq:asy_exp_N_g_2} is nonzero, $\phi_0(x)$ does not change sign on $(-\infty, 0)$ and $\phi_2(x)$ does not change sign on $(0, +\infty)$. Then we have that $h(x)$ has the same sign as $c_1$ as $x \to +\infty$ and has the opposite sign as $c_1$ as $x \to -\infty$, so that $h(x)$ has opposite signs on $(-\infty, 0)$ and $(0, +\infty)$. This is contradictory to the assumption that $h(x)$ does not change sign at $x = 0$.

\paragraph{Case 2: $-1 < \alpha < 1$, $\gamma \neq -1$ and $N_2(0) \neq 0$}

In this case, $I_1(R)$ and $I_2(R)$ in \eqref{eq:vanishing_integral_on_circle} are well-defined when $k=0$. From formula \eqref{eq:phi_2_at_origin_generic} for $\phi_2(x)$ as $x \to 0_+$, and from formulas \eqref{eq:phi_0_at_origin_generic} with $\alpha \in (-1, 0) \cup (0, 1)$ and \eqref{eq:phi_0_at_origin_even} with $\alpha = 0$ for $\phi_0(x)$ as $x \to 0_-$, we have that $h(x)$ does not change sign at $0$. Then we take $k = 0$, and use the arguments in Case 1 to show that either $K_0 > 1$ or $m + n > 0$. Then we derive the contradiction in the same way as in Case 1.
\end{proof}

\subsection{Proof of Theorem \ref{thm:RHP_unique_solvability}}
We now prove the solvability of RH problem \ref{RHP:model} by using the vanishing lemma following the strategy in \cite[Section 5.3]{Deift-Kriecherbauer-McLaughlin-Venakides-Zhou99a} and \cite[Section 2.2]{Claeys-Vanlessen07a}. 

We transform RH problem \ref{RHP:model} into an equivalent one for $\Phihat$, whose jump matrix is continuous on a jump contour $\Gamma_{\Phihat}$ and converges exponentially to the identity matrix as $\xi \to \infty$ on $\Gamma_{\Phihat}$, and whose solution is normalized at infinity. To this end, we define for $\xi \in \compC \setminus \Gamma_{\Phi}$
\begin{equation}
  \Phihat(\xi) = \Phi^{(\gamma, \alpha, \rho, \sigma, \tau)}(\xi) \times
  \begin{cases}
    e^{\Theta(\xi)} \Omega^{-1}_{\pm} \Upsilon(\xi)^{-1}, & \lvert \xi \rvert > 1, \\
    E^{-1} \diag \left( \xi^{-\frac{1}{2} + \frac{\alpha}{3}}, \xi^{-\frac{1}{2} - \frac{\alpha}{6}}, \xi^{-\frac{\alpha}{6}} \right), & \lvert \xi \rvert < 1,
  \end{cases}
\end{equation}
where $\Theta(\xi)$, $\Omega_{\pm}$ and $\Upsilon(\xi)$ are defined in \eqref{eq:Theta} and \eqref{def:Lpm}, and $E$ is the same as in \eqref{eq:limit_Phi_at_0_generic}. Then by straightforward analytic continuation, we find that $\Phihat(\xi)$ is analytic in $\lvert \xi \rvert < 1$, and satisfies the following RH problem:
\begin{RHP} \label{RHP:Phihat}
  $\Phihat$ is a $3\times 3$ matrix-valued function analytic in $\compC \setminus \Gamma_{\Phihat}$, where
  \begin{equation}
    \Gamma_{\Phihat} = \{ \lvert \xi \rvert = 1 \} \cup \{ e^{\frac{\pi}{4} i} t : t > 1 \} \cup \{ e^{-\frac{\pi}{4} i} t : t > 1 \} \cup \{ e^{\frac{3\pi}{4} i} t : t > 1 \} \cup \{ e^{-\frac{3\pi}{4} i} t : t > 1 \}.
  \end{equation}
  Here the unit circle of $\Gamma_{\Phihat}$ is positively oriented and all four rays of $\Gamma_{\Phihat}$ are oriented outward. 
  \begin{enumerate}
  \item
    We have $\Phihat_+(\xi) = \Phihat_-(\xi) J_{\Phihat}(\xi)$ for $\xi\in \Gamma_{\Phihat}$, where
    \begin{equation}
      J_{\Phihat}(\xi) =
      \begin{cases}
        \Upsilon(\xi) \Omega_{\pm} e^{-\Theta(\xi)} E^{-1} \diag \left( \xi^{-\frac{1}{2} + \frac{\alpha}{3}}, \xi^{-\frac{1}{2} - \frac{\alpha}{6}}, \xi^{-\frac{\alpha}{6}} \right), & \xi \text{ on the unit circle}, \\
        \Upsilon(\xi) \Omega_{\pm} e^{-\Theta(\xi)} J_{\Phi}(\xi) e^{\Theta(\xi)} \Omega^{-1}_{\pm} \Upsilon(\xi)^{-1}, & \xi \text{ on the four rays}.
      \end{cases}
    \end{equation}
  \item \label{enu:RHP:Phihat:3} We have
    \begin{equation} \label{eq:RHP:Phihat:3}
      \Phihat(\xi) = I + \bigO(\xi^{-1}), \quad \xi \to \infty.
    \end{equation}
  \end{enumerate}
\end{RHP}

The solvability of RH problem \ref{RHP:model} is therefore equivalent to that of RH problem \ref{RHP:Phihat}. By the general theory of the construction of solutions of RH problems, this reduces to the analysis of the associated singular integral operator
\begin{equation}
  C_{\Phihat}: L^2(\Gamma_{\Phihat}) \to L^2(\Gamma_{\Phihat}): \quad f \mapsto C_+[f(1 - J^{-1}_{\Phihat})],
\end{equation}
where $C_+$ is the boundary value of the following Cauchy operator from the positive side:
\begin{equation}
  Cf(\xi) = \frac{1}{2\pi i} \int_{\Gamma_{\Phihat}} \frac{f(s)}{s - \xi} ds, \quad \xi \in \compC \setminus \Gamma_{\Phihat}.
\end{equation}
Suppose that $I - C_{\Phihat}$ is invertible in $L^2(\Gamma_{\Phihat})$; then there exists $\mu \in L^2(\Gamma_{\Phihat})$ such that $(I - C_{\Phihat}) \mu = C_+(I - J^{-1}_{\Phihat})$, and
\begin{equation}
  \Phihat(\xi) = I + \frac{1}{2\pi i} \int_{\Gamma_{\Phihat}} \frac{(I + \mu(s))(I - J_{\Phihat}(s)^{-1})}{s - \xi} ds, \quad \xi \in \compC \setminus \Gamma_{\Phihat}
\end{equation}
satisfies RH problem \ref{RHP:Phihat} in the $L^2$ sense. Furthermore, one can use the analyticity of $J_{\Phihat}$ to show that $\Phihat$ satisfies the jump condition in the sense of continuous boundary values as well; see \cite[Step 3 of Sections 5.2 and 5.3]{Deift-Kriecherbauer-McLaughlin-Venakides-Zhou99a}. Also, it follows from the exponential decay of $I - J^{-1}_{\Phihat}$ as $\xi \to \infty$ on $\Gamma_{\Phihat}$ that the boundary condition of $\Phihat$ at infinity is also satisfied; see \cite[Proposition 5.4]{Deift-Kriecherbauer-McLaughlin-Venakides-Zhou99a}. Hence, RH problem \ref{RHP:Phihat} is solvable if the singular integral operator $I - C_{\Phihat}$ is invertible in $L^2(\Gamma_{\Phihat})$.

To prove that $I - C_{\Phihat}$ is invertible, it suffices to show that it is a Fredholm operator of index zero with trivial kernel. As in \cite[Steps 1 and 2 of Section 5.3]{Deift-Kriecherbauer-McLaughlin-Venakides-Zhou99a}, one can show that $I - C_{\Phihat}$ is indeed a Fredholm operator of zero index. To show that its kernel is trivial, we argue by contradiction using the vanishing lemma.

Suppose $\mu_0 \in L^2(\Gamma_{\Phihat})$ is nonzero such that $(I - C_{\Phihat}) \mu_0 = 0$. One can then show that the matrix-valued function $\Phihat_0$ defined by
\begin{equation}
  \Phihat_0(\xi) = \frac{1}{2\pi i} \int_{\Gamma_{\Phihat}} \frac{\mu_0(s)(I - J_{\Phihat}(s)^{-1})}{s - \xi} ds, \quad \xi \in \compC \setminus \Gamma_{\Phihat}
\end{equation}
is a nontrivial solution to the homogeneous version of RH problem \ref{RHP:Phihat}, that is, with \eqref{eq:RHP:Phihat:3} replaced by
\begin{equation}
  \Phihat(\xi) = \bigO(\xi^{-1}), \quad \xi \to \infty.
\end{equation}
The vanishing lemma confirms that the homogeneous version of RH problem \ref{RHP:model} has only the trivial solution, and this is equivalent to the homogeneous version of RH problem \ref{RHP:Phihat} having only the trivial solution. Hence, we derive the contradiction that $\Phihat_0(\xi) = 0$ and conclude the proof.

\section{Proof of Theorem \ref{prop:boussinesq}} \label{sec:boussinesq}

We recall $\Phi^{(\gamma, \alpha, \rho, \sigma, \tau)}$ defined by RH problem \ref{RHP:model_Psi} and $M = (m^{(\gamma, \alpha)}_{ij}(\rho, \sigma, \tau))^3_{i, j = 1}$ defined in \eqref{eq:asy_expansion_Phi}. Throughout this section, we use $\Phi$ and $m_{ij}$ to mean $\Phi^{(\gamma, \alpha, \rho, \sigma, \tau)}$ and  $m_{ij}(\gamma, \alpha, \rho, \sigma, \tau)$ respectively, unless otherwise stated. We note that by identity \eqref{eq:Phi_det}, we have
\begin{equation} \label{eq:trace_zero}
  m_{11} + m_{22} + m_{33} = 0.
\end{equation}

From \eqref{eq:Phi_det}, we see that $ \Phi $ is invertible. Therefore we define
\begin{align} \label{eq:LaxPair}
L = {}& -\xi \frac{\partial}{\partial \xi} \Phi   \Phi^{-1}, &  A={}& -\frac{\partial}{\partial \tau} \Phi \Phi^{-1}, &   B= {}&-\frac{\partial}{\partial \sigma} \Phi  \Phi^{-1}, &  C= {}&-\frac{\partial}{\partial \rho} \Phi \Phi^{-1}.
\end{align}
We find that $A$, $B$ and $C$ are analytic in $\xi \in \compC \setminus \{ 0 \}$, and from the asymptotic behavior of $\Phi$ at $0$, we have that $A$, $B$ and $C$ are polynomials in $\xi$. By direct computation, we have the asymptotic expansion at $\infty$
\begin{equation}
  A =
  \begin{pmatrix}
    0 & 1 & 0 \\
    0 & 0 & 1 \\
    \xi & 0 & 0
  \end{pmatrix}
  + \left[ M,
    \begin{pmatrix}
      0 & 0 & 0 \\
      0 & 0 & 0 \\
      1 & 0 & 0
    \end{pmatrix}
  \right] + \xi^{-1} (A^{(-1)} - M_{\tau}) + \bigO(\xi^{-2}),
\end{equation}
where $M_{\tau}$ is the derivative of $M$ with respect to $\tau$, and 
\begin{equation}
  A^{(-1)} = \left[M,
    \begin{pmatrix}
      0 & 1 & 0 \\
      0 & 0 & 1 \\
      0& 0 & 0
    \end{pmatrix}
  \right] + \left[ \hat{M},
    \begin{pmatrix}
      0 & 0 & 0 \\
      0 & 0 & 0 \\
      1 & 0 & 1
    \end{pmatrix}
  \right] + \left[
    \begin{pmatrix}
      0 & 0 & 0 \\
      0 & 0 & 0 \\
      1 & 0 & 1
    \end{pmatrix}
    M, M \right].
\end{equation}
Since $A$ is a polynomial in $\xi$, we have
\begin{equation} \label{eq:M_tau=A^-1}
  A^{(-1)} = M_{\tau},
\end{equation}
and
\begin{equation}\label{eq:A}
  A =
  \begin{pmatrix}
    0 & 1 & 0 \\
    0 & 0 & 1 \\
    \xi & 0 & 0
  \end{pmatrix}
  +
  \begin{pmatrix}
    m_{13}& 0 & 0 \\
    m_{23} & 0 & 0 \\
    m_{33} - m_{11}& - m_{12} & -m_{13}
  \end{pmatrix}.
\end{equation}
Similarly, we have
\begin{equation} \label{eq:B}
  B=
  \begin{pmatrix}
    0 & 0 & 1 \\
    \xi & 0 & 0 \\
    0 & \xi & 0
  \end{pmatrix}
  +
  \begin{pmatrix}
    m_{12}& m_{13} & 0 \\
    m_{22}   -m_{11} &m_{23}   -m_{12} &   -m_{13} \\
    m_{32} - m_{21}& m_{33}-m_{22} & -m_{23}
  \end{pmatrix},
\end{equation}
and
\begin{equation} \label{eq:C_in_A}
  C = \xi A + A^{(-1)} + \bigO(\xi^{-1}) = \xi A + M_{\tau}.
\end{equation}
Since the upper-right $2 \times 2$ block of $A^{(-1)}$ is independent of $\hat{M}$ and given explicitly by
\begin{equation} \label{eq:M_x} 
  A^{(-1)} = 
  \begin{pmatrix}
    * & m_{11}-m_{22}-m_{12}m_{13}& m_{12}-m_{23}-m_{13}^2  \\
    *  & m_{21}-m_{32}-m_{12}m_{23} & m_{22}-m_{33}-m_{13}m_{23} \\
    * &  * &  * 
  \end{pmatrix},
\end{equation}
we have four differential identities from \eqref{eq:M_tau=A^-1}:
\begin{align}
  (m_{12})_{\tau} = {}& m_{11} - m_{33} - m_{13}(m_{12} + m_{23}), & (m_{13})_{\tau} = {}& m_{12}-m_{23}-m_{13}^2 = \frac{1}{3}u, \label{eq:m_13_x} \\
  (m_{22})_{\tau} = {}& m_{21}-m_{32}-m_{12}m_{23}, & (m_{23})_{\tau} = {}& m_{22}-m_{33}-m_{13}m_{23}, \label{eq:m_23_x}
\end{align}
and these identities yield \eqref{eq:u_v_expr}, where the functions $u$ and $v$ are defined in \eqref{eq:uv_u} and \eqref{eq:uv_v}.

We recall matrix-valued functions $P$ and $\Phihat=P\Phi$ defined in \eqref{eq:transform}. By the transformation \eqref{eq:transform}, we have
\begin{align} \label{eq:LaxPair_tilde_no_tilde}
  -\frac{\partial}{\partial \star} \Phihat = {}& (PAP^{-1} - P_{\star} P^{-1}) \Phihat, \quad \star = \tau, \sigma, \text{ and } \rho.
\end{align}

\paragraph{Proof of the second item of \eqref{eq:LaxTriple}}

By direct calculation, we have
\begin{align} \label{eq:PAPinv}
  PAP^{-1} = {}&
                 \begin{pmatrix}
                   0 & 1 & 0 \\
                   \frac{1}{3}u & 0 & 1 \\
                   \xi - \frac{1}{3}u_{\tau} +\frac{2}{3}v & -\frac{2}{3}u & 0
                 \end{pmatrix},
  &
    P_{\tau} P^{-1} = {}&
                          \begin{pmatrix}
                            0 & 0 & 0 \\
                            \frac{1}{3} u & 0 & 0 \\
                            \frac{1}{3}(-u_{\tau} - v) & \frac{1}{3} u & 0
                          \end{pmatrix},
\end{align}
and verify that $PAP^{-1} - P_{\tau} P^{-1} = \hat{A}$, where $\hat{A}$ is defined in \eqref{eq:Ahat}, and we prove the second item of \eqref{eq:LaxTriple}.

\paragraph{Proof the third item of \eqref{eq:LaxTriple}}

From  \eqref{eq:transform}, \eqref{eq:uv_u}, \eqref{eq:uv_v} and \eqref{eq:B}, we have

\begin{equation}\label{}
  \hat B=
    \begin{pmatrix}
      0 & 0 & 1 \\
      \xi & 0 & 0 \\0 & \xi & 0
    \end{pmatrix}
    +
    \begin{pmatrix}
      \frac{2}{3}u&0 & 0 \\
      b_{21}&  -\frac{1}{3}u & 0 \\
      b_{31}& b_{32}&  -\frac{1}{3}u 
    \end{pmatrix}.
  \end{equation}
  To determine the remaining entries, we use the zero-curvature equation
\begin{equation}\label{eq:Zero_Curvature_tau}
  \frac{\partial \hat{A}}{\partial \sigma} - \frac{\partial \hat{B}}{\partial \tau} = \hat{A} \hat{B} - \hat{B} \hat{A}.
\end{equation}
From the $(1,1)$, $(2,2)$ and $(2,1)$ entries of this equation, we have
\begin{equation}
  b_{21} - v = -\frac{2}{3}u_{\tau}, \quad
  b_{32}-b_{21}= \frac{1}{3}u_{\tau}, \quad
  b_{31}= -(b_{21})_{\tau},
\end{equation}
which implies the third item of \eqref{eq:LaxTriple}.

  \paragraph{Proof of the fourth item of \eqref{eq:LaxTriple}}

  We note that by \eqref{eq:C_in_A},
  \begin{equation}
    PCP^{-1} - P_{\rho} P^{-1} = \xi PAP^{-1} + P M_{\tau} P^{-1} - P_{\rho} P^{-1},
  \end{equation}
 where  $PAP^{-1}$ is obtained in \eqref{eq:PAPinv}.  By the lower triangular structure of $P$, we have that $P_{\rho} P^{-1}$ is strictly lower triangular with diagonal entries $0$. Then by direct calculation of $P M_{\tau} P^{-1}$, we have
  \begin{equation}
    P M_{\tau} P^{-1} - P_{\rho} P^{-1} =
    \begin{pmatrix}
      c_{11} & -\frac{1}{3}(v - u_x) & \frac{1}{3} u \\
      c_{21} & c_{22} & -\frac{1}{3} v \\
      c_{31} & c_{32} & c_{33}
    \end{pmatrix},
  \end{equation}
  where $c_{11}$, $c_{22}$ and $c_{33}$ are polynomials in $m_{ij}$ and $(m_{ij})_{\tau}$, and $c_{21}$, $c_{31}$ and $c_{32}$ are polynomials in $m_{ij}$, $(m_{ij})_{\tau}$ and $(m_{ij})_{\sigma}$. Similar to \eqref{eq:trace_zero}, we have
  \begin{equation} \label{eq:trac_C}
    c_{11} + c_{22} + c_{33} = 0.
  \end{equation}
  Next, we apply the zero-curvature equation
  \begin{equation}\label{eq:Zero_Curvature_t}
    \frac{\partial \hat A}{\partial \rho} - \frac{\partial \hat C}{\partial \tau} = \hat A \hat C-\hat C\hat A.
  \end{equation}
  From the $(1, 2)$ and $(2, 3)$ entries of the identity, we have
  \begin{align}
    c_{22} - c_{11} = {}& \frac{1}{3}(-u_{\tau \tau} + v_{\tau} - u^2), & c_{33}-c_{22} = {}&\frac{1}{3}v_{\tau}.
  \end{align}
  These identities, together with \eqref{eq:trac_C}, imply \eqref{eq:c11}. Next, the $(1, 1)$, $(2, 1)$ and $(2, 2)$ entries of the identity are
  \begin{align} \label{eq:u_t}
    c_{21} = {}& \frac{1}{3} uv - (c_{11})_{\tau}, & c_{31} = {}& -\frac{1}{3} v^2 - (c_{21})_{\tau}, & c_{32} = {}& c_{21} + \frac{1}{3} uv - (c_{22})_{\tau}.
  \end{align}
  These identities imply \eqref{eq:c21} -- \eqref{eq:c32}.

  \paragraph{Boussinesq equations}
  
  The $(3, 2)$ and $(3, 1)$ entries of the zero-curvature equation \eqref{eq:Zero_Curvature_tau} are non-trivial, and they imply \eqref{eq:1st_Bous_hier} and further \eqref{eq: Boussinesq equation}. Similarly, the $(3, 2)$ and $(3, 1)$ entries of the zero-curvature equation \eqref{eq:Zero_Curvature_t} imply \eqref{eq:2nd_Boussinesq_hierarchy:1} and \eqref{eq:2nd_Boussinesq_hierarchy:2}.

  \paragraph{Proof of the similarity reduction formulas \eqref{eq:Boussinesq_equation_coupled_equation:1} and \eqref{eq:Boussinesq_equation_coupled_equation:2}}

  We introduce the differential operator 
  \begin{equation}
    D=-\frac{1}{4} \tau \partial_{\tau} - \frac{1}{2}\sigma \partial_{\sigma} - \rho \partial_{\rho}.
  \end{equation}
  From the behavior of $\Psi$ at infinity, we have
  \begin{equation}\label{eq: dzPsi}
    \frac{3}{4}\xi\Phi_{\xi}(\xi) + D \Phi(\xi) = \frac{1}{4} \diag(0,1,2) \Phi(\xi).
  \end{equation}
  From the vanishing of the coefficient of $1/\xi$ in the large-$\xi$ expansion of the above equation, we have
  \begin{equation}\label{eq: DM}
    DM - \frac{3}{4} M + \frac{1}{4} [M, \diag(0,1,2)] = 0, \quad \text{or equivalently} \quad D m_{ij} = \frac{3 + i - j}{4} m_{ij}, \quad i, j = 1, 2, 3.
  \end{equation}
  Then we apply the operator $D$ to $u$ and $v$, and have
  \begin{align}
    4\rho u_{\rho} + 2\sigma u_{\sigma} + \tau u_{\tau} + 2u = {}& 0, & 4\rho v_{\rho} + 2\sigma v_{\sigma} + \tau v_{\tau} + 3v = {}& 0.
  \end{align}
  Substituting \eqref{eq:1st_Bous_hier}, \eqref{eq:2nd_Boussinesq_hierarchy:1} and \eqref{eq:2nd_Boussinesq_hierarchy:2} for $u_{\sigma}$, $v_{\sigma}$, $u_{\rho}$ and $v_{\rho}$ into the above equations, we obtain the similarity reduction of the Boussinesq hierarchy \eqref{eq:Boussinesq_equation_coupled_equation:1} and \eqref{eq:Boussinesq_equation_coupled_equation:2}.

  \paragraph{Proof of the first item of \eqref{eq:LaxTriple}}
  
  By \eqref{eq: DM},  we have
  \begin{equation}\label{eq:DP}
    DPP^{-1} = -\frac{1}{4} P \diag(0,1,2) P^{-1} + \frac{1}{4} \diag(0,1,2).
  \end{equation}
  By \eqref{eq:transform}, \eqref{eq: dzPsi} and \eqref{eq:DP}, we have 
  \begin{equation}
    \begin{split}
      -\frac{3}{4}\xi\Phihat_{\xi}\Phihat^{-1}= {}& P\left(-\frac{3}{4}\xi\Phi_{\xi}\Phi^{-1}\right)P^{-1} \\
      = {}& P\left(-\frac{1}{4} \diag(0,1,2) +D\Phi\Phi^{-1}\right)P^{-1} \\
      = {}& -\frac{1}{4} \diag(0,1,2)+DPP^{-1}+PD\Phi\Phi^{-1}P^{-1} \\
      = {}& -\frac{1}{4} \diag(0,1,2)+D\Phihat\Phihat^{-1}.
    \end{split}
  \end{equation}
  Substituting the partial derivatives of $\Phihat$ with respect to $\tau$, $\sigma$ and $\rho$ into the above equation, we obtain the first item of \eqref{eq:LaxTriple}.

  \paragraph{Chazy and Painlev\'{e} equations}

  Finally, we assume $\rho = 0$ and derive two more equations for \eqref{eq:Boussinesq_equation_Reduction} and relate them to Chazy's equations. Without loss of generality, we take $\sigma = 3/2$. We can compute $L$ in \eqref{eq:LaxPair} directly by the asymptotic expansion of $\Phi$ as $\xi \to \infty$ in \eqref{eq:expansion_at_infty_Phi}, like the computation for $A$ in \cite[Equation (43)]{Wang-Xu25}. Then we have $\hat{L} = PLP^{-1}$. We have that the $1/\xi$ coefficient of $L$ is 
  \begin{equation}
    -L_{-1}= -M + \left[ M,
        \begin{pmatrix}
          0 & -\frac{\tau}{3} & -1 \\
          0 & \frac{1}{3} & -\frac{\tau}{3} \\
          0& 0 & \frac{2}{3}
        \end{pmatrix}
      \right] + \left[ M,
        \begin{pmatrix}
          0 & 0& 0 \\
          1 & 0 & 0\\
          \frac{\tau}{3}& 1& 0
        \end{pmatrix}
      \right] M-\left[\hat{M},
        \begin{pmatrix}
          0 & 0& 0 \\
          1 & 0 & 0\\
          \frac{\tau}{3}& 1& 0
        \end{pmatrix}
      \right],
  \end{equation}
  whose derivation is similar to that of $A_{-2}$ in \cite[Equation (418)]{Wang-Xu25}. Since we know that $L$ and $\hat{L}$ are analytic at $\xi = 0$, this coefficient must be zero, which implies particularly, analogous to \cite[Equation (419)]{Wang-Xu25},
  \begin{equation} \label{eq:vanishing_L-1_13}
    - \left( L_{-1} \right)_{13} = m_{33} - m_{11} + (m_{12} + m_{23})m_{13} - \frac{1}{3} \left( \tau(m_{12} - m_{23}) - \tau m^2_{13} + m_{13} \right)=0.
  \end{equation}
  Below we denote $m_{13}$ by $f$ as in \eqref{eq:defn_f(tau)}. Then \eqref{eq:u_v_expr} implies that $u = 3f_{\tau}$. Identity \eqref{eq:vanishing_L-1_13}, together with \eqref{eq:m_13_x}, implies
  \begin{equation} \label{eq:(tau_f)_tau}
    \begin{split}
      m_{33}-m_{11}+(m_{12}+m_{23})f = {}& \frac{1}{3}\left(\tau (m_{12}-m_{23})-\tau f^2+f\right)\\
      = {}& \frac{1}{3}\left(\tau (f_{\tau} +f^2)-\tau f^2+f\right)\\
      = {}& \frac{1}{3}(\tau f)_{\tau} .
    \end{split}
  \end{equation}
  We also have, by \eqref{eq:m_13_x} and \eqref{eq:m_23_x},
  \begin{equation} \label{eq:ftautau}
    f_{\tau \tau} = m_{11} - 3 m_{12} m_{13} + 2 m_{13}^{3} + 3 m_{13} m_{23} - 2 m_{22} + m_{33}.
  \end{equation}
  As the linear combination of \eqref{eq:(tau_f)_tau} and \eqref{eq:ftautau}, we have
  \begin{equation}\label{eq:v_f} 
    v = \frac{1}{2}(\tau f)_{\tau} +\frac{3}{2}f_{\tau \tau}.
  \end{equation}
  
  From \eqref{eq:LaxTriple}, we have
 \begin{align}
     &-\xi\Phihat_{\xi}\Phihat^{-1} = \hat{L} \nonumber \\ 
     = {}& -\frac{1}{3} \diag(0,1,2)+\frac{\tau}{3}\hat A + \hat B \nonumber \\
     = {}& -\frac{1}{3} \diag(0,1,2)+\frac{\tau}{3}
           \begin{pmatrix}
             0 & 1 & 0 \\
             0 & 0 & 1 \\
             \xi +v& -u& 0
           \end{pmatrix}
           +
           \begin{pmatrix}
             0 & 0 & 1 \\
             \xi & 0 & 0 \\0 & \xi & 0
           \end{pmatrix}
           +
           \begin{pmatrix}
             \frac{2}{3}u& 0 & 0 \\
             v - \frac{2}{3}u_{\tau} & - \frac{1}{3}u & 0 \\
             -v_{\tau} + \frac{2}{3}u_{\tau \tau}& v - \frac{1}{3}u_{\tau} & - \frac{1}{3}u
           \end{pmatrix} \nonumber \\  
     = {}& \frac{1}{3}
           \begin{pmatrix}
             2u & \tau  & 3 \\
             3\xi - 2u_{\tau} +3v & -u-1& \tau  \\
             \tau \xi + 2u_{\tau \tau}-3v_{\tau} +\tau v& 3\xi - (\tau u)_{\tau} +3v& -u-2
           \end{pmatrix}.
 \end{align}
 From the behavior of $\Phi$ near the origin, we have that the eigenvalues of $-\hat{L}_0$ are $1/2-\frac{\alpha}{3}$, $\frac{\alpha}{6}$ and $1/2+\frac{\alpha}{6}$, where $\hat{L}_0=\hat{L}(0)$. (This observation was also used in \cite[Section 7]{Wang-Xu25}.) Therefore, by computing the characteristic polynomial of $\hat{L}_0$, we obtain 
\begin{equation}\label{eq:First_integral_1} 
  3( -2u_{\tau \tau}+3v_\tau -\tau v)-\tau (-3u_\tau -\tau u+6v)-3u^2 - 3u=\frac{1}{4}(1+3\alpha-3\alpha^2),
\end{equation}
and 
\begin{multline}\label{eq:First_integral_2} 
  (-2u_{\tau} - \frac{2}{3}\tau u+3v)( -u_{\tau} - \frac{2}{3}\tau u+3v+\frac{2}{3}\tau )-(-u-\frac{1}{3}\tau ^2-1)(2u_{\tau \tau}-3v_\tau +\frac{2}{3}u^2 + \frac{4}{3}u+\tau v)= \\
  \frac{1}{12}\alpha^3+\frac{1}{8}\alpha^2-\frac{3}{8}\alpha.
\end{multline} 
Substituting $u = 3f_{\tau}$ and \eqref{eq:v_f}  into \eqref{eq:First_integral_1}, we obtain a third-order differential equation 
\begin{equation} 
f_{\tau \tau \tau} + 6f_{\tau}^2 + \frac{1}{3}{\tau}^2 f_{\tau} +\tau f + \frac{1}{18}(1+3\alpha-3\alpha^2)=0.
\end{equation}
Let $y(x)=f(x)+\frac{x^3}{108}$, then $y$ satisfies the Chazy's  equation \eqref{eq:Third_order_ODE}. 
Substituting $u = 3f_{\tau}$ and \eqref{eq:v_f}  into \eqref{eq:First_integral_2} and using \eqref{eq:Third_order_ODE} to eliminate $f_{\tau \tau \tau }$ from this equation,   we obtain  a second-order differential equation 
\begin{equation}
\left(f_{\tau \tau} + \frac{2}{9}\tau \right)^2+4\left(f_{\tau} +\frac{1}{9} {\tau}^2 \right)^3 -(\tau f_{\tau} - f)^2-\frac{1}{3}(\alpha^2-\alpha+1)(f_{\tau} +\frac{1}{9} \tau^2)+\frac{1}{54}(\alpha+1)(2\alpha-1)(\alpha-2)=0.
\end{equation}
By the transformation $\tilde{y}(x)=\sqrt{2}f(\sqrt{2}x)+\frac{4}{27}x^3$, we obtain another Chazy's equation \eqref{eq:second_order_ODE} that is equivalent to a Painlev\'{e} IV equation. 

  \paragraph{Elementary solution for special parameters}

After  completing  the proof of Theorem \ref{prop:boussinesq}, we consider the special cases in Remark \ref{rmk:poly_solution}. If $\rho = 0$, $\sigma = 3/2$, $\gamma = 1$ and $\alpha = 0$, from the behavior of $\Phi$ near the origin given by \eqref{eq:limit_Phi_at_0_generic} with $E$ expressed by \eqref{eq:E}, we find that $\hat{L}_0 + \frac{1}{2} I_3$ is of rank one. Then the $2 \times 2$ upper-right minor of $\hat{L}_0 + \frac{1}{2} I_3$ has zero determinant, that is,
\begin{equation}
  \frac{\tau^2}{9} - (-\frac{u}{3} + \frac{1}{6}) = 0,
\end{equation}
which implies that $3f_{\tau} = u = -\tau^2/3 + 1/2$. Plugging this into \eqref{eq:Third_order_ODE} with $\alpha = 0$, we obtain that $f = -\tau^3/27 + \tau/6$. Similarly, if $\rho = 0$, $\sigma = 3/2$, $\gamma = -1$ and $\alpha = 1$, we have that $\hat{L}_0 + \frac{1}{6} I_3$ is of rank one. By considering the $2 \times 2$ upper-right minor of $\hat{L}_0 + \frac{1}{6} I_3$, we have  $3f_{\tau} = u = -\tau^2/3 - 1/2$ and $f = -\tau^3/27 - \tau/6$. Thus we obtain \eqref{eq:poly_f}.

\section{Proof of Corollary \ref{cor:MB} and Propositions \ref{pro:sol_Psi} and \ref{pro:sum_kernel_Psi}} \label{sec:remaining}

\subsection{Proof of Corollary \ref{cor:MB}} \label{subsec:proof_MB}

Setting $\gamma = 0$, the limit identities \eqref{eq:limit_K^mult} and \eqref{eq:limit_K^mult_crit} imply \eqref{eq:MB_limit_Pearcey} and \eqref{eq:MB_limit_tri} respectively. It therefore remains to show the equivalence between \eqref{eq:MB_limit_Pearcey} and \cite[Equation (21)]{Wang-Xu25}. 

For this purpose, we first derive an identity for $\Phi^{(\gamma, \alpha, \rho, \sigma, \tau)}(\xi)$.
Let $\star = (\gamma, \alpha, \rho, \sigma, \tau)$. From the second item in \eqref{eq:LaxPair} and \eqref{eq:A}, we have
\begin{equation} \label{eq:defn:derivative_Phi(y)Phi(x)}
  \frac{d}{d\tau} \Phi^{\star}(y)^{-1} \Phi^{\star}(x) = \Phi^{\star}(y)^{-1} (A(y) - A(x)) \Phi^{\star}(x) = (y - x) \Phi^{\star}(y)^{-1}
  \begin{pmatrix}
    0 \\
    0 \\
    1 
  \end{pmatrix} 
  \begin{pmatrix}
    1 & 0 & 0
  \end{pmatrix}
  \Phi^{\star}(x).
\end{equation}
Next, let $\star$ be $(\gamma, \alpha, 0, \frac{3}{4},\tau)$ in \eqref{eq:limit_K^mult} or $(\gamma, \alpha, -\frac{1}{4}, \frac{3}{2}\sigma, \tau)$ in \eqref{eq:limit_K^mult_crit}. Since $K^{\star}(\xi, \eta)$ can be expressed by the $(0, 0)$, $(0, 1)$, $(1, 0)$ and $(1, 1)$ entries of $\Phi^{\star}(y)^{-1}\Phi^{\star}(x)$, we have, by substituting entry-wise identities of \eqref{eq:defn:derivative_Phi(y)Phi(x)} into \eqref{eq:limit_K^mult},
  \begin{equation}
    \begin{split}
      \frac{d}{d\tau} K^{\star}(\xi, \eta) = {}& \frac{1}{2\pi i} \left( -\Phi^{\star}_{00}(\xi) (\Phi^{\star}(\eta))^{-1}_{02} + e^{-\frac{2\alpha}{3} \pi i} \Phi^{\star}_{00}(\xi) (\Phi^{\star}(\eta))^{-1}_{12}(\eta) \right. \\
                                                                & \phantom{\frac{1}{2\pi i} } \left.
                                                                  \vphantom{e^{-\frac{2\alpha}{3} \pi i}}
                                                                  - e^{\frac{2\alpha}{3} \pi i} \Phi^{\star}_{01}(\xi) (\Phi^{\star}(\eta))^{-1}_{02} + \Phi^{\star}_{01}(\xi) (\Phi^{\star}(\eta))^{-1}_{12} \right)_+ \frac{\xi^{\frac{\alpha}{3} - \frac{1}{2}}}{\eta^{\frac{\alpha}{3} - \frac{1}{2}}} \\
      = {}& \frac{-1}{2\pi} \hat{\phi}^{\star}(\xi) \hat{\psi}^{\star}(\eta) \frac{\xi^{\frac{\alpha}{3} - \frac{1}{2}}}{\eta^{\frac{\alpha}{3} - \frac{1}{2}}},
    \end{split}
  \end{equation}
  where
  \begin{align}
    \hat{\phi}^{\star}(\xi) = {}& \left( e^{\frac{\alpha}{3} \pi i} \Phi^{\star}_{01}(\xi) + e^{-\frac{\alpha}{3} \pi i} \Phi^{\star}_{00}(\xi) \right)_+, \\
    \hat{\psi}^{\star}(\eta) = {}& \left( -i e^{\frac{\alpha}{3} \pi i} (\Phi^{\star}(\eta))^{-1}_{02} + i e^{-\frac{\alpha}{3} \pi i} (\Phi^{\star}(\eta))^{-1}_{12} \right)_+.
  \end{align}
  Now we take $\star = (0, \alpha, 0, \frac{3}{4}, \tau)$. We note that the correlation kernel $K^{(\tau)}$ defined in \cite[Theorem 1.6 and Equations (35) and (36)]{Wang-Xu25} has the expression
  \begin{equation} \label{eq:K^0alphatau_integral}
    K^{(\tau)}(\xi, \eta) = 2\xi \cdot \frac{1}{2\pi} \int^{\infty}_{\tau} \hat{\phi}^{(0, \alpha, 0, \frac{3}{4}, \tilde{\tau})}(\xi^2) \hat{\psi}^{(0, \alpha, 0, \frac{3}{4}, \tilde{\tau})}(\eta^2) \frac{\xi^{\frac{2\alpha}{3} - 1}}{\eta^{\frac{2\alpha}{3} - 1}} d\tilde{\tau}.
  \end{equation}
Consequently, we obtain the relation 
  \begin{equation} \label{eq:equivalence_of_MB_kernels}
    2\xi K^{(0, \alpha, 0, \frac{3}{4}, \tau)}(\xi^2, \eta^2) = K^{(\tau)}(\xi, \eta) + C,
  \end{equation}
  where $C$ is independent of $\tau$. Similar to \cite[Equation (486)]{Wang-Xu25}, we can derive estimates of $\phi^{(0, \alpha, 0, \frac{3}{4}, \tau)}_j(\xi)$ and $\psi^{(0, \alpha, 0, \frac{3}{4}, \tau)}_j(\eta)$ as $\tau \to \infty$, which shows that $C = 0$ in \eqref{eq:equivalence_of_MB_kernels}. This completes the proof of Corollary \ref{cor:MB}.

  \begin{rmk}
    In the same way, we can prove that $K^{(\gamma, \alpha, \rho, \sigma, \tau)}(\xi, \eta)$ also has an integral representation analogous to \eqref{eq:K^0alphatau_integral} if the parameters satisfy the conditions in Theorem \ref{thm:RHP_unique_solvability}. The proof relies on the estimates of $\phi^{(\gamma, \alpha, \rho, \sigma, \tau)}_j(\xi)$ and $\psi^{(\gamma, \alpha, \rho, \sigma, \tau)}_j(\eta)$ as $\tau \to \infty$ and other parameters fixed. We will give the details in a future publication.
  \end{rmk}

\subsection{Proof of the Propositions \ref{pro:sol_Psi} and \ref{pro:sum_kernel_Psi}}

We first show that there exists a gauge transformation relating $\Phi^{(\gamma, \alpha, \rho, \sigma, \tau)}$ for different choices of parameters:
\begin{lemma} \label{lem:Backlund}
  Let $\Phi^{(\gamma, \alpha, \rho, \sigma, \tau)}(\xi)$ be the matrix-valued function defined by RH problem \ref{RHP:model}, where $\gamma, \alpha, \rho, \sigma, \tau$ satisfy the conditions in Theorem \ref{thm:RHP_unique_solvability}. Then
  \begin{multline}\label{eq:Backlund_transform}
    \xi^{\frac{1}{3}} 
    \begin{pmatrix}
      0& 0& 1\\
      1& 0 & 0 \\
      0 & 1 & 0
    \end{pmatrix}
    \Phi^{(-\gamma, \alpha+1,\rho, \sigma, \tau)}(\xi)  \diag(\omega^{\pm 1}, \omega^{\mp 1},1) = \\
   H^{(\gamma, \alpha, \rho, \sigma, \tau, -)}(\xi)^T H^{(-\gamma, \alpha + 1, \rho, \sigma, \tau, +)}(\xi) \Phi^{(\gamma, \alpha, \rho, \sigma, \tau)}(\xi)
  \end{multline}
  for $\xi \in \compC_{\pm}$, where $H^{(\gamma, \alpha, \rho, \sigma, \tau, +)}(\xi)$ is defined in \eqref{eq:H}, and 
   \begin{equation} 
     H^{(\gamma, \alpha, \rho, \sigma, \tau, -)}(\xi)=
     \begin{pmatrix}
       \xi^{\frac{1}{2}} & 0 & 0\\
       -m^{(\gamma, \alpha)}_{12}(\rho, \sigma, \tau)  & 1 & 0 \\
       -m^{(\gamma, \alpha)}_{13}(\rho, \sigma, \tau)  & 0 & 1
     \end{pmatrix}.
  \end{equation}
\end{lemma}
\begin{proof}
  We consider the difference between the two sides of \eqref{eq:Backlund_transform}. By direct calculation, we verify that it satisfies the same jump condition of $ \Phi^{(\gamma, \alpha, \rho, \sigma, \tau)}(\xi)$ that is specified in \eqref{eq:RHP:model_1} and the same boundary condition of $\Phi^{(\gamma, \alpha, \rho, \sigma, \tau)}(\xi)$ at $0$ that is specified in \eqref{eq:expansion_at_infty_Phi}. We can also show that as $\xi \to \infty$, the difference has the asymptotic behaviour $\bigO(\xi^{-1}) \Upsilon(\xi) \Omega_{\pm} e^{-\Theta(\xi)}$ that is specified in \eqref{eq:limit_Phi_at_0_generic}. Hence, the difference satisfies the homogeneous version of RH problem \ref{RHP:model}. By Lemma \ref{lem:vanishing}, we have that the difference vanishes, and then prove the lemma.
\end{proof}

\begin{proof}[Proof of Proposition \ref{pro:sol_Psi}]
  First we prove that $\Psi^{(\gamma, \alpha, \rho, \sigma, \tau)}(\xi)$ defined by \eqref{eq: Psi_def} is a solution of RH problem \ref{RHP:model_Psi}. We note that by \eqref{eq:Backlund_transform}, \eqref{eq: Psi_def} can also be written as
  \begin{multline} \label{eq:Psi_alt}
    \Psi^{(\gamma, \alpha, \rho, \sigma, \tau)}(\xi)=\xi^{-\frac{1}{3}} (H^{(\gamma, \alpha, \rho, \sigma, \tau, -)}(\xi^2)^{-1})^T \\
    \times
    \begin{pmatrix}
      0 & 0 & 1 \\
      1 & 0 & 0 \\
      0 & 1 & 0
    \end{pmatrix}
    \times
    \begin{cases}
      \Phi^{(-\gamma, \alpha + 1, \rho, \sigma, \tau)}(\xi^2) \diag(\omega, \omega^2, 1), & \arg \xi \in (0, \frac{\pi}{2}), \\
      \Phi^{(-\gamma, \alpha + 1, \rho, \sigma, \tau)}(\xi^2) \diag(\omega^2, \omega, 1), & \arg \xi \in (-\frac{\pi}{2}, 0), \\
      \Phi^{(-\gamma, \alpha + 1, \rho, \sigma, \tau)}(\xi^2)(1\oplus\sigma_1)  \diag(\omega, \omega^2, 1), & \arg \xi \in(\frac{\pi}{2}, \pi), \\
      \Phi^{(-\gamma, \alpha + 1, \rho, \sigma, \tau)}(\xi^2)(1\oplus\sigma_1)  \diag(\omega^2, \omega, 1), & \arg \xi \in(-\pi, -\frac{\pi}{2}).
    \end{cases}
  \end{multline}
  It is straightforward to see that the function $\Psi^{(\gamma, \alpha, \rho, \sigma, \tau)}$ defined by \eqref{eq: Psi_def} satisfies the jump condition \eqref{eq:RHP:model_Psi} and boundary condition \eqref{eq:RHP_Psi_infty} at infinity. Now we verify the boundary condition at the origin, and without loss of generality assume that $\alpha \notin \mathbb{Z}$. We let $(\psi_0(\xi), \psi_1(\xi), \psi_2(\xi))$ be one of the three rows of $\Psi^{(\gamma, \alpha, \rho, \sigma, \tau)}(\xi)$. In the sector $ \arg \xi \in (\pi/8, 3\pi/8) $, using \eqref{eq:limit_Phi_at_0_generic} with the expression of $E$ given in \eqref{eq:generic_E}, we have that from \eqref{eq: Psi_def}, 
  \begin{equation} \label{eq:psi_one_way}
    \begin{split}
      (\psi_0(\xi), \psi_1(\xi), \psi_2(\xi)) = {}& (a_1(\xi^2) + a_2(\xi^2) \xi) \xi^{-\frac{2\alpha}{3}} (1, 0, 0) \\
                                                  & + (b_1(\xi^2) + b_2(\xi^2) \xi) \xi^{\frac{\alpha}{3}} (\frac{e^{\frac{2\alpha}{3} \pi i}}{1 - e^{\alpha \pi i}} (1 - \gamma), 1 - \gamma, -e^{\frac{\alpha}{3} \pi i}) \\
                                                  & + (c_1(\xi^2) + c_2(\xi^2) \xi) \xi^{-1 + \frac{\alpha}{3}} (\frac{e^{\frac{2\alpha}{3} \pi i}}{1 + e^{\alpha \pi i}} (1 + \gamma), 1 + \gamma, e^{\frac{\alpha}{3} \pi i}),
    \end{split}
  \end{equation}
  and from \eqref{eq:Psi_alt},
  \begin{equation} \label{eq:psi_another_way}
    \begin{split}
      (\psi_0(\xi), \psi_1(\xi), \psi_2(\xi)) = {}& (\tilde{a}_1(\xi^2) + \tilde{a}_2(\xi^2) \xi^{-1}) \xi^{-\frac{2\alpha}{3}} (1, 0, 0) \\
                                                  & + (\tilde{b}_1(\xi^2) + \tilde{b}_2(\xi^2) \xi^{-1}) \xi^{\frac{\alpha}{3}} (\frac{e^{\frac{2\alpha}{3} \pi i}}{1 - e^{\alpha \pi i}} (1 - \gamma), 1 - \gamma, -e^{\frac{\alpha}{3} \pi i}) \\
                                                  & + (\tilde{c}_1(\xi^2) + \tilde{c}_2(\xi^2) \xi^{-1}) \xi^{1 + \frac{\alpha}{3}} (\frac{e^{\frac{2\alpha}{3} \pi i}}{1 + e^{\alpha \pi i}} (1 + \gamma), 1 + \gamma, e^{\frac{\alpha}{3} \pi i}),
    \end{split}
  \end{equation}
  such that the functions $a_1, a_2, b_1, b_2, c_1, c_2, \tilde{a}_1, \tilde{a}_2, \tilde{b}_1, \tilde{b}_2, \tilde{c}_1, \tilde{c}_2$ are all analytic functions at $0$. By comparing \eqref{eq:psi_one_way} and \eqref{eq:psi_another_way}, we have that
  \begin{align}
    A(\xi) = {}& a_1(\xi^2) + a_2(\xi^2) \xi = \tilde{a}_1(\xi^2) + \tilde{a}_2(\xi^2) \xi^{-1}, \\
    B(\xi) = {}& b_1(\xi^2) + b_2(\xi^2) \xi = \tilde{b}_1(\xi^2) + \tilde{b}_2(\xi^2) \xi^{-1}, \\
    C(\xi) = {}& (c_1(\xi^2) + c_2(\xi^2) \xi) \xi^{-1} = (\tilde{c}_1(\xi^2) + \tilde{c}_2(\xi^2) \xi^{-1}) \xi
  \end{align}
  are all analytic functions at $\xi = 0$, and then
  \begin{multline} \label{eq:expr_psi_row}
      (\psi_0(\xi), \psi_1(\xi), \psi_2(\xi)) =  A(\xi) \xi^{-\frac{2\alpha}{3}} (1, 0, 0) \\
       + B(\xi) \xi^{\frac{\alpha}{3}} (\frac{e^{\frac{2\alpha}{3} \pi i}}{1 - e^{\alpha \pi i}} (1 - \gamma), 1 - \gamma, -e^{\frac{\alpha}{3} \pi i}) 
       + C(\xi) \xi^{\frac{\alpha}{3}} (\frac{e^{\frac{2\alpha}{3} \pi i}}{1 + e^{\alpha \pi i}} (1 + \gamma), 1 + \gamma, e^{\frac{\alpha}{3} \pi i}).     
  \end{multline}
  Since all three rows of $\Psi^{(\gamma, \alpha, \rho, \sigma, \tau)}$ can be expressed like \eqref{eq:expr_psi_row}, we verify that $\Psi^{(\gamma, \alpha, \rho, \sigma, \tau)}(\xi)$ defined by \eqref{eq: Psi_def} satisfies the boundary condition \eqref{eq:limit_Psi_at_0_generic}.
 
 Next, we prove the uniqueness. Suppose $\Psihat$ is another solution of this RH problem. By \eqref{eq:Phi_det}, $\Psi$ is invertible, so the product $\Psihat\Psi^{-1}$ is well defined. Since both $\Psihat$ and $\Psi$ satisfy the same jump condition, the function $\Psihat\Psi^{-1}$ is meromorphic in the complex plane with only possible singularities at the origin and infinity. The boundary condition at the origin implies that $\Psihat\Psi^{-1}$ is bounded near the origin; hence the singularity at $\xi=0$ is removable. From the behavior at infinity, we have $\Psihat(\xi)\Psi^{-1}(\xi)\to I$ as $\xi\to \infty$. By Liouville's theorem, $\Psihat\Psi^{-1}=I$, which proves the uniqueness.
\end{proof}

\begin{proof}[Proof of Proposition \ref{pro:sum_kernel_Psi}]
  We prove the equation \eqref{eq:sum_kernel_Psi} for $\xi, \eta > 0$, and the general case is obtained by analytic continuation.

  From \eqref{eq: Psi_def}, we can rewrite \eqref{eq:limit_kernel_mult} as
\begin{multline} \label{eq:limiting_kernel_BT}
  K^{(\gamma, \alpha, \rho, \sigma, \tau)}(\xi^2, \eta^2) = \frac{e^{\frac{2}{3}\alpha \pi i}}{2\pi i (\xi^2 - \eta^2)}
\begin{pmatrix}
      1 & 0 & 0
    \end{pmatrix}
    \tilde{\Psi}^{(\gamma, \alpha, \rho, \sigma, \tau)}(\eta)^{-1} \\
    H^{(-\gamma, \alpha + 1, \rho, \sigma, \tau, +)}(\eta^2) H^{(-\gamma, \alpha + 1, \rho, \sigma, \tau, +)}(\xi^2)^{-1} \tilde{\Psi}^{(\gamma, \alpha, \rho, \sigma, \tau)}(\xi)
    \begin{pmatrix}
      0 \\ 1 \\ 0
    \end{pmatrix},
  \end{multline}
  and from \eqref{eq:Psi_alt}
  \begin{multline} \label{eq:limiting_kernel_BT_alt}
    K^{(-\gamma, \alpha + 1, \rho, \sigma, \tau)}(\xi^2, \eta^2) = \frac{e^{\frac{2}{3}\alpha \pi i}}{2\pi i (\xi^2 - \eta^2)}
\begin{pmatrix}
      1 & 0 & 0
    \end{pmatrix}
    \tilde{\Psi}^{(\gamma, \alpha, \rho, \sigma, \tau)}(\eta)^{-1} \\
    (H^{(\gamma, \alpha, \rho, \sigma, \tau, -)}(\eta^2)^{-1})^T H^{(\gamma, \alpha, \rho, \sigma, \tau, -)}(\xi^2)^T \tilde{\Psi}^{(\gamma, \alpha, \rho, \sigma, \tau)}(\xi)
    \begin{pmatrix}
      0 \\ 1 \\ 0
    \end{pmatrix}.
  \end{multline}
  From the identity
  \begin{multline}
    \xi H^{(-\gamma, \alpha + 1, \rho, \sigma, \tau, +)}(\eta^2) H^{(-\gamma, \alpha + 1, \rho, \sigma, \tau, +)}(\xi^2)^{-1} + \eta (H^{(\gamma, \alpha, \rho, \sigma, \tau, -)}(\eta^2)^{-1})^T H^{(\gamma, \alpha, \rho, \sigma, \tau, -)}(\xi^2)^T = \\
    (\xi + \eta) I_3,
  \end{multline}
it follows that \eqref{eq:limiting_kernel_BT} and \eqref{eq:limiting_kernel_BT_alt} together imply \eqref{eq:sum_kernel_Psi}.
\end{proof}

\appendix

\section{The Airy parametrix} \label{app:Airy}

Let $y_0$, $y_1$ and $y_2$ be the functions defined by
\begin{equation}
  y_0(\zeta) = \sqrt{2\pi}e^{-\frac{\pi i}{4}} \Ai(\zeta), \quad y_1(\zeta) = \sqrt{2\pi}e^{-\frac{\pi i}{4}} \omega\Ai(\omega \zeta), \quad y_2(\zeta) = \sqrt{2\pi}e^{-\frac{\pi i}{4}} \omega^2\Ai(\omega^2 \zeta),
\end{equation}
where $\Ai$ is the usual Airy function (cf. \cite[Chapter 9]{Boisvert-Clark-Lozier-Olver10}) and $\omega=e^{2\pi i/3}$. We then define a $2\times 2$ matrix-valued function $\Psi^{(\Ai)}$ by
\begin{equation} \label{eq:defn_Psi_Ai}
  \Psi^{(\Ai)}(\zeta)
  = \left\{
    \begin{array}{ll}
      \begin{pmatrix}
        y_0(\zeta) &  -y_2(\zeta) \\
        y_0'(\zeta) & -y_2'(\zeta)
      \end{pmatrix}, & \hbox{$\arg \zeta \in (0,\frac{2\pi}{3})$,} \\
      \begin{pmatrix}
        -y_1(\zeta) &  -y_2(\zeta) \\
        -y_1'(\zeta) & -y_2'(\zeta)
      \end{pmatrix}, & \hbox{$\arg \zeta \in (\frac{2\pi}{3},\pi)$,} \\
      \begin{pmatrix}
        -y_2(\zeta) &  y_1(\zeta) \\
        -y_2'(\zeta) & y_1'(\zeta)
      \end{pmatrix}, & \hbox{$\arg \zeta \in (-\pi,-\frac{2\pi}{3})$,} \\
      \begin{pmatrix}
        y_0(\zeta) &  y_1(\zeta) \\
        y_0'(\zeta) & y_1'(\zeta)
      \end{pmatrix}, & \hbox{$\arg \zeta \in  (-\frac{2\pi}{3},0)$.}
    \end{array}
  \right.
\end{equation}
It is well-known that $\det (\Psi^{(\Ai)}(\zeta))=1$ and $\Psi^{(\Ai)}(\zeta)$ is the unique solution of the following $2 \times 2$ RH problem; cf. \cite[Section 7.6]{Deift99}.
\begin{RHP} \label{rhp:Ai}
\begin{enumerate}
\item
  $ \Psi^{(\Ai)}(\zeta)$ is analytic in $\mathbb{C} \setminus \Gamma_{\Ai}$, where the contour $\Gamma_{\Ai}$ is defined by
    \begin{equation} \label{def:AiryContour}
      \Gamma_{\Ai}:=e^{-\frac{2\pi i}{3}}[0,+\infty) \cup \mathbb{R} \cup e^{\frac{2\pi i}{3}}[0,+\infty)
    \end{equation}
 with all rays oriented from left to right.
\item
  For $\zeta \in \Gamma_{\Ai}$, we have
  \begin{equation} \label{eq:Ai_jump}
     \Psi^{(\Ai)}_+(\zeta) =  \Psi^{(\Ai)}_-(\zeta)
    \begin{cases}
      \begin{pmatrix}
        1 & 1 \\
        0 & 1
      \end{pmatrix},
      & \arg \zeta =0, \\
      \begin{pmatrix}
        1 & 0 \\
        1 & 1
      \end{pmatrix},
      & \arg \zeta = \pm \frac{2\pi }{3}, \\
      \begin{pmatrix}
        0 & 1  \\
        -1 & 0
      \end{pmatrix},
      & \arg \zeta = \pi.
    \end{cases}
  \end{equation}
\item
  As $\zeta \to \infty$, we have
    \begin{equation} \label{eq:defn_Psi^Ai_infty}
      \Psi^{(\Ai)}(\zeta) = \Psi^{(\Ai)}_{\infty}(\zeta) (I+\bigO(\zeta^{-\frac32}))e^{-\frac23 \zeta^{\frac32}\sigma_3}, \quad \Psi^{(\Ai)}_{\infty}(\zeta) = \zeta^{-\frac{1}{4} \sigma_3} \frac{1}{\sqrt{2}}
      \begin{pmatrix}
        1 & 1 \\
        -1 & 1
      \end{pmatrix}
      e^{-\frac{\pi i}{4} \sigma_3}.
    \end{equation}
\end{enumerate}
\end{RHP}


\end{document}